\documentclass[12pt, reqno]{amsart}

\usepackage{a4wide}
\usepackage{bbm}
\usepackage{amssymb,amsthm,amsmath}
\usepackage{graphicx, enumerate}
\usepackage{psfrag}
\usepackage{natbib}
\usepackage[unicode, a4paper]{hyperref}
\usepackage{hypernat}
\usepackage{algorithm, algorithmic}
\usepackage{url}
\usepackage{color}
\usepackage{colortbl}
\usepackage[onehalfspacing]{setspace} 
\usepackage[foot]{amsaddr}
\usepackage{multibib}
 \usepackage{subfig}

\newcites{supp}{References}
\newcommand{\Pcal}{\mathcal{P}}
\newcommand{\Pfrak}{\mathfrak{P}}
\newcommand{\R}{\mathbb{R}}
\newcommand{\N}{\mathbb{N}}

\newcommand{\Ic}{\mathcal{I}(c_n)}

\renewcommand{\S}{\mathcal{S}}

\newcommand{\diff}{\textnormal{\,d}}

\DeclareMathOperator*{\argmin}{\textnormal{argmin}}
\DeclareMathOperator*{\argmax}{\textnormal{argmax}}

\newcommand{\E}[1]{\mathbf{E}\left(#1\right)}

\newcommand{\Prob}{\mathbf{P}}
\newcommand{\Var}[1]{\mathbf{Var}\left(#1\right)}

\newcommand{\eps}{\varepsilon}

\newcommand{\set}[1]{\left\{ #1 \right\}}
\newcommand{\abs}[1]{\left| #1 \right|}
\newcommand{\norm}[1]{\left\| #1 \right\|}

\newcommand{\smallo}{\textnormal{o}}
\newcommand{\bigo}{\mathcal{O}}
\newcommand{\ra}{\rightarrow}
\newcommand{\eye}{\mathbf{1}}

\newcommand{\vt}{\vartheta}

\definecolor{tablebg}{rgb}{0.85,0.85,0.85}

\newcommand{\ha}[1]{{#1}}
\newcommand{\kl}[1]{{#1}}
\newcommand{\am}[1]{{#1}}
                     
\theoremstyle{plain}

\newtheorem{thm}{Theorem}[section]
\newtheorem{prop}[thm]{Proposition}

\theoremstyle{definition}

\newtheorem{lem}[thm]{Lemma}

\newtheorem{cor}[thm]{Corollary}
\newtheorem{example}[thm]{Example}

\newtheorem*{example*}{Example}
\newtheorem*{dfn*}{Definition}
\newtheorem*{alg*}{Algorithm}

\theoremstyle{remark}

\begin{document}  

\author{Klaus Frick $^1$ Axel Munk $^{1,2}$ Hannes Sieling $^1$}  

\address{$^1$Institute for Mathematical Stochastics\\
University of G{\"o}ttingen\\      
Goldschmidtstra{\ss}e 7, 37077 G{\"o}ttingen}
   
\address{$^2$Max Planck Institute for Biophysical Chemistry \\
Am Fa{\ss}berg 11, 37077 G{\"o}ttingen}   
   
\email{\{frick, munk, hsielin\}@math.uni-goettingen.de}



\date{\today}

\title[Multiscale Change-Point Inference]{Multiscale Change-Point Inference}

\begin{abstract}
We introduce a new estimator SMUCE (simultaneous multiscale change-point estimator) for the change-point problem in exponential family regression.
An unknown step function is estimated by minimizing the number of change-points over the acceptance region of a
multiscale test at a level $\alpha$.

The probability of overestimating the true number of change-points $K$ is controlled by
the asymptotic null distribution of the multiscale test statistic.
Further, we derive exponential bounds 
for the probability of underestimating $K$.
By balancing these quantities, $\alpha$ will be chosen such that the probability of
correctly estimating $K$ is maximized.
All results are even non-asymptotic for the normal case.

Based on the aforementioned bounds, we construct (asymptotically) honest confidence sets for the unknown
step function and its change-points. At the same time, we obtain exponential bounds for estimating
the change-point locations which for example yield the minimax rate $\bigo(n^{-1})$ up 
to a log term. 
Finally, SMUCE achieves the optimal detection rate of vanishing signals as $n\ra\infty$,
\am{even for an unbounded number of change-points.}

We illustrate how dynamic programming techniques can be employed for efficient
computation of estimators and confidence regions. The performance of the proposed multiscale approach is illustrated
by simulations and in two cutting-edge applications from genetic engineering and
photoemission spectroscopy.

\end{abstract}

\maketitle
 
\section{Introduction}

Assume that we observe independent random variables $Y = (Y_1,\ldots,Y_n)$ through the exponential family regression model
\begin{equation}\label{intro:model}
Y_i \sim F_{\vt(i\slash n)},\quad\text{ for }i=1,\ldots,n,
\end{equation}
where $\set{F_\theta}_{\theta\in\Theta}$ is a one dimensional exponential family
with densities $f_\theta$ and $\vt:[0,1)\ra \Theta\subseteq\R$ a
right-continuous step function with an unknown number $K$ of change-points. The two upper panels in Figure
\ref{intro:example} depict such a step function with $K=8$ change-points and
corresponding data $Y$ for the Gaussian family
$F_\theta = \mathcal{N}(\theta, \sigma^2)$ with fixed variance $\sigma^2$.

The \emph{change-point problem} consists in estimating
\begin{enumerate}[(i)]
  \item the number of change-points of $\vt$,
  \item the change-point locations and the function values (intensities) of $\vt$.
\end{enumerate}
Additionally, we address the more involved issue of constructing
\begin{enumerate}[(i)]
 \setcounter{enumi}{2}
 \item confidence bands for the function $\vt$ and simultaneous
 confidence intervals for its change-point locations.
\end{enumerate}

\subsection{Multiscale statistics and estimation} \label{subsec:intro:estimator}

The goals (i) - (iii) will be achieved based on a new estimation and inference method for
the change-point problem in exponential families: the {S}imultaneous {MU}ltiscale {C}hange-point {E}stimator ({SMUCE}). Let
$\mathcal{S}$ denote the space of all \am{right-continuous} step functions \am{with an arbitrary but finite number of jumps} on the unit interval
$[0,1)$ with values in $\Theta$. For $\vt\in\mathcal{S}$ we denote by 
$J(\vt)$ the ordered vector of change-points and by $\#J(\vt)$ its
length, i.e. the number of change-points. In a first step, \am{SMUCE requires} to solve
the (nonconvex) optimization problem
\begin{equation}\label{intro:optprob}
\inf_{\vt \in\mathcal{S}} \#J(\vt) \quad\text{ s.t.
}\quad T_n(Y,\vt)\leq q,
\end{equation}
where $q$ is a threshold to be specified later. $T_n(Y,\vt)$ is a certain
\emph{multiscale statistic} for a candidate function $\vt\in\S$.
Optimization problems of the type \eqref{intro:optprob} have been recently considered in \citep{Hoe08} for Gaussian change-point regression (see 
also \citep{BoyKemLieMunWit09} for a related approach) and for volatility estimation in \citep{DavHoeKra12}.
$T_n$ in \eqref{intro:optprob} evaluates
the maximum over the local likelihood ratio statistics on all discrete intervals
$[i\slash n, j\slash n]$ such that $\vt$ is constant on these with value $\theta=\theta_{i,j}$, i.e.
\begin{equation}\label{intro:mrstat}
T_n(Y,\vt) = \max_{\substack{1\leq i<j\leq n \\ \vt(t) = \theta \text{ for } t
\in [i\slash n, j\slash n]}} \left(  \sqrt{2 T_i^j(Y,\theta)} -
\sqrt{2\log\frac{en}{j-i+1}} \right),
\end{equation}
where $e=\exp(1)$ and $\log$ denotes the natural logarithm.  The local
likelihood ratio statistic $T_i^j$ for testing $H_0:\theta=\theta_0$ against
$H_1:\theta\neq\theta_0$ on the interval $[i/n,j/n]$ is defined as
\begin{equation}\label{intro:likeratstat}
T_i^j(Y,\theta_0) = \log\left(\frac{\ha{\sup}_{\theta\in\Theta}\prod_{l=i}^j
f_\theta(Y_l)}{\prod_{l=i}^j f_{\theta_0}(Y_l)} \right).
\end{equation}
It measures how well the data can be described \emph{locally} by a constant value $\theta_0$ on the interval
$[i\slash n, j\slash n]$.
We stress that the multiscale statistic $T_n$ does not act on all intervals $[i/n,j/n]\subseteq[0,1]$ but only on those which the candidate function $\vt$
is constant on, see also \citep{DavHoeKra12,Hoe08,OlsVenLucWig04}.
Thus the system of intervals
appearing in \eqref{intro:mrstat} makes up the specific multiscale nature of
$T_n$. The $\log$-expression in \eqref{intro:mrstat} can be seen as a scale
calibrating term that puts different scales on equal footing.
As argued in \citep{DueSpok01} and \citep{ChaWal11} this improves the power of the multiscale test over the majority of scales. 
Roughly speaking, from a multiscale point of view, scale-calibration becomes advantageous, since there are many more small intervals than large ones.

SMUCE integrates the multiscale test on the r.h.s. in \eqref{intro:mrstat} into
two simultaneous estimation steps: Model selection (estimation of $K$) and estimation of $\vt$ given $K$.
The minimal value of $\#J$ in \eqref{intro:optprob} gives the estimated number
of change-points, denoted by $\hat K(q)$. To obtain the final estimator for
$\vt$ first consider the set of all solutions of \eqref{intro:optprob} given by
\begin{equation}\label{intro:confset}
\mathcal{C}(q) = \set{\vt\in\S ~:~ \#J(\vt) = \hat K(q) \text{ and }
T_n(Y,\vt) \leq q},
\end{equation}
which constitutes a confidence set for the true regression function
$\vt$ as we will discuss later on. Then, the SMUCE $\hat \vt(q)$ is defined to be the \textit{constrained maximum likelihood
estimator} within this confidence set $\mathcal{C}(q)$, i.e. 
\begin{equation}\label{intro:smre}
\hat \vt(q) = \argmax_{\vt \in \mathcal{C}(q)} \sum_{i=1}^n
\log\left(f_{\vt(i\slash n)}(Y_i)\right).
\end{equation}
The lower panel in Figure \ref{intro:example} shows an example of a SMUCE (red
solid line) for Gaussian observations. As stressed above, the multiscale
constraint on the r.h.s. of \eqref{intro:optprob} renders the SMUCE sensitive to
the multiscale nature of the signal $\vt$. The signal in Figure
\ref{intro:example} is a case in point: It exhibits large and small scales
simultaneously and remarkably the SMUCE $\hat\vt(q)$ recovers them both equally
well.

\begin{figure}[h!]
\includegraphics[width=\columnwidth]{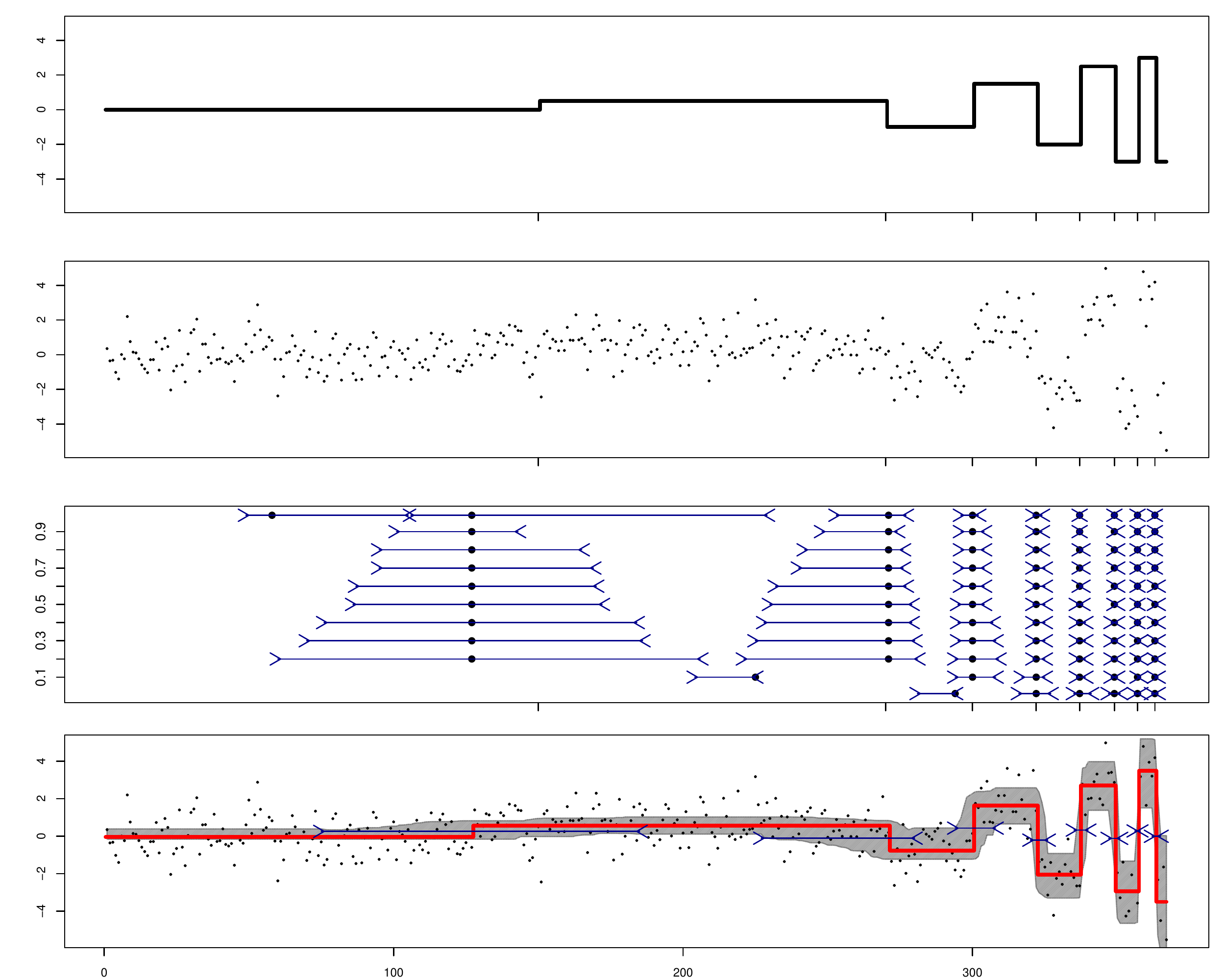}
\caption{From top to bottom: 1. True regression function $\vt$. 2. Gaussian
observations $Y$ with $n=367$ and variance $\sigma^2=1$. 3. Estimated change-point locations with confidence
intervals for different \am{ values of $\alpha$ ($y$-axis)}. 4. SMUCE $\hat\vt(q_\alpha)$ with confidence bands
(gray hatched area) and confidence intervals for the change-point locations
(inward pointed arrows) at $\alpha = 0.4$.}\label{intro:example}
\end{figure}


\subsection{Deviation bounds and confidence sets}
 
The parameter $q\in\R$ in \eqref{intro:optprob}   plays a crucial role because it governs the trade-off between
data-fit (the r.h.s. in \eqref{intro:optprob}) and parsimony (the l.h.s. in \eqref{intro:optprob}). It has an immediate statistical interpretation. 
From \eqref{intro:optprob} it follows that 
\begin{equation}\label{intro:prob:overest}
\Prob\left(\hat K(q) > K\right) \leq \Prob(T_n(Y,\vt) > q).
\end{equation}
Hence, by choosing $q=q_{1-\alpha}$ to be the $1-\alpha$-quantile of the
(asymptotic) null distribution of $T_n(Y,\vt)$, we can (asymptotically) control the
\emph{probability of overestimating} the number of change-points by $\alpha$.
In fact, we show that the null distribution of $T_n(Y,\vt)$ can be bounded asymptotically by a distribution
which does not depend on $\vt$ anymore (see Section \ref{subsec:limit}).
 It is noteworthy that for Gaussian observations this bound is even non-asymptotic (see Section
\ref{subsec:gauss:opt}). The third panel in Figure \ref{intro:example}
shows for different choices of $\alpha$ ($y$-axis) the corresponding estimates
for the change-point locations (black dots; the vertical ticks mark the true
change-point locations). The number of estimated change-points is monotonically increasing
in $\alpha$ in accordance with \eqref{intro:prob:overest} which guarantees at \am{error} level $\alpha$
that SMUCE has not more jumps than the true signal $\vt$. We emphasize that the SMUCE is
remarkably stable w.r.t. the choice of $\alpha$, i.e. the number of
change-points $K=8$ is estimated correctly for $0.2 \leq \alpha \leq 0.9$. Our
simulations in Section \ref{sim} confirm this stability even in
non-Gaussian scenarios. 

As mentioned before, \am{the threshold $q_{1-\alpha}$ for} SMUCE \am{automatically} controls the error of undersmoothing \eqref{intro:prob:overest}, i.e. the \emph{probability of overestimating} the number of change-points.
In addition, we prove an exponential inequality that
bounds the error of oversmoothing, i.e. the \emph{probability of underestimating} the number of change-points.
\ha{Any such bound necessarily has to depend on the magnitude of the signal $\vartheta$ on the smallest scale, as no method can recover arbitrary fine details for given sample size $n$, see \citep{Don88} for a similar argument in the context of density estimation. Our bound (see Theorem \ref{conscp:corunderest})
\begin{equation}\label{intro:prob:underest}
\Prob\left(\hat K(q) < K\right) \leq 2 K e^{- C n\lambda
\Delta^2}\left[e^{\frac{1}{2}\left(q+\sqrt{2\log(2e/\lambda)}\right)^2} + e^{- 3 C
n\lambda \Delta^2}\right]
\end{equation}
reflects this fact and indeed only depends on the smallest interval length $\lambda$, the smallest
absolute jump size $\Delta$ and the number of change-points $K$ of the true
regression function $\vt$.} Here, $C>0$ is some known universal constant only depending on the family of distributions (see Section \ref{sec:conscp}).

As a consequence of the inequalities \eqref{intro:prob:overest} and \eqref{intro:prob:underest}, $\mathcal{C}(q_{1-\alpha})$ in \eqref{intro:confset}
constitutes an asymptotic confidence set at level $1-\alpha$ and we will explain in Section
\ref{sec:alg:cb} how confidence bands for the graph of $\vt$ and confidence
intervals for its change-points can be obtained from this. See the lowest panel
of Figure \ref{intro:example} for illustration. 

Of course, honest (i.e. uniform) confidence sets cannot be obtained on the
entire set of step functions $\S$, as $\Delta$ and $\lambda$ can become arbitrarily small.
Nevertheless, we can show that simultaneously both, confidence bands for $\vt$ and
intervals for the change-points are \textit{asymptotically honest w.r.t. to a sequence of nested models} $\S^{(n)}\subset \S$
that satisfy
\begin{equation}\label{intro:neste_mod}
\frac{n}{\log n}\Delta_n^2 \lambda_n\ra \infty,\quad\text{ as }n\ra\infty,
\end{equation}
i.e. the confidence level $\alpha$ is kept uniformly over $\S^{(n)}$ as $n\ra\infty$
(c.f. Section \ref{sec:convset}). Here $\lambda_n$ and $\Delta_n$ denote the
smallest interval length and smallest absolute jump size in $\S^{(n)}$,
respectively.

\subsection{Choice of $q$}

Balancing the probabilities for over- and underestimation \am{in \eqref{intro:prob:overest} and \eqref{intro:prob:underest}} gives an upper bound on $\Prob(\hat K(q)\neq K)$, the
probability that the number of change-points is misspecified.
This bound  depends on $n,q,\lambda$ and $\Delta$ in an explicit way and opens
the door for several strategies to select $q$, \am{e.g.} such that $\Prob(\hat K(q)=K)$ is maximized. One may \am{additionally} incorporate prior information on $\Delta$ and $\lambda$ and we suggest a simple
way how to do this in Section \ref{sim:thresh}. 

\am{A further consequence of \eqref{intro:prob:overest} and \eqref{intro:prob:underest} is that under} a suitable choice of $q = q_n$ the probability of misspecification \am{$\Prob(\hat K(q_n)\neq K)$} tends to zero
 and hence $\hat K(q_n)$ converges to the true number of change-points
$K$ (model selection consistency), such that the underestimation error in \eqref{intro:prob:underest} vanishes 
exponentially fast.

Finally, we obtain explicit bounds on the precision of estimating the change-point locations
which again depend on $q,n,\lambda$ and $\Delta$. For any fixed $q>0$ they are recovered for all estimators in 
$\mathcal{C}(q)$, including SMUCE, at the optimal rate $1\slash n$ (up to a $\log$-factor).
Moreover, these bounds can be used to derive slower rates uniformly over nested models as in 
\eqref{intro:neste_mod} (see Section \ref{sec:convset}).
\subsection{Detection power for vanishing signals}

For the case of Gaussian observations we derive the detection power
of the multiscale statistic $T_n$ in \eqref{intro:mrstat}, i.e. we determine the
maximal rate at which a signal may vanish with increasing $n$ but still can be
detected with probability $1$, asymptotically. For the task of detecting a
single constant signal against a noisy background, we obtain the optimal rate
and constant (cf. \citep{DueSpok01, DueWal08,ChaWal11,Jen10}). We extend this result to
the case of an arbitrary number of change-points, retrieving the same optimal
rate but different constants (Section \ref{subsec:gauss:van}). Similar results
have been derived recently in \citep{Jen10} for sparse signals, where the
estimator takes into account the explicit knowledge of sparsity. We stress that
the SMUCE does not rely on any sparsity assumptions still it adapts
automatically to sparse signals due to its multiscale nature.
  
\subsection{Implementation, simulations and applications}  
The applicability of dynamic programming to the change-point problem has been
subject of research recently (cf. e.g.
\citep{BoyKemLieMunWit09,Fea06,FriKemLieWin08,HarLev10}). The SMUCE $\hat \vt(q)$ can
also be computed by a dynamic program due to the restriction of the local
likelihoods to the constant parts of candidate functions. This has already been observed by \citep{Hoe08}
for the multiscale constraint considered there. We prove that \eqref{intro:smre} can be rewritten into a minimization problem of a penalized cost
function with a particular data driven penalty (see Lemma \ref{impl:equivprob}).

Much in the spirit of the dynamic program suggested in \citep{KillFeaEck12}, our
implementation exploits the structure of the constraint set in
\eqref{intro:smre} to include pruning steps. These reduce the worst case
computation time $\bigo(n^2)$ considerably in practice and makes it applicable
to large data sets. Simultaneously, the algorithm returns a confidence band for
the graph of $\vt$ as well as confidence intervals for the location of the
change-points (Section \ref{sec:alg}), the latter without any additional cost.
An R-package (stepR) including an implementation of SMUCE is available online
\footnote{\url{http://www.stochastik.math.uni-goettingen.de/smuce}}.

Extensive simulations reveal that the SMUCE is competitive with
(and indeed often outperforms) state-of-the-art methods for the
change-point problem which all have been tailor-made to specific exponential
families (Section \ref{sim}). Our simulation study includes the CBS method
\citep{OlsVenLucWig04}, the fused lasso \citep{TibSauRosZhuKni04} and the modified BIC \citep{ZhaSie07}
for Gaussian regression, the multiscale estimator in \citep{DavHoeKra12} for piecewise
constant volatility and the extended taut string method for quantile
regression in \citep{DueKov09}.
In our simulations we consider several risk measures, \am{including} the MSE and 
the model selection error $\Prob(\hat K \neq K)$.
 Moreover, we study the feasibility of our
approach for different real-world data sets; including two benchmark examples
from genetic engineering \citep{Lai} and a new example from photoemission spectroscopy
\citep{Hue03} which amounts to Poisson change-point regression.
\am{Finally, in Section \ref{sec:disc}, we briefly discuss possible extensions to serially dependent data, among others}

\subsection{Literature survey and connections to existing work}

The problem of detecting changes in the characteristics of a
sequence of observations has a long history in statistics and related
fields, dating back to the 1950's (see e.g.\citep{Pag55}). In recent years, it
experienced a renaissance in the context of regression analysis due to novel
applications that mainly came along with the rapid development in genetic engineering 
\citep{BraMueMue00,OlsVenLucWig04,ZhaSie07,Jen10,LebPic11} and financial econometrics (cf.
\citep{IncTia94,LavTey07,DavHoeKra12, Spo09}). Due to the widespread occurrence of
change-point problems in different communities \am{and areas of applications}, such as statistics \am{\citep{CarMueSie94}}, electrical
engineering and \am{signal processing \citep{BlyBunMeiMul12}, mobile communication \citep{ZhaDanCan09}}, machine learning \am{\citep{HarLev08}}, biophysics \am{\citep{Hot12}}, \am{quantum optics \citep{Sch12}}, econometrics and \am{quality control} \am{\citep{BayPer98}} and biology \am{\citep{Sie13}}, an exhaustive list of existing
methods is beyond reach. For a selective survey, we refer the reader also to the books
\citep{BroDar93,CsoHor97,BasNik93,CheGup00, Wu05} and the extensive list in \citep{KhoAsg08}.

Our approach as outlined above can be considered as a hybrid method of two
well-established approaches to the change-point problem:

\emph{Likelihood ratio} and related statistics, on the one
hand, are frequently employed to test for a change in the parameter of the
distribution family and to construct confidence regions for change-point locations. Approaches of this
type date back as far as \citep{CheZac64,KanZac66}  and have gained considerable
attention afterwards \citep{Hin70,HinHin70,HusAnt03,Wor83,Wor86,Sie88,Due91} and
\citep{Bha87,SieYak00,AriCanDur11} for generalizations to the multivariate
case). The likelihood ratio test was also extensively studied for sequential
change-point analysis \citep{Sie86,YakPol98,SieVen95}. All these methods
are primarily designed to detect a predefined maximal number (mostly one) of
change-points.

On the other hand, if the number of change-points is unknown, an additional
\emph{model selection step} is required, which can be achieved by proper
penalization of model complexity, e.g. measured by the number of change-points itself or
by surrogates for it. This is often approached by maximizing a \emph{penalized
likelihood function} of the form 
\begin{equation*}
\vt\mapsto l(Y,\vt) - \text{pen}(\vt)
\end{equation*}  
over a suitable space of functions, e.g. $\S$ as in this paper or functions of bounded variation
\citep{MamGee97}, etc. Here
$l(Y,\vt)$ is the (log) likelihood function. The penalty term $\text{pen}(\vt)$ penalizes the complexity of $\vt$ and prevents overfitting. It increases with the dimension of
the model and serves as a model selection criterion. First
approaches include BIC-type penalties \citep{Yao88} and more sophisticated
penalties have been advocated later on (see e.g.
\citep{YaoAu89,LavMou00,BraMueMue00,BirMas01,Lav05,LavTey07,BoyKemLieMunWit09,ArlCelHar12, WitKemWinLie08, WinLie02}).
 Further prominent penalization approaches include the fused lasso procedure
(see \citep{FriHasHoeTib07,TibSauRosZhuKni04} and \citep{HarLev10}) that uses a linear combination of the
total-variation and the $\ell^1$-norm penalty as a convex surrogate for the number of
change-points which has been primarily designed for the situation when $\vt$ is sparse.
Recently, aggregation methods \citep{RigTsy12} have been advocated recently for the change-point regression problem as well.

Most similar in spirit to our approach are estimators which minimize target functionals under a statistical multiscale constraint.
For some early references see \citep{Nem85, Don95} and more recently \citep{DavKovMei09, CanTao07, FriMarMun12, DavKov01}.
In our case this target functional equals the number of change-points.

The multiscale calibration in \eqref{intro:mrstat} is based on the work of \citep{DueSpok01, DueWal08, ChaWal11}.
 Multiscale penalization methods have been suggested in \citep{ZhaSie07, KolNow04}, multiscale partitioning methods including binary segmentation in \citep{SenSri75,Vos81,OlsVenLucWig04, Fry12}, and recursive partitioning in \citep{KolNow05}.

Aside to the connection to frequentist's work cited above, we
claim that our analysis also provides an interface for incorporating a priori
information on \ha{the true signal} into the estimator (see Section \ref{sim:thresh}). We stress that for
minimizing the bounds in \eqref{intro:prob:overest} and \eqref{intro:prob:underest} on the model selection error $\Prob(\hat K(q)\neq K)$ it is not necessary
to include full priors on the space of step functions $\S$. 
Instead it suffices to simply specify a prior on the
smallest interval length $\lambda$ and the smallest absolute jump
size $\Delta$. The parameter choice strategy discussed in Section
\ref{sim:thresh} or the limiting distribution of $T_n(Y,\vt)$ in Section
\ref{subsec:limit}, for instance, can be refined within such a Bayesian
framework. This, however, will not be discussed in this paper in detail and is postponed
to future work. For recent work on a Bayesian approach to the change-point
problem we refer to \citep{DuKou12, Fea06,LuoRozNue12, RigLebRob12} and the references therein.

We finally stress that there is a conceptual analogy
of SMUCE to the Dantzig selector as introduced in \citep{CanTao07} for estimating sparse signals in gaussian high
dimensional linear regression models (see \citep{JamRad09} for an extension to exponential families). Here the $\ell_1$-norm of the signal is to
be minimized subject to the constraint that the residuals are pointwise within
the noise level. The SMUCE, in contrast, minimizes the $\ell_0$-norm of the
discrete derivative of the signal subject to the constraint that the residuals are tested to contain no signal on \emph{all scales}.
We will briefly address this and other relations to recent concepts in high dimensional statistics in
a discussion in Section \ref{sec:disc}. In summary, the change-point problem is an ``$n=p$'' problem and
hence substantially different from high dimensional regression where ``$p\gg n$''. As we will show, multiscale detection of sparse signals becomes then 
possible without any sparsity assumption entering the estimator. Another major statistical consequence of this paper
is that post model selection inference is doable over a large range of scales uniformly over nested models in the sense of \eqref{intro:neste_mod}.
\section{Theory}\label{sec:mrstat}

This section summarizes our main theoretical findings.  In Section
\ref{sec:conscp} we discuss consistency of the estimated number of change-points.
This result follows from an exponential bound for the probability of
underestimating the number of change-points on the one hand.
On the other hand we show how to control the probability of overestimating the 
number of change-points by means of the limiting distribution of $T_n(Y,\vt)$
  as $n\ra\infty$ (cf.
Section \ref{subsec:limit}). We give improved results,
including a non-asymptotic bound for the probability of overestimating the
number of change-points, for Gaussian observations (cf. Sections
\ref{subsec:gauss:opt} \& \ref{subsec:gauss:van}). In Section
\ref{sec:convset} we finally show that the change-point locations can be
recovered as fast as the sampling rate up to a $\log$-factor and discuss how 
asymptotically honest confidence sets for $\vt$ can be constructed over a suitable sequence of nested models.

\subsection{Notation and model} \label{sec:model}

We shall henceforth assume that
$\mathcal{F} = \set{F_\theta}_{\theta \in \Theta}$ is a one-dimensional, 
standard exponential family with $\nu$-densities
\begin{equation}\label{model:expfam}
f_\theta(x) = \exp\left(\theta x -
\psi(\theta)\right)  ,\quad x\in \R.
\end{equation}
Here $\Theta = \set{\theta
\in \R~:~ \int_\R \exp(\theta x)\diff\nu(x) < \infty}\subseteq\R$ denotes the natural parameter space. We will assume that
$\mathcal{F}$ is \emph{regular} and \emph{minimal} which means that $\Theta$ is
an open interval and that the cumulant transform $\psi$ is strictly convex on
$\Theta$. We will frequently make use of the functions
\begin{equation}\label{model:not_varmean}
m(\theta) := \dot\psi(\theta) = \E{X} \quad\text{ and }\quad
v(\theta) := \ddot\psi(\theta) = \Var{X},
\end{equation}
for $X\sim F_\theta$. Note that $m$ and $v$ are strictly increasing and positive
on $\Theta$, respectively.  

\subsubsection{Observation model and step functions} We assume that $Y =
(Y_1,\ldots,Y_n)$ are independent observations given by \eqref{intro:model}
where $\vartheta:[0, 1) \ra \Theta$ is a right continuous step function,
that is
\begin{equation}\label{model:regfunc}
\vartheta(t) = \sum_{k=0}^K \theta_k \eye_{[\tau_k,\tau_{k+1})}(t),
\end{equation}
where $0=\tau_0<\tau_1<\ldots<\tau_K<\tau_{K+1}=1$ are the change-point
locations and $\theta_k\in \Theta$ the corresponding intensities, such that
$\theta_k\neq \theta_{k+1}$ for $k=0,\ldots,K$. The collection of step functions
on $[0,1)$ with values in $\Theta$ and an arbitrary but finite number of
change-points will be denoted by $\S$. For $\vt\in \S$ as in
\eqref{model:regfunc} we denote by $J(\vt) = (\tau_1,\ldots,\tau_K)$  the
increasingly ordered vector of change-points and by $\# J(\vt) = K\in\N$ its length.
We will denote the set of step functions
with $K$ change-points and change-point locations restricted to the sample grid by
$\S_n[K]\subset\S$.

For any estimator $\hat\vartheta$ of $\vt\in\S$, the
estimated number of change-points will be denoted by $\#J(\hat\vt) = \hat K$, the change-point locations
by $J(\hat\vt) = (\hat\tau_1,\ldots,\hat\tau_{\hat K})$ and we set  $\hat
\theta_k = \hat \vartheta(t)$ for $t\in [\hat \tau_k, \hat \tau_{k+1})$. For simplicity, for each $n\in \N$ we
restrict to estimators which have change-points only at sampling points, i.e.
$\hat\vt\in \S_n[K]$ with $\hat \tau_k = \hat l_k \slash n$ for some $1\leq \hat
l_k\leq n$. To keep the presentation simple, throughout the following we restrict ourselves to an equidistant sampling scheme as in
\eqref{intro:model}. However, we mention that extensions to more general designs are possible.

\subsubsection{Multiscale statistic} 
Let $1\leq i \leq j \leq n$. Then, the likelihood ratio statistic
$T_i^j(Y,\theta)$ in \eqref{intro:likeratstat} can be rewritten into
\begin{equation*}
T_i^j(Y,\theta_0) = \sup_{\theta \in\Theta}\left(\sum_{l=i}^{j} (\theta Y_l -
\psi(\theta))\right) - \sum_{l=i}^{j}(\theta_0 Y_l -
\psi(\theta_0)).
\end{equation*}
Introducing the notation $\phi(x) = \sup_{\theta\in\Theta} \theta x -
\psi(\theta)$ for the \emph{Legendre-Fenchel conjugate} of $\psi$ and $J(x,
\theta) = \phi(x) - (\theta x - \psi(\theta))$ we find
\begin{equation*}
T_i^j(Y,\theta_0) = (j-i+1) J(\overline Y_i^j, \theta_0) \geq 0,
\end{equation*}
where $\overline Y_i^j = (\sum_{i\leq l\leq j} Y_l)\slash (j-i+1)$. The
multiscale  statistic $T_n(Y,\vt)$ in \eqref{intro:mrstat} was defined to be the (scale calibrated)
maximum over all $\sqrt{2T_i^j}$ such that $\hat l_k \leq i \leq j < \hat
l_{k+1}$ for some $0\leq k\leq \hat K$. As mentioned in the introduction we sometimes will restrict the minimal
interval length (scale) by a sequence of lower bounds $(c_n)_{n\in\N}$ tending
to zero. In order to ensure that the asymptotic null distribution is non
degenerate, we assume for non-Gaussian families (see also \citep{SchMunDue11})
\begin{equation}\label{def:lowerbound}
n^{-1}\log^3 n\slash c_n\ra 0.
\end{equation}
Then, the modified version of \eqref{intro:mrstat} reads as
\begin{equation}\label{smre:mrstat}
T_n(Y,\vt;c_n) = \max_{0\leq k\leq \ha{K}}\max_{\substack{ l_k \leq
i\leq j <  l_{k+1} \\ (j-i+1)/n\geq c_n}} \left(\sqrt{2T_i^j(Y,\theta_k)} -
\sqrt{2\log \frac{n e}{j-i+1}}\right).
\end{equation}

\subsection{Asymptotic null distribution}\label{subsec:limit}

We give a representation of the limiting
distribution of the multiscale statistic $T_n$ in \eqref{smre:mrstat} in
terms of 
\begin{equation}\label{limit:limitstat} M :=  \sup_{0\leq s < t \leq 1}
\left(\frac{\abs{B(t) - B(s)}}{\sqrt{t-s}} - \sqrt{2\log\frac{e}{t-s}}\right),
\end{equation}
where $(B(t))_{t\geq 0}$ denotes a standard Brownian motion. We stress that the
statistic $M$ is finite almost surely and has a continuous
distribution supported on $[0,\infty)$ (cf. \citep{DueSpok01, DuePitZho06}).

\begin{thm}\label{limit:mainthm}
Assume that $(c_n)_{n\in\N}$ satisfies \eqref{def:lowerbound}. Then,  
\begin{equation}\label{limit:mainres}
T_n(Y,\vartheta;c_n)\stackrel{D}{\ra} \max_{0\leq k\leq K}\sup_{\tau_k\leq
s<t\leq\tau_{k+1}} \left(\frac{\abs{B(t) - B(s)}}{\sqrt{t-s}} -
\sqrt{2\log\frac{e}{t-s}}\right).
\end{equation}
Further, let $M_0,\ldots,M_K$ be independent copies of $M$ as in \eqref{limit:limitstat}. Then, the right hand
side in \eqref{limit:mainres} is stochastically bounded from above by $M$ and from
below by
\begin{equation*}
\max_{0\leq k\leq K} \left(M_k - \sqrt{2\log\frac{1}{\tau_{k+1} -
\tau_k}}\right).
\end{equation*} 
\end{thm}
 
It is important to note that the limit distribution in \eqref{limit:mainres} (same as the lower bound)
depends on the unknown regression function $\vartheta$ only through the number
of change-points $K$ and the change-point locations $\tau_k$, i.e. the function
values of $\vt$ do not play a role. From the upper bound in Theorem
\ref{limit:mainthm} we obtain
\begin{equation} \label{limit:confineq}
\lim_{n\ra\infty}\Prob\left(T_n(Y,\vartheta;c_n) \leq q_\alpha \right) \geq
\alpha,
\end{equation}
with $q_\alpha$ being the $\alpha$-quantile of $M$.
In practice the distribution of $M$ is obtained by simulations.
In Section \ref{subsec:gauss:opt} we will see that for the Gaussian case even a nonasymptotic version
of Theorem \ref{limit:mainthm} can be obtained, which allows for finite sample refinement of the null distribution of $T_n$.
\ha{As the asymptotics is rather slow, this finite sample correction is helpful even for relatively large samples, say if $n$ is of the order of a few thousands. This is highlighted in Figure \ref{limit:distM} where it becomes apparent that the empirical null distributions for finite samples, obtained from simulations, is in general not supported 
in $[0,\infty)$}. 
\begin{figure}[htp]
 \includegraphics[width=0.7\columnwidth]{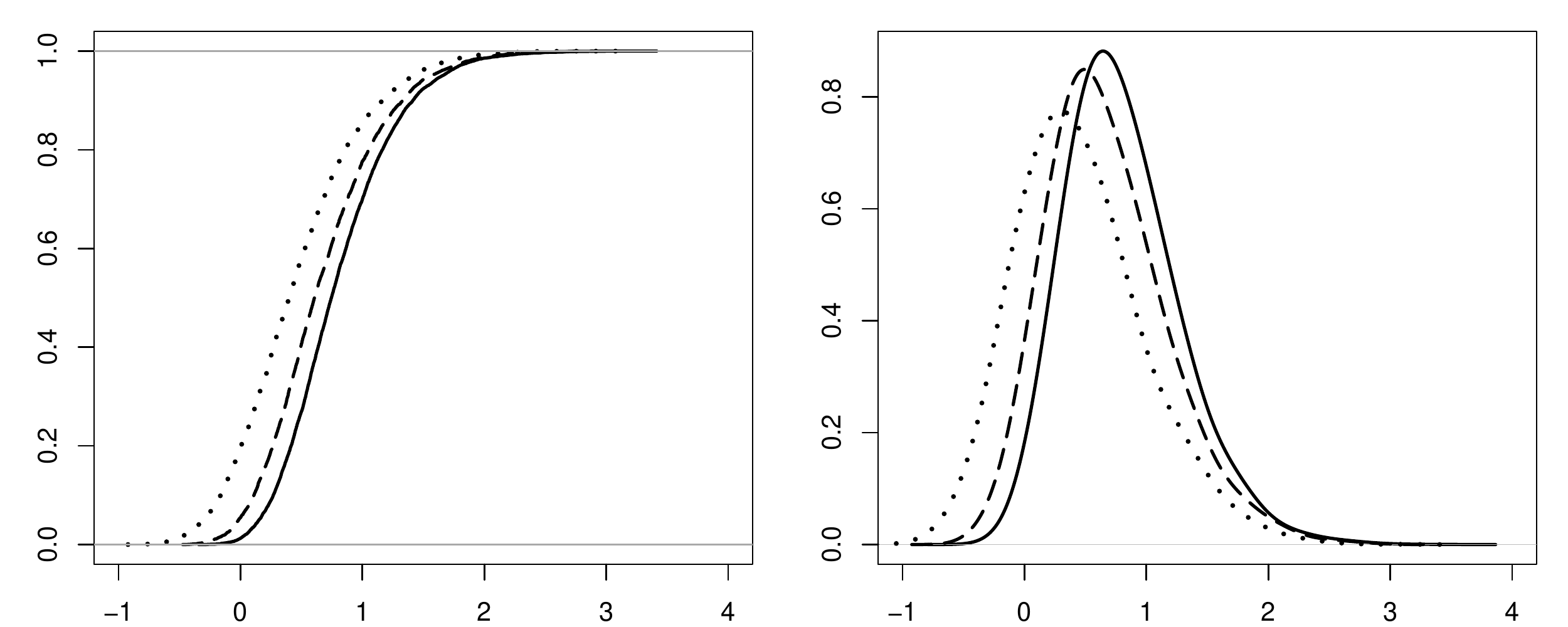}
 \caption{Simulations of the cdf (left) and density (right) of $M$ as in
 \eqref{limit:limitstat} for \ha{$n=50$(dotted line), $n=500$(dashed line) and $n=5000$(solid line) equidistant discretization points}.}
 \label{limit:distM} 
\end{figure}
To the best of our knowledge, it is an open and challenging problem to derive \am{tight} bounds for the tails of $M$ (cf.
\citep{DueSpok01, DueWal08, DuePitZho06}) which is not addressed in this article.
By such bounds the probability of overestimating the number of
change-points could be controlled explicitly, as we will see in the upcoming
section.
\ha{Moreover, we point out that the inequality in \eqref{limit:confineq} is not sharp, if the true functions has at least one change-point. This is due to the fact that we bound \am{$T_n$ in \eqref{limit:confineq} by $q_\alpha$, the quantile of $M$ which serves as the bound for the r.h.s. in \eqref{limit:mainres}}. For an illustration of this, Figure \ref{limit:exactnulldis} shows P-P plots of the exact null distribution of signals with $2$, $4$ and $10$ equidistant change-points against the null distribution of a signal without change-points for sample size $n=500$.}
\am{Of course, further information on the minimal number and location of change-points can be used to improve the distributional bound by $M$ in Theorem \ref{limit:mainthm}. We will not pursue this further.}

\begin{figure}[htp]
 \includegraphics[width=0.7\columnwidth]{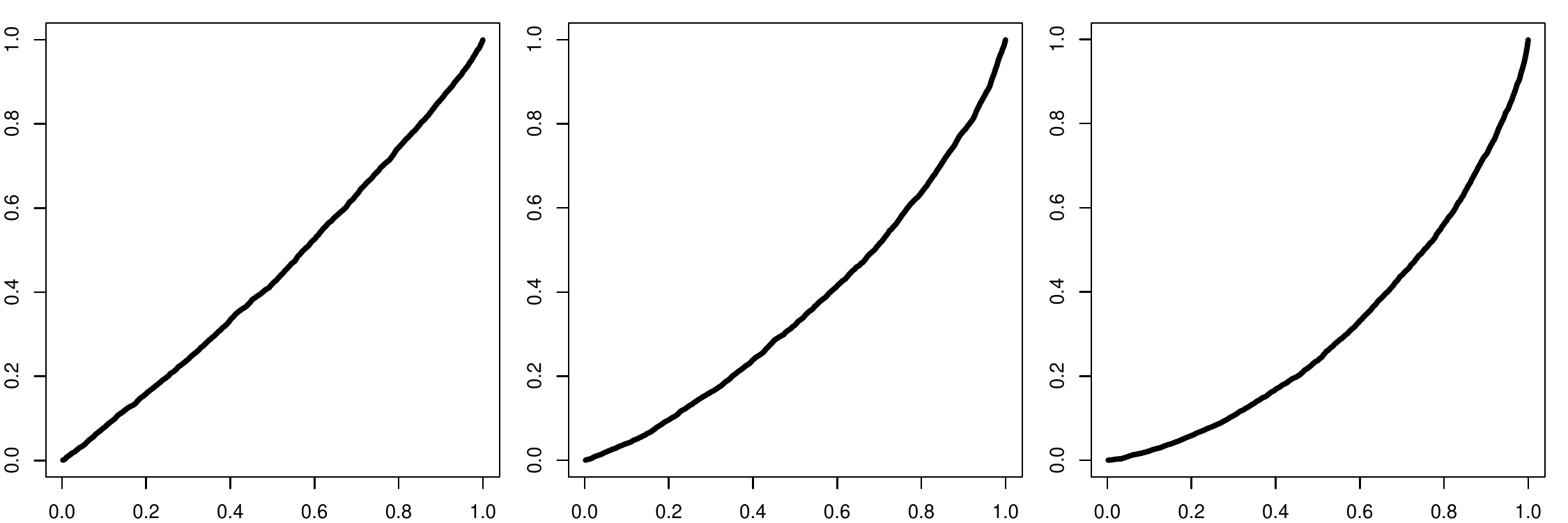}
 \caption{ \ha{Probability-Probability plots of the empirical null distribution of a signal without change-points ($x$-axis) against signals with $2$(left), $5$(middle) and $10$(right) equidistant change-points ($y$-axis) for $n=500$}.}
 \label{limit:exactnulldis} 
\end{figure}

\subsection{Exponential inequality for the estimated number of change-points}\label{sec:conscp}

In this section we derive explicit bounds on the probability that $\hat K(q)$ as
defined in \eqref{intro:optprob} underestimates the true number of change-points
$K$. In combination with the results in Section \ref{subsec:limit}, these bounds
will imply model selection consistency, i.e. $\Prob(\hat K(q_n)= K)\ra 1$ for a
suitable sequence of thresholds $(q_n)_{n\in\N}$ in \eqref{intro:optprob}.

We first note, that with the additional constraint in \eqref{smre:mrstat} on the
minimal interval length, the estimated number of change-points is given by
\begin{equation}\label{conscp:estnocp} 
\hat K(q)  =  \min \set{ K\in\N ~:~\exists \vartheta \in
\S_n[K]: T_n(Y,\vartheta;c_n) \leq q},\quad q\in\R.
\end{equation}
Now let $\Delta$ and $\lambda$ be the smallest absolute jump size and the smallest
interval length of the true regression function $\vt\in \S$, respectively and assume that
$\vt(t)\in [\underline\theta, \overline \theta]$ for all $t\in [0,1]$. We 
give the aforementioned exponential upper bound on the probability that
the number of change-points is underestimated. The results follows from
the general exponential inequality in the supplement, Theorem \ref{conscp:thmunderest}.

\begin{thm}[Underestimation bound]\label{conscp:corunderest}
Let $q\in \R$ and $\hat K(q)$ be defined as in \eqref{conscp:estnocp} with $\lambda \geq 2 c_n$. Then, there exists a constant
$C=C(\mathcal{F},\underline\theta, \overline\theta)>0$ s.t.
\begin{equation}
\Prob\left(\hat K(q) < K\right) \leq 2 K e^{- C n\lambda
\Delta^2}\left[e^{\frac{1}{2}\left(q+\sqrt{2\log(2e/\lambda)}\right)^2} + e^{-
3C n\lambda \Delta^2}\right],\quad n\in \N.
\end{equation}
\end{thm}
From Theorem \ref{conscp:thmunderest} and Lemma \ref{conscp:kappalemma} it
follows that 
\begin{equation}\label{conscp:constant}
C(\mathcal{F},\underline\theta, \overline\theta) =
\frac{1}{32}\frac{\inf_{\underline\theta\leq \theta \leq\overline\theta}
v(t)^2}{\sup_{\underline\theta\leq \theta \leq \overline\theta} v(t)},
\end{equation}
which gives $C=1\slash 32$ for the Gaussian family and $C=\underline
\mu ^2\slash (32\overline\mu)$ for the Poisson family, given
$m(\vt)\in[\underline\mu, \overline\mu]$ in the latter case.

On the one hand, if
$q=q_n$ and $q_n\slash \sqrt{n}\ra 0$ as $n\ra \infty$, it becomes clear from Theorem \ref{conscp:corunderest} that $\hat K(q_n)\geq K$ with high probability. On the other hand, it follows from Theorem
\ref{limit:mainthm} that $T_n(Y,\vartheta;c_n)$ is bounded almost surely
as $n\ra \infty$ if $c_n$ is as in \eqref{def:lowerbound}. This in turn implies that the probability for $\hat
K(q_n)\leq K$ tends to $1$, since
\begin{equation}\label{conscp:overestim}
\Prob\left(\hat K(q_n) > K \right) \leq \Prob(T_n(Y,\vartheta;c_n) >
q_n)\ra 0,
\end{equation}
whenever $q_n\ra \infty$, as $n \ra \infty$. Thus, we summarize

\begin{thm}[Model selection consistency]\label{conscp:mainthm}
Let the assumptions of Theorems \ref{limit:mainthm} and \ref{conscp:corunderest}
hold and additionally assume that $q_n\ra \infty$ and $q_n\slash \sqrt n\ra 0$ as
$n\ra\infty$. Then,
\begin{equation*}
\lim_{n\ra\infty} \Prob( \hat K(q_n) = K) = 1.
\end{equation*}
\end{thm}

Giving a non-asymptotic bound for the probability for overestimating the true number
of change-points (in the spirit of \eqref{conscp:overestim})
appears to be rather difficult in general. For the Gaussian case though this is possible, as we will show in the next section.

\subsection{Gaussian observations}\label{subsec:gauss:opt}

We now derive sharper results for the case when $\mathcal{F}$ is the
Gaussian family of distributions with constant variance. In this case
\eqref{intro:model} reads as
\begin{equation}\label{gauss:opt:model}
Y_i = \mu(i\slash n) + \sigma \eps_i,\quad i=1,\ldots,n
\end{equation}
where $\eps_1,\ldots,\eps_n$ are independent $\mathcal{N}(0,1)$ random
variables, $\sigma>0$ and $\mu\in\S$ denotes the expectation of $Y$. To ease notation we assume in the following that $\sigma=1$. For the general case 
replace $\Delta$ by $\Delta/\sigma$.

In the Gaussian case it is possible to get rid of the lower bound for the
smallest scales $c_n$ as in \eqref{def:lowerbound}  because the strong approximation by Gaussian observations in the proof of Theorem \ref{limit:mainthm} 
becomes superfluous. We obtain the following {\it non-asymptotic} result on the
null distribution.

\begin{thm}[Null Distribution of $T_n$] \label{gauss:limdis}
For any $n\in \N$
\begin{equation*}
\max_{0\leq k\leq K} \left( M_k -\sqrt{2\log\frac{1}{\tau_{k+1} -\tau_k}} \right)
\stackrel{\mathcal{D}}{\leq} T_n(Y,\vartheta) \stackrel{\mathcal{D}}{\leq} M^{(n)} \stackrel{\mathcal{D}}{\leq} M,
\end{equation*}
where $M^{(n)}$ is as $M$ in \eqref{limit:limitstat} where the supremum is only taken over the system of discrete intervals $[i/n,j/n]$.
\end{thm}
In contrast to Theorem  \ref{limit:mainthm}, this result is nonasymptotic and the inequality holds for any sample size.
 For this reason, we get the following improved upper bound for
the probability of overestimating the number of change-points.
\begin{cor}[Overestimation bound]\label{gauss:opt:cor1}
Let $q\in \R$ and $\hat K(q)$ be defined as in \eqref{conscp:estnocp}. Then for any $n \in \N$
\begin{equation*}
 \Prob \left( \hat K (q) > K \right) \leq \Prob\left( M \geq q
 \right).
\end{equation*} 
\end{cor}
This corresponds to the ``worst case scenario'' for overestimation when the true signal $\vt$ has no jump.

For the probability of underestimating the number of change-points, we can improve
 Theorem \ref{conscp:corunderest} for Gaussian observations (see Theorem \ref{gauss:opt:mainthm}) to 

\begin{equation}\label{gauss:q:under}
 \Prob \left( \hat K (q) < K  \right) \leq  2 K\left[\exp\left( - \frac{1}{8}
 \left( \frac{\Delta\sqrt{\lambda n}}{2\sqrt{2}} - 
 q-\sqrt{2\log\frac{2e}{\lambda}}
 \right)_+^2 \right) +  \exp\left(- \frac{\lambda n
 \Delta^2}{16} \right)\right].
\end{equation}

\subsection{Multiscale detection of vanishing signals for Gaussian
observations}\label{subsec:gauss:van} 
We will now discuss the ability of SMUCE to
detect vanishing changes in a signal. We begin with the problem of detecting a
signal on a single interval against an unknown background.

\begin{thm}\label{gauss:opt:unknownbackground}
 Let \ha{$\vt_n(t) = \theta_0 + \delta_n I_n(t)$ for some
 $\theta_0,\theta_0+\delta_n\in\Theta$ and for some sequence of intervals $I_n\subset[0,1]$} and
 $Y$ be given by \eqref{gauss:opt:model}.  Further let $(q_n)_{n\in\N}$
 be bounded away from zero and assume
 \begin{enumerate}
 \item for signals on a large scale (i.e. $\liminf \abs{I_n} > 0$),
 that $\sqrt{\abs{I_n}n}\delta_n \slash {q_n}\ra\infty$, 
 \item for signals on a small scale (i.e. $\abs{I_n}\ra 0$), that
 $\sqrt{\abs{I_n}n}\delta_n \geq (\sqrt{2} + \eps_n ) \sqrt{\log (1/\abs{I_n})}$
 with $\eps_n$, s.t. $\eps_n \sqrt {\log(1\slash \abs{I_n})} \ra \infty$ and 
 $\sup_{n\in\N} q_n / (\eps_n\sqrt{\log(1\slash \abs{I_n})})<1$.
\end{enumerate}
 Then, 
\begin{equation}\label{gauss:opt:unknownbackgroundeqn}
 \sup_{\vt_0\equiv\theta\in\Theta} \Prob_{\vt_n} \left(  T_n(Y,\vartheta_0)
 \leq q_n \right)\ra 0.
\end{equation} 
\end{thm} 
 
Theorem \ref{gauss:opt:unknownbackground} gives sufficient conditions on the
signals $\vt_n$ (through the interval length $\abs{I_n}$ and the jump height
$\delta_n$) as well as on the thresholds $q_n$ such that the multiscale

statistic $T_n$ detects the signals with probability $1$, asymptotically; put
differently, this means $\Prob(\hat K (q_n)>0)\ra 1$. We stress that the
above result is optimal in the following sense: No test can detect signals satisfying $\sqrt{\abs{I_n} n}\delta \geq (\sqrt{2}-\eps_n) \sqrt{\log(1/I_n)}$ with asymptotic power $1$ (see
\citep{DueSpok01,ChaWal11,Jen10}). 

For the special case, when $q_n \equiv
q_\alpha$ is a fixed $\alpha$-quantile of the null distribution $T_n(Y,\vt_n)$ (or of the limiting distribution $M$ in
\eqref{limit:limitstat}), the result boils down to the findings in
\citep{DueSpok01, ChaWal11}. In particular, aside to the optimal asymptotic power
\eqref{gauss:opt:unknownbackgroundeqn}, the error of first kind is bounded by $\alpha$. The result in Theorem \ref{gauss:opt:unknownbackground} goes beyond
that and allows to shrink the error of first kind to zero asymptotically, by
choosing $q_n\ra\infty$. 

We finally generalize the results in Theorem
\ref{gauss:opt:unknownbackground} to the case when $\vt\in
\mathcal{S}$ has more than one change-point. To be more precise, we formulate
conditions on the smallest interval and the smallest jump in $\vt$ such that no
change-point is missed asymptotically. 

\begin{thm}\label{gauss:opt:nojumpmissing}
Let $(\vt_n)_{n\in \N}$ be a sequence in $\mathcal{S}$ with $K_n$
change-points and denote by $\Delta_n$ and $\lambda_n$ the smallest absolute jump size and
smallest interval in $\vt_n$, respectively. Further, assume that $q_n$ is
bounded away from zero and
\begin{enumerate}
 \item for signals on large scales (i.e. $\liminf \lambda _n > 0$), that 
 $\sqrt{\lambda_n n}\Delta_n\slash q_n \ra \infty$.
 \item for signals on small scales (i.e. $\lambda_n\ra 0$) with $K_n$ bounded, that 
 $\sqrt{\lambda_n n}\Delta_n \geq (4 + \eps_n) \sqrt{\log(1\slash
 \lambda_n)}$ with $\eps_n \sqrt{\log(1\slash \lambda_n)} \ra \infty$ and
 $\sup_{n\in\N}q_n\slash (\eps_n \sqrt{\log(1\slash \lambda_n)}) <\ha{1/(2\sqrt{2})}$.
 \item \am{the same as in (2), with $K_n$ unbounded and the constant $12$ instead of $4$.}
\end{enumerate}
Then, 
\begin{equation*}
\Prob_{\vt_n}\left( \hat K(q_n) \geq K_n \right)\ra 1.
\end{equation*}
\end{thm} 

Theorem \ref{gauss:opt:nojumpmissing} amounts to say that the
statistic $T_n$ is capable of detecting multiple change-points simultaneously \emph{at the same
optimal rate} (in terms of the smallest interval and jump) as a single change-point. The
only difference being the constants that bound the size of the signals that can
be detected. These increase with the complexity of the problem: $\sqrt{2}$ for a
single change against an unknown background, $4$ for a bounded (but
unknown), and $12$ for an unbounded
number of change-points. In \citep{Jen10} it was shown that for step functions
that exhibit certain sparsity patterns the optimal constant $\sqrt{2}$ can be
achieved. It is important to note that we do not make any sparsity assumption
on the true signal.
Finally we mention an analogy to Theorem 4.1. of \citep{DueWal08} in the context of detecting local increases and decreases
of a density. As in Theorem \ref{gauss:opt:nojumpmissing} only the constants and not the detection rates changes with the complexity
of the alternatives.
\subsection{Estimation of change-point locations and simultaneous confidence sets}\label{sec:convset}
In this section we will provide several results on confidence sets associated with \text{SMUCE}.
We will see that these are linked in a natural way to estimation of change-point locations.
We generalize the set $\mathcal{C}(q)$ in \eqref{intro:confset} by replacing
$T_n(Y,\vt)$ in \eqref{intro:mrstat} with $T_n(Y,\vt,c_n)$ as in \eqref{smre:mrstat} and consider the set of solutions of the optimization problem 
\begin{equation}\label{convset:optprob}
\inf_{\vt \in \S} \# J(\vt)\quad\text{ s.t. }\quad T_n(Y,\vt;c_n)\leq q.
\end{equation}
Any candidate in $\mathcal{C}(q)$ recovers the change-point locations of the true regression function $\vt$
with the same convergence rate. It is determined by the smallest scale $(c_n)_{n\in\N}$
for the considered
interval lengths in the multiscale statistic $T_n$ in \eqref{smre:mrstat} and hence equals the sampling rate up to a log factor.

\begin{thm}\label{convset:mainthm}
Let $q\in\R$ and $\mathcal{C}(q)$ be the set of solutions of
\eqref{convset:optprob} and $(c_n)_{n\in\N}$ a sequence in
$(0,1]$. Further let $C = C(\mathcal{F},\underline\theta,
\overline\theta) > 0$ as in \eqref{conscp:constant}. Then, for all $n\in\N$
\begin{equation*}
\Prob\left(\sup_{\hat\vt \in \mathcal{C}(q)}\kl{\max_{\tau \in J(\vt)}
\min_{\hat \tau \in J(\hat \vt)} } \abs{\hat \tau - \tau} > c_n \right) \leq 2 K
e^{- 2C nc_n \Delta^2}\left[e^{\frac{1}{2}\left(q+\sqrt{2\log(e/c_n)}\right)^2} +
e^{- \ha{6} C n c_n \Delta^2}\right].
\end{equation*}
\end{thm}

For a fixed signal $\vt\in\S$, a sufficient condition for the r.h.s. in Theorem 
\ref{convset:mainthm} to vanish as $n\ra \infty$ is 
\begin{equation*}
 c_n \geq \frac{1}{\Delta^2 C} \frac{\log{n}}{n}.
\end{equation*}
Here the constant $C$ matters, e.g.
in the Gaussian case $C=1/8$ (cf. Section \ref{sec:conscp}). 
This improves several results obtained for other methods, e.g. in \citep{HarLev10}
 for a total variation penalized estimator a $\log^2{n}/n$ rate has been shown.

 In the following we will apply Theorem \ref{convset:mainthm} to determine
 subclasses of $\S$ in which the change-point locations are reconstructed \emph{uniformly} with rate
$c_n$. These subclasses are delimited by conditions on the smallest absolute
jump height $\Delta_n$ and on the number of change-points $K_n$ (or the
smallest interval lengths $\lambda_n$ by using the relation $K_n\leq 1\slash
\lambda_n$) of its members. For instance, the rate function $c_n = n^{-\beta}$
with some $\beta\in [0,1)$ implies the condition
\begin{equation*}
\frac{n^{\beta}\exp(-n^{1-\beta} \Delta_n)}{\lambda_n}\ra 0.
\end{equation*}
The choice $\beta=0$ gives the largest subclass but no convergence rate is
guaranteed since $c_n = 1$ for all $n$. A value of $\beta$ close
to $1$ implies a much smaller subclass of functions which then can be reconstructed uniformly with convergence rate arbitrarily close to
the sampling rate $1\slash n$. 
We finally point out that the result in Theorem
\ref{convset:mainthm} does not presume the number of change-points to
be estimated correctly. If $c_n$ additionally satisfies \eqref{def:lowerbound}
and if in Theorem \eqref{convset:mainthm} $q=q_n\ra\infty$ slower than $-\log
c_n$, we find from Theorem \ref{conscp:mainthm} that $\Prob(\hat K(q)= K) \ra 1$
and it follows from Theorem \eqref{convset:mainthm} that for $n$ large enough
\begin{equation*}
\Prob\left(\sup_{\hat\vt \in \mathcal{C}(q_n)}c_n^{-1}\abs{\tau_k-\hat\tau_k}>
1\right)\ra 0, \quad \text{for } k=1,\dots,K.
\end{equation*}
The solution set of the optimization problem \eqref{convset:optprob} 
constitutes a confidence set for the true regression function $\vartheta$.
Indeed, we find that
\begin{align}\label{convset:confbandestim}
\Prob\left(\vartheta \in \mathcal{C}(q) \right) & =
\Prob\left(T_n(Y,\vartheta;c_n)\leq q,\;K \leq \hat
K(q)\right) \\ 
& \nonumber \geq \Prob\left(T_n(Y,\vartheta;c_n)\leq q \right)-\Prob\left(\hat K(q) < K \right).
\end{align}
In particular, it follows from Theorem \ref{conscp:mainthm} that if $q_{1-\alpha}$
is the $1-\alpha$-quantile of $M$, the set $\mathcal{C}(q_{1-\alpha})$ is an
\emph{asymptotic confidence set} at level $1-\alpha$.

\begin{cor}\label{estimator:confbandcor}
Let $\alpha\in(0,1)$ and $q_{1-\alpha}$ to be the $1-\alpha$-quantile of the
statistic $M$ in \eqref{limit:limitstat}. Then,
\begin{equation}\label{convset:genbound}
\Prob\left( \vartheta \in \mathcal{C}(q_{1-\alpha}) \right) \geq
1- \alpha - 2 K e^{- C n\lambda
\Delta^2}\left[e^{\frac{1}{2}\left(\ha{q_{1-\alpha}}+\sqrt{2\log(2e/\lambda)}\right)^2} + e^{- \ha{3} C
n\lambda \Delta^2}\right] + \smallo(1)
\end{equation}
with $C=C(\mathcal{F}, \underline{\theta}, \overline{\theta})>0$ as in Theorem
\ref{conscp:mainthm}. Consequently one finds
\begin{equation*} 
\lim_{n\ra \infty} \Prob \left( \vt \in \mathcal{C}(q_{1-\alpha}) \right) \geq 1- \alpha.
\end{equation*}
for any $\vt\in \S$.
\end{cor}
We mention that for the Gaussian family (see Section \ref{subsec:gauss:opt}) the inequality \eqref{convset:genbound}
even holds for any $n$, i.e. the $\smallo(1)$ term on the r.h.s. can be omitted. Thus the r.h.s. of 
\eqref{convset:genbound} gives an explicit and nonasymptotic lower bound for the true confidence level of $C(q_\alpha)$.

In the following we use this result to determine classes of step functions on which confidence statements hold
uniformly. Being a subset of $\S$, the confidence set $\mathcal C(q)$ is hard to
visualize in practice. Therefore, in Section \ref{sec:alg:cb} we compute a confidence band $B(q)\subset [0,1]\times \Theta$ 
that contains the graphs of all functions in $\mathcal C(q)$ as well as disjoint
confidence intervals for the change-point locations denoted by $[\tau_k^l(q), \tau_k^r(q)]\subset [0,1]$ for $k=1,\ldots,\hat K(q)$.
For the sake of simplicity, we denote the collection $\{\hat K(q), B(q), \set{[\tau_k^l(q),\tau_k^r(q)]}_{k=1,\ldots,\hat K(q)}\}$ 
by $I(q)$ and agree upon the notation
\begin{align}\label{convset:maindef}
 \vt \prec I(q) \quad & \text{if } \hat K(q) = K, (t,\vt(t))\in B(q)\text{ and } \tau_k\in[\tau_k^l(q),\tau_k^r(q)] \text{ for }
 k=1,\ldots,K,\\ \nonumber \vt \nprec I(q) \quad	  & \text{otherwise}.
\end{align}
Put differently, $\vt \prec I(q)$ implies that \emph{simultaneously} the number
of change-points is estimated correctly, the change-points lie within the
confidence intervals and the graph is contained in the confidence band. As we
will show in Section \ref{sec:alg:cb}, the confidence set $\mathcal{C}(q)$ and $I(q)$ are linked by the following relation:
\begin{equation} \label{confreg:rel1}
 \vt \in \mathcal C(q) \Rightarrow \vt \prec I(q).
\end{equation}
Following the terminology in \citep{Li89}, $I(q)$ is called \emph{asymptotically
honest} for the class $\S$ at level $1-\alpha$ if 
\begin{equation*}
 \liminf_{n\ra  \infty} \inf_{\vt \in \S}\Prob \left( \vt \prec I(q) \right)\geq 1- \alpha. 
\end{equation*}
Such a condition obviously cannot be fulfilled over the entire class $\S$,
since signals cannot be detected if they vanish too fast as $n\ra\infty$.
For Gaussian observations this was made precise in Section
\ref{subsec:gauss:opt}.

To overcome this difficulty, we will relax the notion of asymptotic honesty.
 Let $\S^{(n)}\subset \S$, $n\in\N$ be a sequence of subclasses of $\S$.
Then $I(q)$ is \emph{sequentially honest w.r.t. $S^{(n)}$} at level $1-\alpha$
if 
\begin{equation*}
 \liminf_{n\ra  \infty} \inf_{\vt \in \S^{(n)}}\Prob \left( \vt \prec I(q)
 \right)\geq 1-\alpha.
\end{equation*}

By combining \eqref{convset:confbandestim}, \eqref{confreg:rel1} and Corollary
\ref{conscp:corunderest}
 we obtain the following result about the asymptotic honesty of $I(q_{1-\alpha})$.
\begin{cor}\label{confband:corhonest}
Let $\alpha\in(0,1)$ and $q_{1-\alpha}$ be the $1-\alpha$-quantile of the
statistic $M$ in \eqref{limit:limitstat} and  assume
that $(b_n)_{n\in\N}\ra\infty$ is a sequence of positive numbers. Define
$\mathcal{S}^{(n)} = \left\{ \vt \in \S: n \lambda \Delta^2 / \log(1/\lambda)
\geq b_n,\; \underline\theta\leq
\vt\leq\overline\theta \right\}$. Then $I(q_{1-\alpha})$ is \emph{sequentially honest
w.r.t.} $\S^{(n)}$ at level $1-\alpha$, i.e.
\begin{equation*}
 \lim_{n\ra \infty} \inf_{\vt \in \S^{(n)}}\Prob \left( \vt\prec I(q_\ha{{1-\alpha}}) \right)
 \geq 1-  \alpha.
\end{equation*}
\end{cor}
By estimating $1\slash \lambda \leq n$ we find that the confidence level
$\alpha$ is kept uniformly over nested models $\S^{(n)}\subset \S$, as long as
$\frac{n}{\log n}\Delta_n^2 \lambda_n\ra\infty$. Here $\lambda_n$ and $\Delta_n$ is the smallest interval length and smallest absolute
jump size in $\S^{(n)}$, respectively. 

\section{Implementation}\label{sec:alg}

We now explain how the SMUCE, i.e. the estimator $\hat\vt(q)$ with maximal
likelihood in the confidence set $\mathcal{C}(q)$, can be computed efficiently
within the dynamic programming framework. In general the proposed algorithm is
of complexity $\bigo(n^2)$. We will show, however, that in many situations the
computation can be performed much faster.

Our algorithm uses dynamic programming ideas from  \citep{FriKemLieWin08} in the context
of complexity penalized M-estimation.  See also \citep{Hoe08, DavHoeKra12} for a special case in our context. 
Moreover, we include  pruning steps as \citep{KillFeaEck12}, who also provide a survey on dynamic programming in 
change-point regression from a general point of view.
We will show that it is always possible to rewrite $\hat\vt(q)$ as a solution of
a minimization of a complexity penalized cost function with data dependent penalty.
To this end, we will denote the log-likelihood of $\hat\vt$ as
\begin{equation*}
l(Y,\hat \vt) = \sum_{i=1}^n \psi(\hat\vt(i\slash n)) - \hat\vt(i\slash n) Y_i.
\end{equation*}
Without restriction, we will assume that $l(Y,\hat\vt)\geq 0$ for all
$\hat\vt\in \S$. 

Following \citep{FriKemLieWin08}, we call a collection $\Pcal$ of discrete
intervals  a partition if its union equals the set $\set{1,\ldots,n}$. We
denote by $\Pfrak(n)$ the collection of all partitions of $\set{1,\ldots,n}$.
For $\Pcal\in\Pfrak(n)$ let $\#{\Pcal}$ the number of discrete
intervals in $\Pcal$. Hence, any discrete step function
$\vt\in\S_n[K]$  can be identified with a pair $(\Pcal,
\theta)$, where
\begin{equation*}
\Pcal\in \Pfrak(n),\quad{\# \Pcal} = K\quad\text{ and }\quad \theta =
(\theta_I)_{I\in \Pcal}\in \Theta^{{\# \Pcal}}, 
\end{equation*}
and \ha{$\vartheta(t) = \theta_I \Leftrightarrow
\lceil nt \rceil \in I$}. Next, we note that for a given
$\theta_I\in\Theta$ the negative log-likelihood on a discrete interval $I$ is 
given by $|I|(\psi(\theta_I)-\theta_I \overline Y_I)$. With this  we define
the \emph{costs} of $\theta_I$ on $I$ as
\begin{equation} \label{impl:costs}
d_{I}(Y,\theta_I) = \begin{cases}  
|I|(\psi(\theta_I)-\theta \overline Y_I) & \text{ if } \max_{[j,k]\subset I}
 \sqrt{2T_j^k (Y,\theta_I)} - \sqrt{2\log\frac{en}{k-j+1}} \leq q \\
\infty & \text{ else}.
\end{cases}
\end{equation}
The \emph{minimal costs} on the interval $I$ are then defined by
$d^*_{I} = \min_{\theta_I\in\Theta} d_{I}(Y,\theta_I)$ where we agree upon
$\theta_I^*\in\Theta$ being such that $d_I(Y,\theta_I^*)=d^*_I$. We stress that
$d^*_{I} = \infty$ if and only if no $\theta_I\in \Theta$ exists such
that the multiscale constraint is satisfied on $I$. Finally, for  an estimator
$(\Pcal, \theta)$ the overall costs are given by
\begin{equation*}
 D(\mathcal P, \theta)= \sum_{I \in \mathcal P}d_I(Y,\theta_I).
\end{equation*}

In \citep{FriKemLieWin08} a dynamic program is designed for computing
minimizers of
\begin{equation}\label{impl:penprob}
(\Pcal,\theta)\mapsto D(\Pcal,\theta) + \gamma({\# \Pcal} - 1),\quad \gamma>0.
\end{equation}

It is shown that the computation time amounts to $\bigo(n^2)$ given that the
minimal costs $d_I^*$ can be computed in $\bigo(1)$. We now show that each
minimizer of \eqref{impl:penprob} maximizes the  likelihood over the set
$\mathcal{C}(q)$,  if $\gamma > 0$ is chosen large enough. Note that this 
$\gamma$ can be computed explicitly for any given data $(Y_1,\ldots,Y_n)$
according to the next result.

\begin{lem}\label{impl:equivprob} 
Let $\gamma > 1/2 \left( n q + n \sqrt{2 \log(en)}\right)^2 + l(Y,m^{-1}(\bar Y)) $. Then, any solution of 
\eqref{impl:penprob} is also a solution of \eqref{intro:smre}.
\end{lem}

For completeness, we briefly outline the dynamic programming approach for the
minimization of \eqref{impl:penprob} as established in \citep{FriKemLieWin08}:  Define for
$r\leq n$ the Bellman function by $B(0) = -\gamma$ and
\begin{equation*}
 B(r)= \inf_{\mathcal P \in \mathfrak P(r), \theta \in \Theta^{|\mathcal P|}}
 D(\mathcal P, \theta)+ \gamma (|\mathcal P|-1)
\end{equation*}
and let $\Pcal_r\in \Pfrak(r)$ and $\theta_r \in \Theta^{|\Pcal_r|}$ be such
that $D(\Pcal_r, \theta_r) + \gamma(\bigl|\Pcal_r\bigr|
- 1) = B(r)$. Clearly, $B(n)$ is the minimal value of \eqref{impl:penprob} and
$(\Pcal_n, \theta_n)$ is a minimizer of \eqref{impl:penprob}. A key ingredient
 is the following recursion formula (cf. \citep[Lem. 1]{FriKemLieWin08})
\begin{equation*}
 B(p) = \inf_{1\leq r \leq p} B(r-1) + \gamma + d^*_{[r,p]}.
\end{equation*}
Let $p\leq n$ and assume that $(\Pcal_r, \theta_r)$ are given for all $r<p\leq
n$. Then, compute the best previous change-point position, i.e.
\begin{equation}\label{impl:bestprevpoint}
r_p = \argmin_{1\leq r\leq p} B(r-1) + d^*_{[r,p]}
\end{equation}
and set $\Pcal_p = \Pcal_{r_p-1} \cup \set{[r_p,p]}$ and $\theta_p =
(\theta_{r_p-1}, \theta_{[r_p,p]})$.  With this we can iteratively compute the Bellman function
$B(p)$ and the corresponding minimizers $(\mathcal P_p, \theta_p)$ for
$p=1,\ldots,n$ and eventually obtain $(\mathcal P_n,\theta_n)$, i.e. a
minimizer of \eqref{impl:penprob}. According to Lemma \ref{impl:equivprob}, this
$(\mathcal P_n,\theta_n)$ solves \eqref{intro:smre} if $\gamma$ is chosen large
enough. 

\ha{
We note that for a practical implementation  of the proposed dynamic
program, the efficient computation of the values $d^*_{[r,p]}$ is essential. We postpone this to the upcoming subsection and will discuss the complexity of the algorithm first. Following \citep{FriKemLieWin08} the dynamic programming algorithm is of order $O(n^2)$, given that the minimal costs $d^*_{[i,j]}$ are computed in $O(1)$ steps. Note, that this does not hold true for the costs in \eqref{impl:costs}. However, as we will show in the next subsection, the set of all optimal costs $(d_{[i,j]}^*)_{1 \leq i \leq j \leq n}$ can be computed in $O(n^2)$ steps and hence the complete algorithm is of order $O(n^2)$ again.}

\ha{
In our implementation the specific structure of the costs (see \eqref{impl:costs}) has been employed by including several pruning steps into the dynamic program, similar to \citep{KillFeaEck12}. Since the details are rather technical, we only give a brief explanation why the computation time of the algorithm as described below can be reduced: the speed ups are based on the idea to consider only such $r$ in \eqref{impl:bestprevpoint} that may lead to a minimal value, i.e. those $r$ that are strictly larger than $\max\set{r~:~d^*_{[r,p]} = \infty}$.
The number of intervals, on which the SMUCE is constant, is of order $n^2 \sum_{k=1}^{\hat K+1}(\hat\tau_k-\hat\tau_{k-1})^2$, instead of $n^2$ if all intervals were considered. The number of intervals $[r,p]$ which are needed in \eqref{impl:bestprevpoint} is essentially of the same order.
This indicates that SMUCE is much faster for signals with many detected change-points than for signals with few detected change-points, which has been confirmed by simulations.}

\ha{The pruned algorithm is implemented for the statistical software R in the package \emph{stepR}\footnote{R package available at 
\url{http://www.stochastik.math.uni-goettingen.de/smuce}}. The SMUCE procedure for several exponential families 
is available via the function $\emph{smuceR}$.}

\subsection{Computation of minimal costs} \label{subsec:mincosts}
 
Let $r\leq i\leq j\leq p$.  Since $\set{F_\theta}_{\theta\in\Theta}$ was
assumed to be a regular, \ha{one} dimensional exponential family, the natural
parameter space $\Theta$ is a nonempty, open interval $(\theta_1,\theta_2)$ with
$-\infty\leq \theta_1 < \theta_2\leq\infty$.  Moreover, the mapping
$\theta\mapsto J(\overline Y_i^j,\theta)$ is strictly convex on $\Theta$ and has
the unique global minimum at $m^{-1}(\overline Y_i^j)$  if and only if $m^{-1}(\overline
Y_i^j) \in \text{int}(\Theta)$. In this case it follows from \citep[Thm. 6.2]{Bar73} that
for all $q>0$ 
\begin{multline*}
 \set{\theta\in\Theta: \sqrt{2T_i^j(Y,\theta)}-\sqrt{2\log \frac{en}{j-i+1}} 
 \leq q } \\ = \set{\theta\in\Theta: J(\overline Y_i^j,\theta)  \leq
 \frac{\left(q\ha{+}\sqrt{2\log \frac{en}{j-i+1}}\right)^2}{2(j-i+1)}} = 
 [\underline{b}_{ij},\overline{b}_{ij}],
\end{multline*}
with $-\infty < \underline{b}_{ij} \leq m^{-1}(\overline Y_i^j) \leq
\overline{b}_{ij}< \infty$. In other words, $\underline{b}_{ij}$ and
$\overline{b}_{ij}$ are the two finite solutions of the equation
\begin{equation}\label{alg:eq:bounds}
 J(\overline Y_i^j,\theta)=\frac{\left(q\ha{+}\sqrt{2\log
 \frac{en}{j-i+1}}\right)^2}{2(j-i+1)}.
\end{equation}

If $m^{-1}(\overline Y_i^j) \not\in \text{int}(\Theta)$, then \citep[Thm.
6.2]{Bar73} implies that either $\underline{b}_{ij} = -\infty$ or $\overline{b}_{ij} = \infty$. Let
us assume without restriction that $\underline{b}_{ij} = -\infty$ which in turn shows that
$\Theta= (-\infty, \theta_2)$ and $m^{-1}(\overline Y_i^j) = -\infty$. In this case, the infimum of
$\theta\mapsto J(\overline Y_i^j,\theta)$ is not attained and
\eqref{alg:eq:bounds} has only one finite solution $\overline{b}_{ij}$. The
lower bound $\underline{b}_{ij}=-\infty$ then is trivial.

After computing $\underline{b}_{ij}$ and $\overline{b}_{ij}$ for all $r\leq i \leq j
\leq p$, define $\underline{B}_{rp} = \max_{r\leq i\leq j\leq p}
\underline{b}_{ij}$ and $\overline{B}_{rp} = \min_{r\leq i\leq j\leq p}
\overline{b}_{ij}$. Hence, if $d^*_{[r,p]}< \infty$ we obtain
\begin{equation*}
 \theta^*_{[r,p]}=\argmin_{\theta\in [\underline{B}_{rp}, \overline{B}_{rp}]}
 d_{[r,p]}(Y,\theta) =
 \begin{cases}
 \overline{B}_{rp}  &\text{if} \quad m^{-1}(\overline Y_r^p) \geq
 \overline{B}_{rp} \\
 \underline{B}_{rp} & \text{if} \quad m^{-1}(\overline Y_r^p)\leq \underline{B}_{rp} \\
  m^{-1}(\overline Y_r^p) &\text{otherwise}. \quad
\end{cases}
\end{equation*}
Moreover, $d^*_{[r,p]}=\infty$ if and only if $\underline{B}_{rp}  >
\overline{B}_{rp}$.

To summarize, the computation of $\theta^*_{[r,p]}$ (and hence the
computation of the minimal costs $d^*_{[r,p]}$) reduces to finding the
non-trivial solutions of \eqref{alg:eq:bounds} for all $r\leq i\leq j\leq p$.
This can either be done explicitly (as for the Gaussian family, for example) or
approximately by Newton's method, say.

\ha{ Finally, we obtain that given the $\underline{b}_{ij}$
and $\overline{b}_{ij}$ are computed in $O(1)$, the bounds $(\underline B_{rp})_{1 \leq r \leq p \leq n}$
and $(\overline B_{rp})_{1 \leq r \leq p \leq n}$
are computed in $O(n^2)$. This follows from the observation that for $1\leq r \leq p \leq n$
\begin{equation*}
 \underline B_{rp} = \max \left\{\underline B_{(r+1)p},\underline B_{r(p-1)}, \underline{b}_{rp} \right \} \; \text{and} \;
 \overline B_{rp} = \min \left\{\overline B_{(r+1)p},\overline B_{r(p-1)}, \overline{b}_{rp} \right \},
\end{equation*}
which allows for iterative computation.
}

\subsection{Computation of confidence sets}\label{sec:alg:cb}

The dynamic programing algorithm gives, in addition to the computation of the
SMUCE, an approximation to the solution set $C(q)$ of \eqref{convset:optprob} as
discussed in Section \ref{sec:convset}. The algorithm outputs
disjoint intervals $[\tau_k^l,\tau_k^r]$  as well as a confidence band
$B(q)\subset [0,1]\times \Theta$ such that for each estimator $\hat\vt\in C(q)$:
\begin{equation*}
\hat \tau_k \in [\tau_k^l,\tau_k^r]\text{ for }k=1,\dots,\hat K(q)\quad\text{ and
}\quad (t,\hat\vt(t))\in B(q),\text{ for all }t\in[0,1].
\end{equation*}
To make this clear let $1\leq k\leq \hat K(q)$ and define 
\begin{equation}\label{impl:jumpband}
R_k = \max\set{r:|\mathcal{P}_r| \leq k} \quad \text{ and }\quad L_k = \min
\set{ p: d^*_{[p,R_k]} < \infty }.
\end{equation}
Then, for any estimator $\hat\vt\in \S_n[\hat K(q)]$ that satisfies
$T_n(Y,\hat\vt)\leq q$, it holds that $\hat \tau_k  \in [\tau_k^l,\tau_k^r]$
with $\tau_k^l = n^{-1}L_k$ and $\tau_k^r= n^{-1}R_k$.

Now we construct a confidence band $B(q)$ that contains the graphs of all
functions in $C(q)$. To this end, let $\hat\vt$ be as above and note
that for $1\leq k \leq \hat K(q)$ there is exactly one change-point in the
interval $[\tau_k^l,\tau_k^r]$ and no change-point in $(\tau_k^r,\tau_{k+1}^l)$.
First, assume that $t\in (\tau_k^r,\tau_{k+1}^l)$. Then we get a lower and
an upper bound for $\hat \vt(t)$ by $\underline{B}_{R_k+1 L_{k+1}-1}$ and
$\overline{B}_{R_k+1 L_{k+1}-1}$, respectively. Now let $t\in [\tau_k^l,\tau_k^r]$. Then,
the $k$-th change-point is either to the the left or to the right of $t$ and
hence any feasible estimator is constant either on $[\tau_k^l,t]$ or on
$[t ,\tau_k^r]$. Thus, we obtain a lower bound by
$\min\set{\underline{B}_{L_k,\lfloor tn \rfloor}, \underline{B}_{\lceil nt \rceil,R_k}}$ and an upper bound by $\max\set{\overline{B}_{L_k,\lfloor tn \rfloor},
\overline{B}_{\lceil nt \rceil,R_k}}$.

\section{On the choice of the threshold parameter}\label{sim:thresh}

The choice of the parameter $q$ in \eqref{intro:optprob} is crucial for it balances data fit and
parsimony of the estimator. First we discuss a general recipe that takes into
account prior information on the true signal $\vt$. Based on
this a specific choice is given in the second part which we found particularly suitable for our purposes. Further
generalizations are discussed briefly.

As shown in Corollary \ref{estimator:confbandcor} for the general case, $q$
determines asymptotically the level of significance for the confidence sets
$\mathcal{C}(q)$. For the Gaussian case we have shown in Section \ref{subsec:gauss:opt} that this result is even non-asymptotic, i.e.
from Corollary \ref{gauss:opt:cor1} it follows that
\begin{equation} \label{sim:q:over}
 \Prob (\hat K(q) > K) \leq \alpha(q), 
\end{equation}
where $\alpha(q)$ is defined as $\alpha(q)=\Prob (M\geq q)$. This allows to control the probability
of overestimating the number of change-points. If the latter is considered as a
measure of smoothness, \eqref{sim:q:over} can be interpreted as a \emph{minimal
smoothness guarantee}. This is similar in spirit to results on other multiscale
regularization methods (see \citep{Don95, FriMarMun12}). As argued in Section \ref{sec:convset} in general it is not possible to bound the minimal
number of change-points without further assumptions on the true function $\vt$ (see also  \citep{Don88} in the
context of mode estimation for densities).
However, we can draw a sharp bound for the probability of underestimating the number of change-points from
\eqref{gauss:q:under} in terms of the minimal interval length
$\lambda$ and minimal feature size $\eta^2 = n\lambda \Delta^2$, which gives
\begin{equation*}
 \Prob \left( \hat K (q) < K \right) \leq \frac{2}{\lambda} \left[\exp\left( - \frac{1}{8}
 \left( \frac{\eta}{2\sqrt{2}} - 
 q-\sqrt{2\log\frac{2e}{\lambda}}
 \right)_+^2 \right) +  \exp\left(- \frac{\eta^2}{16}
 \right)\right]=:\beta(q,\eta,\lambda),
\end{equation*}
where we have exploited the fact that $K\leq 1\slash \lambda$. By combining
\eqref{sim:q:over} with the bound above one finds
\begin{equation}\label{sim:q:eq}
 \Prob \left( \hat K (q) = K \right)\geq 1-\alpha(q)-\beta(q,\eta,\lambda).
\end{equation}
In order to optimize the bound on the probability of estimating the correct
number of change-points, one has to balance the error of over- and underestimation.
Therefore, we aim for maximizing the r.h.s. over $q$.
Given $\lambda$ and $\eta^2 = n\lambda \Delta^2$ we therefore suggest to choose $q$
as
\begin{equation}\label{sim:q:qopt1}
 q^*_{\lambda,\eta}= \max_{q>0} \left\{ 1-\alpha(q)-\beta(q,\eta,\lambda)
 \right\}.
\end{equation}
The explicit knowledge of the influence of $\lambda$ and $\eta$ in
\eqref{sim:q:qopt1} paves the way to various strategies for incorporating
prior information in order to determine $q$. One might
e.g. use a full prior distribution on $(\lambda, \eta)$ and minimize the posterior model
selection error 
\begin{equation*}
 \max_{q>0} \E{1-\alpha(q)-\beta(q,\eta,\lambda)}.
\end{equation*}
In the following we suggest a rather simple
way to proceed, which we found empirically to perform quite well.
We stress that there is certainly room for further improvement. 
Motivated by the results of Section \ref{subsec:gauss:opt} we suggest to define 
$\lambda$ and $\eta = \sqrt{n\lambda}\Delta$ in dependence of $n$ implicitly by
the following assumptions

\begin{enumerate}[(i)]
  \item $\eta^*=12 \sqrt{-\log(\ha{\lambda^*})}$ and  \label{sim:q:genass1}
\item $\sqrt{\lambda^*}=g(\Delta, n)$,    \label{sim:q:genass2}
\end{enumerate}
for some function $g$ with values in $(0,1]$.
According to Theorem
\ref{gauss:opt:nojumpmissing}, the first assumption reflects the worst case
scenario among all signals that can be recovered with probability $1$
asymptotically. The second assumption corresponds to a prior belief in the true
function $\vt$. In the following simulations we always choose $g(\Delta,n)=\Delta$ which puts the decay of
$\lambda$ and $\Delta$ on equal footing.
We then come back to the approach in \eqref{sim:q:qopt1} and define
\begin{equation}\label{sim:thresh:optq}
  q^{*}_{n}= \max_{q>0} \left\{ 1-\alpha(q)-\beta(q,\eta^*,\lambda^*) \right\}
\end{equation}
where $\lambda^*$ and $\eta^*$ are defined by \eqref{sim:q:genass1} and
\eqref{sim:q:genass2}. Consequently, the maximizing element $q_n^*$ picks that $q$
which maximizes the probability bound in \eqref{sim:q:eq} of correctly estimating the number of change-points.
Note, that $q^*_n$ does not depend on the true signal $\vt$ but only on the number of observations $n$.

Even though the motivation for $q^*_n$ is build on the assumption of Gaussian
observations,
 simulations indicate that it performs also well for other distributions.
 That is why we choose $q=q_n^{*}$, unless stated differently throughout all
 simulations. \ha{ There $\alpha(q)$ is estimated by Monte-Carlo simulations with sample size $n=3000$. These simulations are rather expensive
 but only need to be performed once. For a given $n$, a solution of \eqref{sim:thresh:optq} may then be approximated numerically by computing the r.h.s.
 for a range of values for $q$.
 }
 We stress again that the general concept given by
 \eqref{sim:q:qopt1} can be employed further to incorporate prior knowledge of the signal as
 will be shown in Section \ref{sec:realdata}.

\section{Simulations}\label{sim}
As mentioned in the introduction, the literature on the change-point problem is
vast and we will now aim for comparing our approach within the plethora of established methods for exponential families.
All SMUCE instances computed in this section are based on the optimization problem \eqref{intro:optprob}, i.e.
we do not restrict the interval lengths, as required in Section \ref{sec:mrstat}
for technical reasons.

\subsection{Gaussian mean regression} \label{subsec:gaussmean}
Recall model \eqref{gauss:opt:model} in Section \ref{subsec:gauss:opt} with constant variance $\sigma^2$ and piecewise constant means $\mu$,
i.e. we set $\theta = \mu\slash \sigma^2$ and $\psi(\theta) = \mu^2\slash (2\sigma^2)$ in \eqref{model:expfam}.
Throughout the following we assume the variance $\sigma^2$ to be known, otherwise one may estimate it by standard methods, see e.g. \citep{DavKov01} or
\citep{DetMunWag98}.

Then, the MR-statistic \eqref{smre:mrstat} evaluated at $\hat\mu \in \S_n[\hat K]$ reads as
\begin{equation*}
T_n(Y,\hat \mu) = \max_{0\leq k\leq \hat K}\max_{\hat l_k<i\leq j\leq \hat
l_{k+1}} \left( \frac{\abs{\sum_{l=i}^j Y_l - \hat \mu_k}}{\sigma\sqrt{j-i+1}} -
\sqrt{2  \log\frac{e n}{j-i+1}} \right).
\end{equation*}
After selecting the model $\hat K(q)$ according to
\eqref{conscp:estnocp}, the SMUCE becomes
\begin{equation*}
\hat \mu(q) = \argmin_{\hat\mu\in \S_n[\hat
K(q)]}  \sum_{k=0}^{\hat K(q)}(\hat
l_{k+1} - \hat l_k)(\overline Y_{\hat l_k}^{\hat l_{k+1}} - \hat
\mu_k)^2\quad\text{ s.t. }\quad T_n(Y,\hat\mu)\leq q.
\end{equation*} 
In our simulation study we consider the following change-point-methods.
A large group follows the common paradigm of maximizing a penalized likelihood
criterion of the form 
\begin{equation}\label{sim:penlike}
\vt\mapsto l(Y,\vartheta) - \text{pen}(\vt)
\end{equation}
over $\vt \in \S_n[k]$ for $k=1,\ldots,n$, where
the function $\text{pen}(\vt)$ penalizes the complexity of the
model. This includes the \emph{Bayes
Information Criterion (BIC)} introduced in \citep{Sch78} which suggests the choice
$ \text{pen}(\vt) = \# J(\vt) \slash 2\log n$. As it was for instance
stressed in \citep{ZhaSie07}, the formal requirements to apply the BIC are not satisfied 
for the change-point problem.
Instead the authors propose the following penalty function
in \eqref{sim:penlike}, denoted as modified BIC:
\begin{equation*} 
\text{pen}(\vt) = -\frac{1}{2}\left(3\#J(\vt)\log n + \sum_{k=1}^{\#J(\vt)+1}
\log(\tau_k-\tau_{k-1}) \right).
\end{equation*}
They compare their mBIC method with the traditional BIC
as well as with the methods in \citep{OlsVenLucWig04} and
\citep{FriSniAntPinAlb04} by means of a comprehensive simulation study and
demonstrated the superiority of their method w.r.t. the number of correctly
estimated change-points. For this
reason we only consider \citep{ZhaSie07} in our simulations.
In addition, we will include the \emph{penalized likelihood oracle (PLoracle)}
as a benchmark, which is defined as follows: \kl{Recall that $K$ denotes the
true number of change-points. For given data $Y$, define $\omega_l$ and
$\omega_u$ as the minimal and maximal element of the set
\begin{equation*}
\set{\omega\in\R~:~ \argmax_{\hat\vt\in\mathcal{S}_n} \left(l(Y,\hat\vt) - \omega
\# J(\hat\vt)\right) \text{ has $K$ change-points}},
\end{equation*}
respectively. In particular, for $\omega_m := (\omega_l
+ \omega_u)\slash 2$ the penalized maximum likelihood estimator, i.e. a
maximizer of \eqref{sim:penlike} obtained with penalty $\text{pen}(\vt)
= \omega_m \# J(\vt)$, has exactly $K$ change-points. For our assessment, we
simulate $10^4$ instances of data $Y$ and compute the median $\omega^*$ of the
corresponding $\omega_m$'s. We then define the PLoracle to be a maximizer of
\eqref{sim:penlike} with $\text{pen}(\vt) = \omega^*\#J(\vt)$}.
Of course, PLoracles are not accessible in practice (since $K$ and $\vt$ are unknown).
However, they represent benchmark instances within the class of estimators given by
\eqref{sim:penlike} \ha{and penalties of the form $\text{pen}(\vt)= \omega\#J(\vt)$}.
\ha{We stress again, that even if SMUCE and the PLoracle have the same number of change-points they are in general not equal, since the likelihood in \eqref{intro:smre} is maximized only over the set $C(q)$.}

Moreover, we consider the \emph{fused lasso} algorithm which is based on
computing solutions of
\begin{equation}\label{sim:gaussmean:fs}
 \min_{\hat\vt \in S} \sum_{i=1}^{n} (Y_i-\hat\vt(i/n))^2 + \lambda_1
 \bigl\|\hat\vt\bigr\|_1 + \lambda_2 \bigl\|\hat\vt\bigr\|_{\text{TV}},
\end{equation}
where $\norm{\cdot}_1$ denotes the $l_1$-norm and $\norm{\cdot}_\text{TV}$ the
total variation semi-norm (see also \citep{HarLev10}). The fused lasso is not specifically designed for
the change-point problem. However, due to its prominent role
and its application to change-point problems (see e.g. \citep{Tib08}), we include it into our
simulations. An optimal choice of the parameters $(\lambda_1,\lambda_2)$ is crucial and
 in our simulations we consider two \emph{fused lasso oracles} $\text{FL}^\text{MSE}$ and $\text{FL}^\text{c-p}$. In $500$ Monte Carlo simulations (using the true signal) we compute 
$\lambda_1$ and $\lambda_2$ such that the MISE is minimized for the $\text{FL}^\text{MSE}$ and such that the frequency of correctly estimated number of change-points is maximized for $\text{FL}^\text{c-p}$.

In summary, we compare \text{SMUCE} with the modified BIC approach suggested in
\citep{ZhaSie07}, the CBS algorithm\footnote{R package available at
\url{http://cran.r-project.org/web/packages/PSCBS}} proposed in
\citep{OlsVenLucWig04}, the fused lasso algorithm\footnote{R package
available at \url{http://cran.r-project.org/web/packages/flsa/}} suggested in
\citep{TibSauRosZhuKni04}, \ha{unbalanced haar wavelets\footnote{R package available at
\url{http://cran.r-project.org/web/packages/unbalhaar/}} \citep{Fry07}} and the PLoracle as defined above. Since
the CBS algorithm tends to overestimate the number of change-points the authors
included a pruning step which requires the choice of an additional parameter.
The choice of the parameter is not explicitly described in \citep{OlsVenLucWig04} and here we only
consider the unpruned algorithm.

\begin{figure}[h!]
\begin{center}
 \includegraphics[width=\columnwidth]{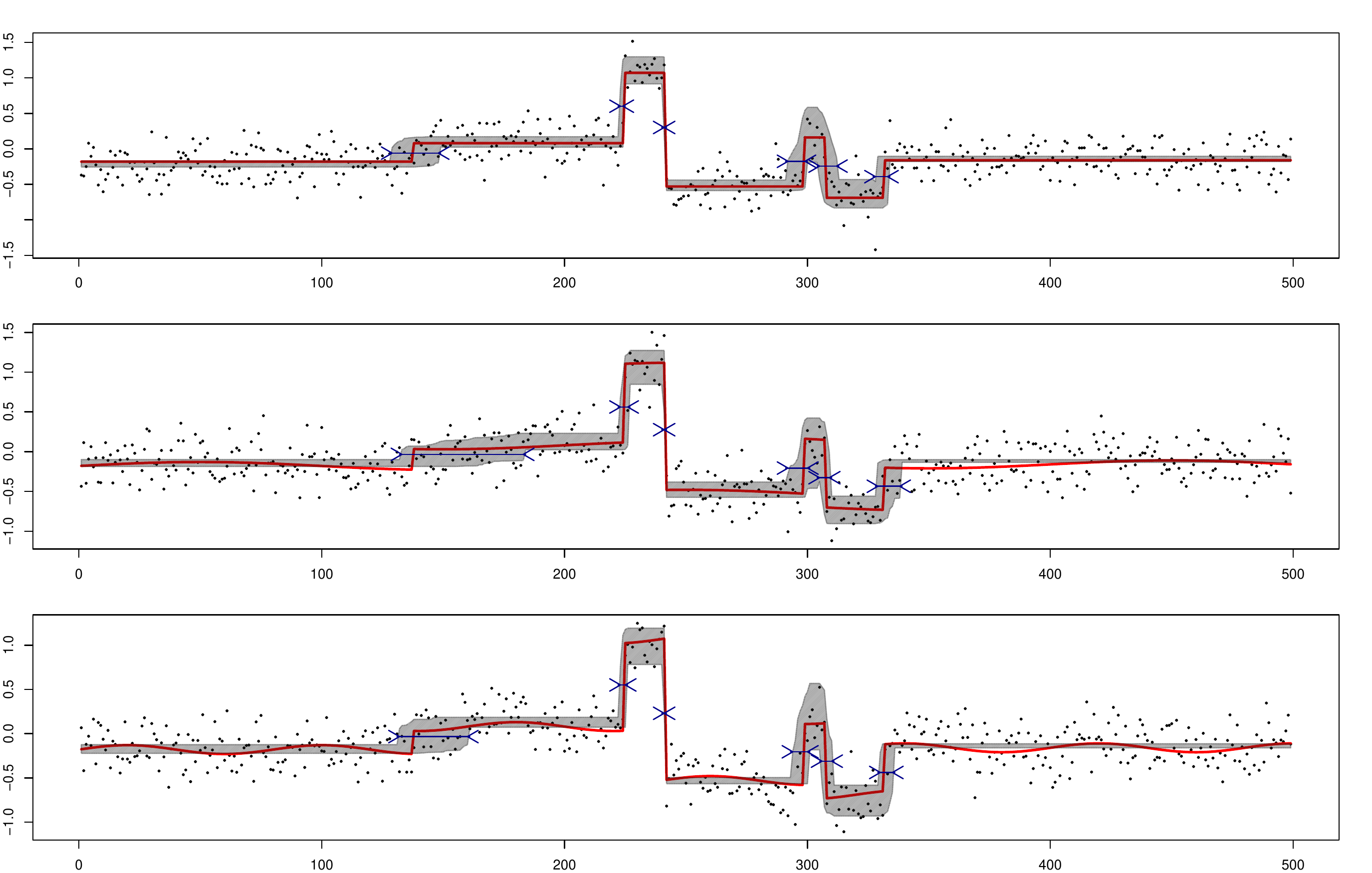}
 \caption{True signal (solid line), simulated data (dots) and
confidence bands (grey hatched) and confidence intervals for the change-points
(inwards pointing arrows) for $a=0$ (left), $a=0.01$ (middle) and $a=0.025$ (right) and $\sigma^2=0.2$}
 \label{sim:gaussmean:sim}
\end{center}
\end{figure}

We follow the simulation setup considered in \citep{ZhaSie07,OlsVenLucWig04}.
The application they bear in mind is the analysis of array-based comparative
genomic hybridization (array-CGH) data. Array-CGH is a technique for recording
the number of copies of genomic DNA (cf. \citep{KalKalSudRut92}). As pointed out
in \citep{OlsVenLucWig04}, piecewise constant regression is a natural model for
array DNA copy number data (see also Section \ref{sec:realdata:gauss}). 
Here, one has $n=497$ observations with constant variance $\sigma^2=0.04$ and the true regression function has $6$ change-points at locations $\tau_i = l_i\slash n$ and $(l_1,\ldots,l_6) =
(138,225,242,299,308,332)$ with intensities $(\theta_0,\ldots,\theta_6) =
(-0.18, 0.08, 1.07, -0.53, 0.16, -0.69,$ $ -0.16$). In order to investigate
robustness against small deviations from the model with step
functions, a local trend component is included in these simulations, i.e.
\ha{\begin{equation}\label{sim:modeltrend}
 Y_i\sim \mathcal{N}(\vartheta(i\slash n) + 0.25 b \sin(a \pi i),
\sigma^2),\quad i=1,\ldots,n.
\end{equation} 
}

\begin{table}[ht]
\tiny
\centering
\begin{tabular}{lll||cc>{\columncolor{tablebg}}ccc|cc}
 & trend & $\sigma$ & $\leq 4$ & 5 & 6 & 7 & $\geq 8$ & MSE & MAE \\ 
  \hline
  SMUCE ($1-\alpha=0.55$)& no & 0.1 & 0.000 & 0.000 & 0.988 & 0.012 & 0.000 & \textbf{0.00019} & 0.00885 \\ 
  PLoracle & no & 0.1 & 0.000 & 0.000 & \textbf{1.000} & 0.000 & 0.000 & \textbf{0.00019} & \textbf{0.00874} \\ 
  mBIC \citep{ZhaSie07} & no & 0.1 & 0.000 & 0.000 & 0.964 & 0.031 & 0.005 & 0.00020 & 0.00888 \\ 
  CBS \citep{OlsVenLucWig04} & no & 0.1 & 0.000 & 0.000 & 0.922 & 0.044 & 0.034 & 0.00023 & 0.00903 \\ 
  unbalhaar \citep{Fry07} & no & 0.1 & 0.000 & 0.000 & 0.751 & 0.137 & 0.112 & 0.00026 & 0.00926 \\ 
  $\text{FL}^\text{c-p}$ & no & 0.1 & 0.124 & 0.122 & 0.419 & 0.134 & 0.201 & 0.00928 & 0.15821 \\ 
  $\text{FL}^\text{MSE}$ & no & 0.1 & 0.000 & 0.000 & 0.000 & 0.000 & 1.000 & 0.00042 & 0.00274 \\  \hline
  SMUCE ($1-\alpha=0.55$)& no & 0.2 & 0.000 & 0.000 & \textbf{0.986} & 0.014 & 0.000 & \textbf{0.00117} & \textbf{0.01887} \\ 
  PLoracle & no & 0.2 & 0.024 & 0.001 & 0.975 & 0.000 & 0.000 & 0.00138 & 0.01915 \\ 
  mBIC \citep{ZhaSie07} & no & 0.2 & 0.000 & 0.000 & 0.960 & 0.037 & 0.003 & 0.00120 & 0.01894 \\ 
  CBS \citep{OlsVenLucWig04} & no & 0.2 & 0.000 & 0.000 & 0.870 & 0.089 & 0.041 & 0.00146 & 0.01969 \\
  unbalhaar \citep{Fry07} & no & 0.2 & 0.000 & 0.000 & 0.637 & 0.222 & 0.141 & 0.00174 & 0.02063 \\
  $\text{FL}^\text{c-p}$ & no & 0.2 & 0.184 & 0.162 & 0.219 & 0.174 & 0.261 & 0.08932 & 0.23644 \\ 
  $\text{FL}^\text{MSE}$ & no & 0.2 & 0.000 & 0.000 & 0.000 & 0.000 & 1.000 & 0.00297 & 0.03692 \\  \hline
  SMUCE ($1-\alpha=0.55$)& long & 0.2 & 0.000 & 0.000 & 0.825 & 0.171 & 0.004 & \textbf{0.00209} & \textbf{0.03314} \\ 
  PLoracle & long & 0.2  & 0.026 & 0.030 & \textbf{0.944} & 0.000 & 0.000 & 0.00245 & 0.03452 \\ 
  mBIC \citep{ZhaSie07} & long & 0.2  & 0.000 & 0.000 & 0.753 & 0.215 & 0.032 & 0.00214 & 0.03347 \\ 
  CBS \citep{OlsVenLucWig04} & long & 0.2  & 0.000 & 0.000 & 0.708 & 0.130 & 0.162 & 0.00266 & 0.03501 \\ 
  unbalhaar \citep{Fry07} & long & 0.2  & 0.000 & 0.000 & 0.447 & 0.308 & 0.245 & 0.00279 & 0.03515 \\  
  $\text{FL}^\text{c-p}$ & long & 0.2  & 0.078 & 0.112 & 0.219 & 0.215 & 0.376 & 0.08389 & 0.22319 \\ 
  $\text{FL}^\text{MSE}$ & long & 0.2  & 0.000 & 0.000 & 0.000 & 0.000 & 1.000 & 0.00302 & 0.03782 \\ \hline
  SMUCE ($1-\alpha=0.55$)& short & 0.2  & 0.000 & 0.002 & \textbf{0.903} & 0.088 & 0.007 & \textbf{0.00235} & \textbf{0.03683} \\ 
  PLoracle & short & 0.2   & 0.121 & 0.002 & 0.877 & 0.000 & 0.000 & 0.00325 & 0.03846 \\ 
  mBIC \citep{ZhaSie07}& short & 0.2   & 0.000 & 0.000 & 0.878 & 0.107 & 0.015 & 0.00238 & 0.03695 \\ 
  CBS \citep{OlsVenLucWig04} & short & 0.2  & 0.000 & 0.000 & 0.675 & 0.182 & 0.143 & 0.00267 & 0.03806 \\ 
  unbalhaar \citep{Fry07}& short & 0.2   & 0.000 & 0.000 & 0.602 & 0.225 & 0.173 & 0.00288 & 0.03849 \\ 
  $\text{FL}^\text{c-p}$ & short & 0.2  & 0.175 & 0.126 & 0.192 & 0.210 & 0.297 & 0.08765 & 0.23105 \\ 
  $\text{FL}^\text{MSE}$ & short & 0.2  & 0.000 & 0.000 & 0.000 & 0.000 & 1.000 & 0.00331 & 0.04111 \\  \hline
  SMUCE ($1-\alpha=0.55$)& no & 0.3   & 0.030 & 0.340 & 0.623 & 0.007 & 0.000 & 0.00660 & 0.03829 \\ 
  PLoracle & no & 0.3    & 0.181 & 0.031 & 0.788 & 0.000 & 0.000 & 0.00505 & 0.03447 \\ 
  mBIC \citep{ZhaSie07}  & no & 0.3 & 0.015 & 0.006 & \textbf{0.927} & 0.050 & 0.002 & \textbf{0.00364} & \textbf{0.03123} \\
  CBS \citep{OlsVenLucWig04} & no &0.3 & 0.006 & 0.019 & 0.764 & 0.157 & 0.054 & 0.00449 & 0.03404 \\ 
  unbalhaar \citep{Fry07}  & no & 0.3   & 0.008 & 0.004 & 0.602 & 0.244 & 0.142 & 0.00556 & 0.03792 \\
  $\text{FL}^\text{c-p}$  & no & 0.3   & 0.038 & 0.059 & 0.088 & 0.115 & 0.700 & 0.08792 & 0.23496 \\ 
  $\text{FL}^\text{MSE}$  & no & 0.3   & 0.531 & 0.200 & 0.125 & 0.078 & 0.066 & 0.09670 & 0.24131 \\ \hline
  SMUCE ($1-\alpha=0.4$) & no & 0.3 & 0.000 & 0.099 & 0.798 & 0.089 & 0.000 & 0.00468& 0.03499\\ \hline
\end{tabular}
\caption{\ha{Frequencies of estimated number of change-points and MISE by model
selection for SMUCE, PLoracle, mBIC
\citep{ZhaSie07}, CBS \citep{OlsVenLucWig04}, the fused lasso oracles $\text{FL}^\text{c-p}$ and
$\text{FL}^\text{MSE}$ as well as the unbalanced haar wavelets estimator \citep{Fry07}. The true
signals, shown in Figure \ref{sim:gaussmean:sim}, have $6$ change-points.}}
\label{sim:gau:tablecomp}
\end{table}

\ha{Following \citep{ZhaSie07} we simulate data for $\sigma=0.2$ and $a=0$ (no trend), $a=0.01$ (long trend) and  $a=0.025$
(short trend) (see Figure \ref{sim:gaussmean:sim}). Moreover, we included a scenario with a smaller signal-to-noise ratio, i.e. $\sigma=0.3$ and $a=0$ and one with a higher signal-to-noise ratio, i.e. $\sigma=0.3$ and $a=0$. For both scenarios we do not display results with a local trend, since we found the effect to be very similar to the results with $\sigma=0.2$.}

\ha{
Table \ref{sim:gau:tablecomp} shows the frequencies of the number of detected
change-points for all mentioned methods and the corresponding MISE and MIAE. 
Moreover, in Figure \ref{sim:gaussmean:typex} we displayed typical observation of model
\eqref{sim:modeltrend} with $a=0.1$ and $b=0.1$ and the aforementioned
estimators.
The results show that the \text{SMUCE} outperforms the mBIC \citep{ZhaSie07} slightly for $\sigma=0.2$ and appears to be less vulnerable for trends, in particular. Notably, \text{SMUCE} often performs even better than the PLoracle. 
For $\sigma=0.3$ SMUCE has a tendency to underestimate the number of change-points by one, while CBS and in particular mBIC estimates the true number $K=6$ with high probability correctly. As it is illustrated in Figure \ref{sim:gaussmean:typex2} this is due to the fact that SMUCE can not detect all change-points at level $1-\alpha \approx 0.55$ as we have chosen it following the simple rule \eqref{sim:thresh:optq} in Section $4$. For further investigation, we lowered the level to $1-\alpha = 0.4$ (see last row in Table \ref{sim:gau:tablecomp}). Even though this improves estimation, SMUCE performs comparably to CBS and the PLoracle now, it is still worse than mBIC.
}

\ha{
For an evaluation of $\text{FL}^\text{MSE}$ and $\text{FL}^\text{c-p}$ one should account for the 
quite different nature of the fused lasso: 
The weight $\lambda_1$ in \eqref{sim:gaussmean:fs} penalizes
estimators with large absolute values,  while $\lambda_2$ penalizes the
cumulated jump height. However, none of them encourages directly sparsity w.r.t the number of change-points.
That is why these estimators often incorporate many
small jumps (well known as the \emph{staircase effect}).
In comparison to SMUCE one finds that SMUCE outperforms the $\text{FL}^\text{MSE}$ w.r.t the MISE
and it outperforms $\text{FL}^\text{c-p}$ w.r.t. the frequency of correctly estimating the number of change-points.
 The example in Figure \ref{sim:gaussmean:typex} suggests
that the major features of the true signal are recovered by $\text{FL}^\text{MSE}$.
But additionally, there are also some artificial features in the estimator which suggest that an additional filtering step has to be included (see \citep{Tib08}).
}

\ha{
The unbalanced Haar estimator also has a tendency to include too many jumps, even though the effect is much smaller than for LASSO type methods, i.e. it is much sparser w.r.t. the number of change-points.
}
\begin{figure}[htp]
 \includegraphics[width= \columnwidth]{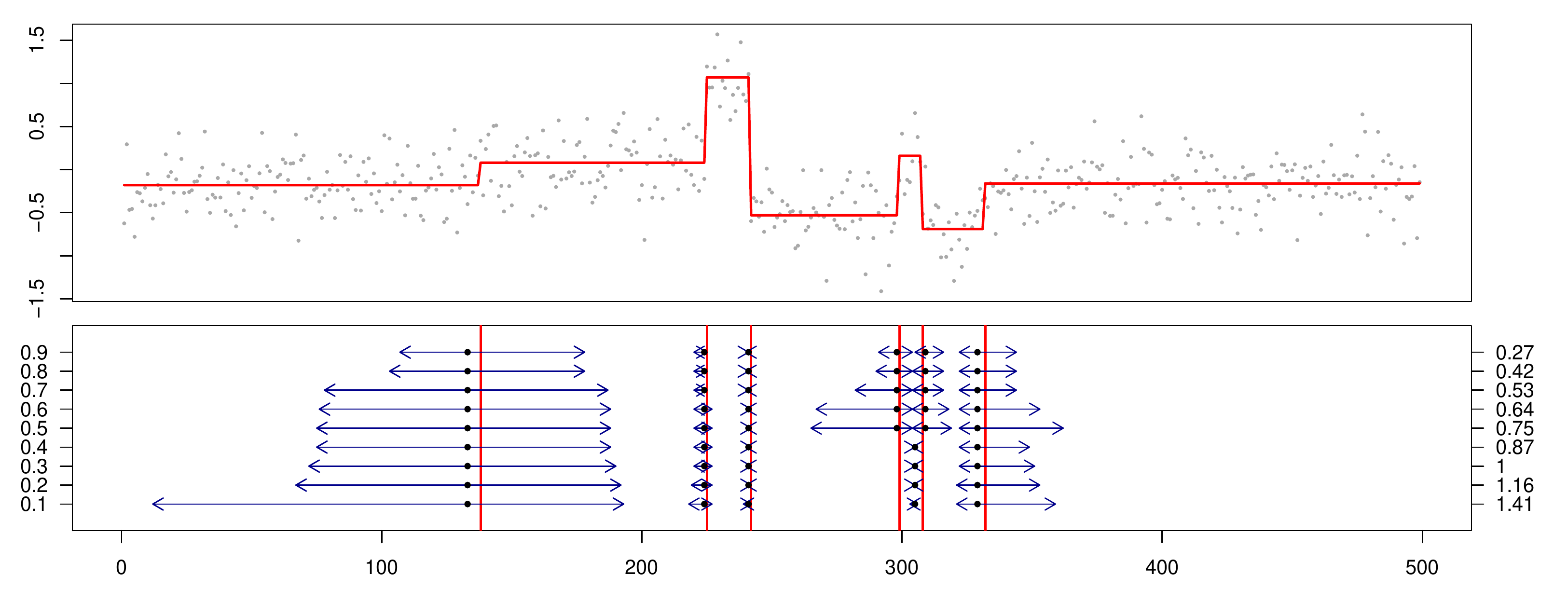}
 \caption{\ha{Top: typical example of model \eqref{sim:modeltrend} for $b=0$ and $\sigma^2=0.3$; bottom: change-points and confidence intervals for SMUCE with $\alpha =0.1,\ldots,0.9$} (left y-axis) and the corresponding quantiles $q_{1-\alpha}$ (right $y$-axis)}
 \label{sim:gaussmean:typex2}
\end{figure}
 
\begin{figure}[b!]
 \includegraphics[width=\columnwidth]{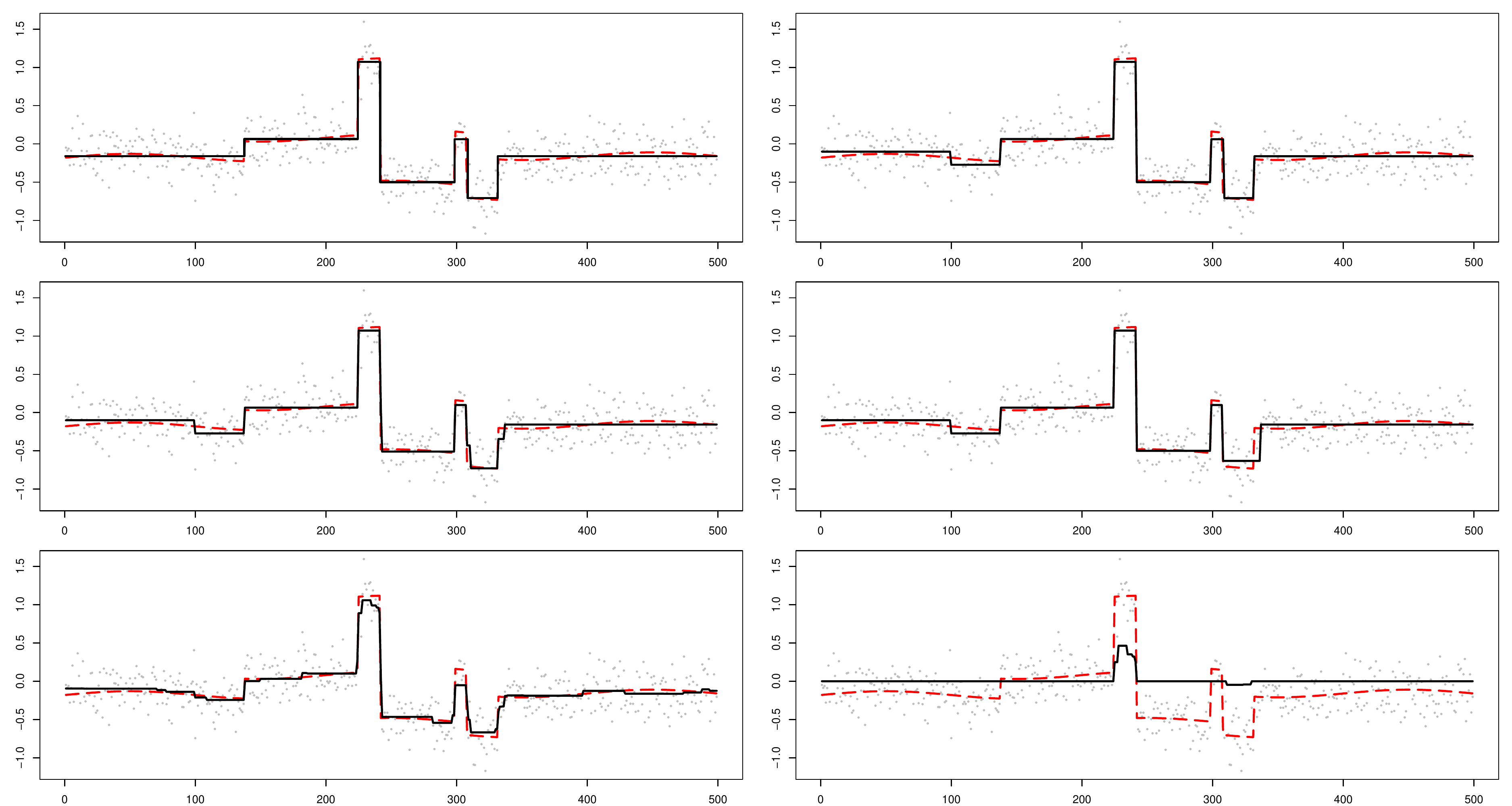}
\caption{\ha{An example of model \eqref{sim:modeltrend} for $a=0.01$, $b=0.1$ and $\sigma=0.2$.
From top left to bottom right: SMUCE, mBIC, unbalhaar, CBS,
$\text{FL}^\text{MSE}$ and $\text{FL}^\text{cp}$} (solid black line) and the true signal (dashed red line).}
\label{sim:gaussmean:typex}
\end{figure}

Again, we note that Table \ref{sim:gau:tablecomp} can be complemented by the
simulation study in \citep{ZhaSie07} which accounts for the classical
BIC \citep{Sch78} and the method suggested in \citep{FriSniAntPinAlb04}.

\subsection{Gaussian variance regression}
Again, we consider normal data $Y_i$, however,
in contrast to the previous section we aim to estimate the variance $\sigma^2 \in
\S$. For simplicity we set $\mu=0$. This constitutes a natural exponential family with natural
parameter $\theta=-(2\sigma^2)^{-1}$ and $\psi(\theta)=-\log(-2\theta)/2$ for
the sufficient statistic $Z_i = Y_i^2$, $i=1,\ldots,n$. It is easily seen that the MR-statistic in this
case reads as
\begin{equation*}
T_n(Z,\hat\sigma^2) = \max_{0\leq k\leq \hat K}\max_{\hat l_k<i\leq j\leq \hat
l_{k+1}} \left( \frac{\sqrt{j-i+1}}{\sqrt{2}}\sqrt{ \frac{
\overline{Z}_i^j}{\hat\sigma_k^2} - \log\frac{\overline{Z}_i^j}{\hat\sigma_k^2}
-1 } - \sqrt{2 \log\frac{e n}{j-i+1}} \right).
\end{equation*}
After selecting the model $\hat K(q)$ according to
\eqref{conscp:estnocp}, the SMUCE is given by
\begin{equation*}
\hat \sigma^2(q) = \argmax_{\hat\sigma^2\in \S_n[\hat
K(q)]} \sum_{k=0}^{\hat K(q)}(\hat l_{k+1} - \hat
l_k)\left(\log \frac{1}{\hat\sigma_k^2} - \frac{\overline
Z_{\hat l_l}^{\hat l_{k+1}}}{\hat\sigma_k^2} \right),\quad\text{ s.t. }\quad
T_n(Z,\hat\sigma^2)\leq q.
\end{equation*}   
 
We compare our method to \citep{Hoe08, DavHoeKra12}. Similar to
\text{SMUCE} they propose to minimize the number of change-points under a multiscale
constraint.
They additionally restrict their final estimator to coincide with the local maximum likelihood 
estimator on  constant segments. As pointed out by the authors this may increase the number of detected change-points.
Following their simulation study we consider test signals $\sigma_k$
with $k=0,1,4,9,19$ equidistant change-points and constant values alternating
from 1 to 2 ($k=1$), from 1 to 2 ($k=4$), from 1 to 2.5 ($k=9$) and from 1 to
3.5 ($k=19$). For this simulation the parameter of both procedures are chosen
such that the number of changes should not be overestimated  with probability $0.9$.
\ha{For any signal we computed both estimates in $1000$ simulations. The difference of true and estimated number change-points as well as the MISE and MIAE are shown in Table \ref{sim:gaussvar:tab1}. 
Considering the number of correctly estimated change-points, it shows that
\text{SMUCE} performs better for few changes ($k=1,4,9$) and worse for many
changes ($k=19$). This may be explained by the fact that the multiscale test in
\citep{DavHoeKra12} does not include a scale-calibration and is hence more
sensible on small scales than on larger ones, see also Subsection \ref{subsec:diss:pen}. With respect to MISE and MIAE
the \text{SMUCE} outperforms in every scenario, interestingly even for $k=19$, where \citep{DavHoeKra12} performs better w.r.t. the estimated number of change-points.}

\begin{table}[ht]
\tiny
\centering
\begin{tabular}{cc||ccc>{\columncolor{tablebg}}cccc|cc}

&k & -3 & -2 & -1 & 0 & +1 & +2 & +3 & MISE & MIAE \\ 
  \hline
SMUCE & 0 & 0.000 & 0.000 & 0.000 & \textbf{0.945} & 0.053 & 0.002 & 0.000 & \textbf{0.00072} & \textbf{0.02040} \\ 
 \citep{DavHoeKra12} & 0 & 0.000 & 0.000 & 0.000 & 0.854 & 0.127 & 0.019 & 0.000 & 0.00093 & 0.02122 \\ \hline
SMUCE &  1 & 0.000 & 0.000 & 0.000 & \textbf{0.975} & 0.024 & 0.001 & 0.000 & \textbf{0.00653} & \textbf{0.04295} \\ 
 \citep{DavHoeKra12} & 1 & 0.000 & 0.000 & 0.000 & 0.901 & 0.089 & 0.009 & 0.001 & 0.00935 & 0.04648 \\ \hline
SMUCE &  4 & 0.000 & 0.000 & 0.000 & \textbf{0.997} & 0.003 & 0.000 & 0.000 & \textbf{0.02153} & \textbf{0.07967} \\ 
  \citep{DavHoeKra12} &4 & 0.000 & 0.000 & 0.000 & 0.957 & 0.042 & 0.001 & 0.000 & 0.03378 & 0.09655 \\ \hline
SMUCE &  9 & 0.000 & 0.001 & 0.023 & \textbf{0.973} & 0.003 & 0.000 & 0.000 & \textbf{0.06456} & \textbf{0.13206} \\ 
 \citep{DavHoeKra12} & 9 & 0.000 & 0.000 & 0.009 & 0.968 & 0.023 & 0.000 & 0.000 & 0.11669 & 0.18297 \\ \hline
SMUCE & 19 & 0.000 & 0.027 & 0.222 & 0.751 & 0.000 & 0.000 & 0.000 & \textbf{0.26076} & \textbf{0.27468} \\ 
 \citep{DavHoeKra12} &  19 & 0.000 & 0.008 & 0.074 & \textbf{0.912} & 0.006 & 0.000 & 0.000 & 0.47105 & 0.40606 \\ 
\end{tabular}
  \caption{\ha{Comparison of \text{SMUCE} and the method in \citep{DavHoeKra12}. Difference between the estimated and the true number of change-points for $k=0,1,4,19$ change-points as well as MISE and MIAE for both
  estimators. }}\label{sim:gaussvar:tab1}
\end{table}

  
\subsection{Poisson regression}
We consider the Poisson-family of
distributions with intensity $\mu>0$. Then, $\theta= \log \mu $ and
$\psi(\theta)=\exp \theta$. The MR-statistic is computed as
\begin{equation*}
T_n(Y,\hat \mu) = \max_{0\leq k\leq \hat K}\max_{\hat l_k<i\leq j\leq \hat
l_{k+1}} \left( \sqrt{2 (j-i+1 )}\sqrt{ \overline Y_i^j \log \frac{\overline Y_i^j}{\mu_k} +
\mu_k - \overline Y_i^j} - \sqrt{2 \log\frac{en}{j-i+1}} \right).
\end{equation*}

For $\hat K(q)$ as in \eqref{conscp:estnocp}, the \text{SMUCE} is given by
\begin{equation*}
\hat \mu(q) = \argmax_{\hat\mu\in \S_n[\hat
K(q)]} \sum_{k=0}^{\hat K(q)}(\hat l_{k+1} - \hat
l_k)(\overline Y_{\hat l_l}^{\hat l_{k+1}} \log \hat\mu_k - \hat
\mu_k)\quad\text{ s.t. }\quad T_n(Y,\hat\mu)\leq q.
\end{equation*} 
In applications (c.f. the example from photoemission spectroscopy below), one is often faced
with the problem of \emph{low count} Poisson data, i.e. when the intensity 
$\mu$ is small. It will turn out that in this case, data transformation towards Gaussian variables such as variance
 stabilizing transformations are not always sufficient and it pays off to take into account
the Poisson likelihood into SMUCE.

\begin{figure}[h!] 
\includegraphics[width=1\columnwidth]{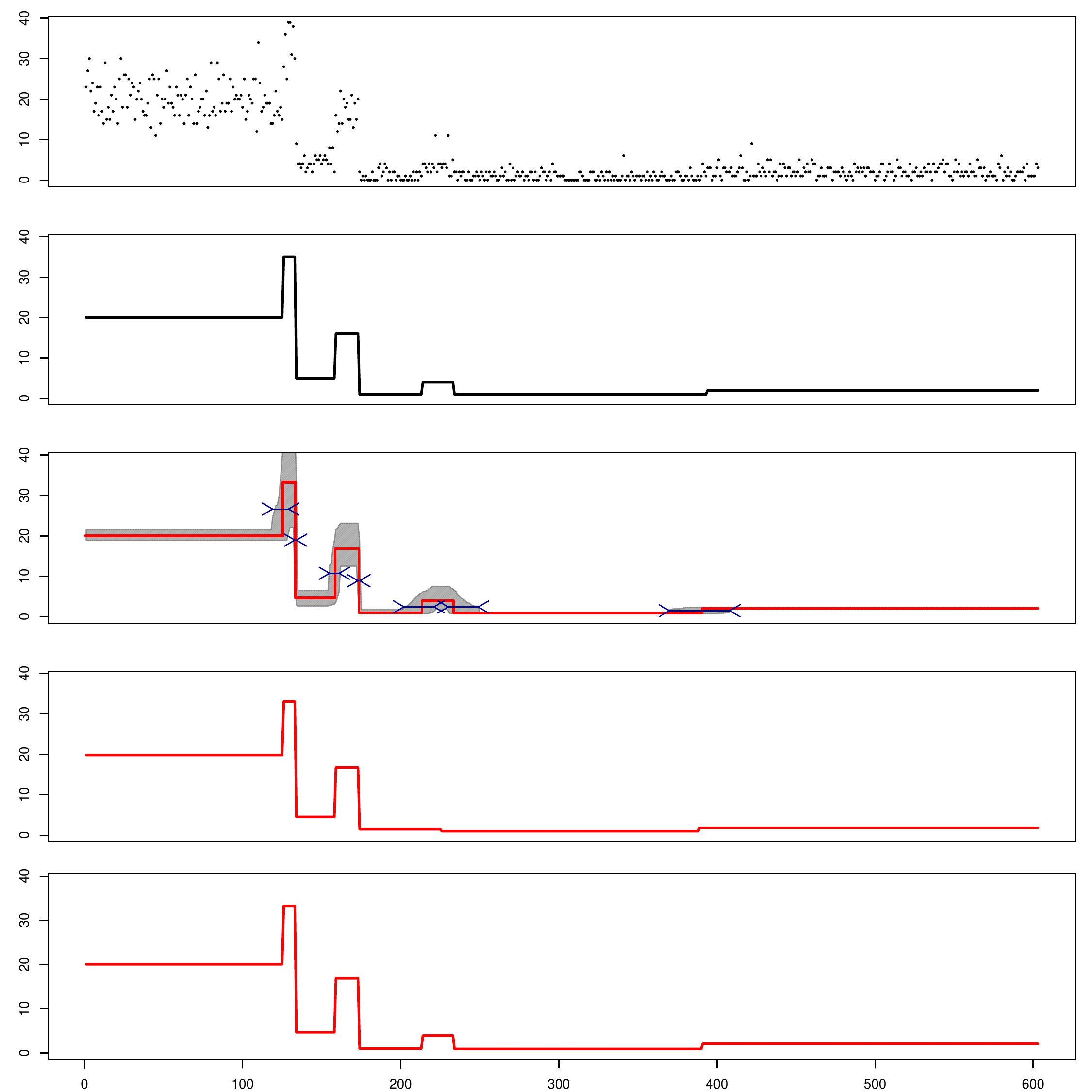}
  \caption{from top to bottom: simulated data, true signal, SMUCE with confidence bands for the signal intensities
  (gray area) and confidence intervals for the change-points (inward pointed
  arrows), $\text{SMUCE}_\text{mm}$ and Ploracle.}
\label{sec:pois:testf}
\end{figure}

In the following we perform a simulation study where we use a signal with a low
count and a spike part (see top panel of Figure \ref{sec:pois:testf}).
 In order to evaluate the performance of the \text{SMUCE} we compare it to the
 BIC estimator and the PLoracle as described before. Moreover, we included a
version of the SMUCE which is based on variance stabilizing transformations of the data.
 To this end, we applied the \emph{mean-matching} transformation \citep{BroCaiZho10} to preprocess the data.
 We then compute the SMUCE under a Gaussian model and retransform the obtained estimator by the inverse mean-matching transform. The resulting
 estimator is referred to as $\text{SMUCE}_\text{mm}$. Moreover, as a benchmark, we compute
 the (parametric) maximum likelihood estimator with $K=7$ change-points, which
 is referred to as MLoracle. 

\begin{table}[ht]
\tiny
\begin{center}
\begin{tabular}{l||cc>{\columncolor{tablebg}}ccc|ccc}
 & $\leq$5 & 6 & 7 & 8 & $\geq$9 & MISE & MIAE & Kullback-Leibler \\ 
  \hline
  SMUCE &0.000 & 0.067 & \textbf{0.929} & 0.004 & 0.004 & \textbf{0.274} & \textbf{0.217} & \textbf{0.0187} \\
  $\text{SMUCE}_{\text{MS}}$ &0.000 & 0.067 & \textbf{0.929} & 0.004 & 0.004 & {0.282} & {0.219} &  {0.0194} \\
   BIC & 0.000 & 0.000 & 0.080 & 0.094 & 0.920 & 0.575 & 0.313 & 0.0417 \\ 
  $\text{SMUCE}_{\text{mm}}$ & 0.013 & 0.420 & 0.561 & 0.005 & 0.006 & 0.434 & 0.364 & 0.0418 \\   \hline
  PLoracle &0.045 & 0.014 & 0.942 & 0.000 & 0.000 & 0.275 & 0.217 & 0.0185 \\ 
  MLoracle & 0.000 & 0.000 & 1.000 & 0.000 & 0.000 & 0.258 & 0.208 & 0.0143 \\ 
\end{tabular}
    \caption{Frequencies of $\hat K$ and distance measures for SMUCE, the BIC
    \citep{Sch78}, the SMUCE for variance stabilized signals as well as the
    PLoracle and MLoracle.}
    \label{sim:pois:values}
\end{center}
\end{table} 

Table \ref{sim:pois:values} summarizes the simulation results. As to be expected the standard BIC
performs far from satisfactorily. We stress that SMUCE clearly outperforms the $\text{SMUCE}_\text{mm}$, which is based on
Gaussian transformations. Note, that the $\text{SMUCE}_\text{mm}$ systematically underestimates the number of change-points $K=7$
which highlights the difficulty to capture those parts of the signal correctly, where the intensity is low. 
Again, SMUCE performs almost as good as the Poisson-oracle PLoracle.
 To get a visual impression along with the results of Table \ref{sim:pois:values}, we
illustrated these estimators in Figure \ref{sec:pois:testf}. 

\subsection{Quantile regression}

\ha{Finally, we extend our methodology to quantile regression. Let the
observations $Y_1,\ldots,Y_n$ be given by model \eqref{intro:model}, without any assumption on
 the underlying distribution. For some $\beta\in(0,1)$, we now aim for estimating the corresponding
 (piecewise-constant) $\beta$-quantile function, which will be denoted by $\vartheta_\beta$.
This problem can be turned into a Bernoulli regression as follows: Given the $\beta$-quantile function $\vt_\beta$ define the random variables $W(\vt)=(W_1,\ldots,W_m)$ as
\begin{equation*}
 W_i=\begin{cases}
      1 & \text{ if } Y_i\leq \vt_\beta(i/n) \\
      0 & \text{ otherwise }
     \end{cases}
,\quad i=1,\ldots,n.
\end{equation*}
Then, $W_1,\ldots,W_n$ are i.i.d. Bernoulli random variables with mean value
$\beta$. Extending the idea in Subsection \ref{subsec:intro:estimator} we compute a solution of
\eqref{intro:smre}, where $T_n(W(\vt_\beta))$ denotes the multiscale statistic for Bernoulli observations which reads as
\begin{equation*}\label{sim:quantile:MSS}
T_n(W(\vt_\beta),\beta)=\max_{\substack{1\leq i \leq j \leq n \\ \vt_\beta \text{is constant on }[i/n,j/n]}} \left( \sqrt{2 T_i^j(W(\vt_\beta),\beta)} - \sqrt{2 \log \frac{e n}{j-i+1}} \right)
\end{equation*}
with
\begin{equation*}
T_i^j(W(\vt_\beta)),\beta) = (j-i+1) \left(\bar{W}_i^j \log\left(\frac{\bar{W}_i^j}{{\beta}} \right) + (1-\bar{W}_i^j) \log \left( \frac{1- \bar{W}_i^j}{1-{ \beta}} \right) \right).
\end{equation*}
In other words, we compute the estimate with fewest change-points, such that the signs of the residuals fulfill the multiscale test for Bernoulli observations with mean $\beta$.
The computation of this estimate hence results in the same type of optimization problem as treated in subsection \ref{subsec:mincosts} and we can apply the proposed methodology.}

\ha{In the following we compare this approach with a generalized taut string algorithm \citep{DavKov01}, which was proposed in \citep{DueKov09}, for estimating quantile functions. The estimate is constructed in such a way that it minimizes the number of local
extreme values among a specified class of functions. Here, a local extreme value is either a local maximum or a local minimum.}

\ha{In contrast to SMUCE the number of change-points is not penalized. In a simulation study the authors showed that their method is particularly suitable to detect local extremes of a signal.
We follow this idea and repeated their simulations. The results which also include the estimated number of change-points, are shown in Table \ref{sim:quantile:tab}. It can be seen that the gen. taut string estimates the number of local extremes slightly better 
than SMUCE, while the number of change-points is overestimated for $n=2048$ and $n=4096$. This may be explained by the fact that the generalized taut string is not primarily designed to have few change-points rather few local extremes.}


\begin{table}[ht]
\tiny
\centering
\begin{tabular}{ll||c|c|c|c|c|c}
& & \multicolumn{3}{c|}{local extreme values} & \multicolumn{3}{c}{change-points} \\ \hline
& n & $\beta=0.5$ & $\beta=0.1$ & $\beta=0.9$ & $\beta=0.5$ & $\beta=0.1$ & $\beta=0.9$  \\ \hline \hline
SMUCE & $512$ & 3 (5.9) & 1 (7.9) & 2 (7.4) & 5 (5.8) & 2 (9.1) & 3 (8.3) \\
gen. taut string & $512$ & 3(6.0) & 3 (6.6) & 3 (6.6) & 12 (2.0) & 6 (4.9) & 7 (4.0) \\ \hline
SMUCE & $2048$ & 9 (0.4) & 4 (5.4) & 3 (5.8) & 11 (0.1) & 6 (5.2) & 5 (5.9) \\
gen. taut string & $2048$ & 9 (0.7) & 5 (4.0) & 3 (5.7) & 26 (15.3) & 18 (7.1) & 16 (5.7) \\ \hline
SMUCE & $4096$ & 9 (0.1) & 4 (4.3) & 5 (4.5) & 11 (0.2) & 8 (3.1) & 6 (4.8) \\
gen. taut string & $4096$ & 9 (0.0) & 6 (3.1) & 3 (5.3) & 35 (24.1) & 25 (13.8) & 21 (9.9) \\
\end{tabular}
\caption{Comparison of \text{SMUCE} and generalized taut string \citep{DueKov09}. 
\ha{Median of local extreme values/ change-points of the estimators and mean absolute difference (in brackets) to true number of local extremes/ change-points. The true number of local extremes equals $9$ and the true number of change-points equals $11$.}}
\label{sim:quantile:tab}
\end{table}

\begin{figure}[h!]
  \includegraphics[width=0.75 \columnwidth]{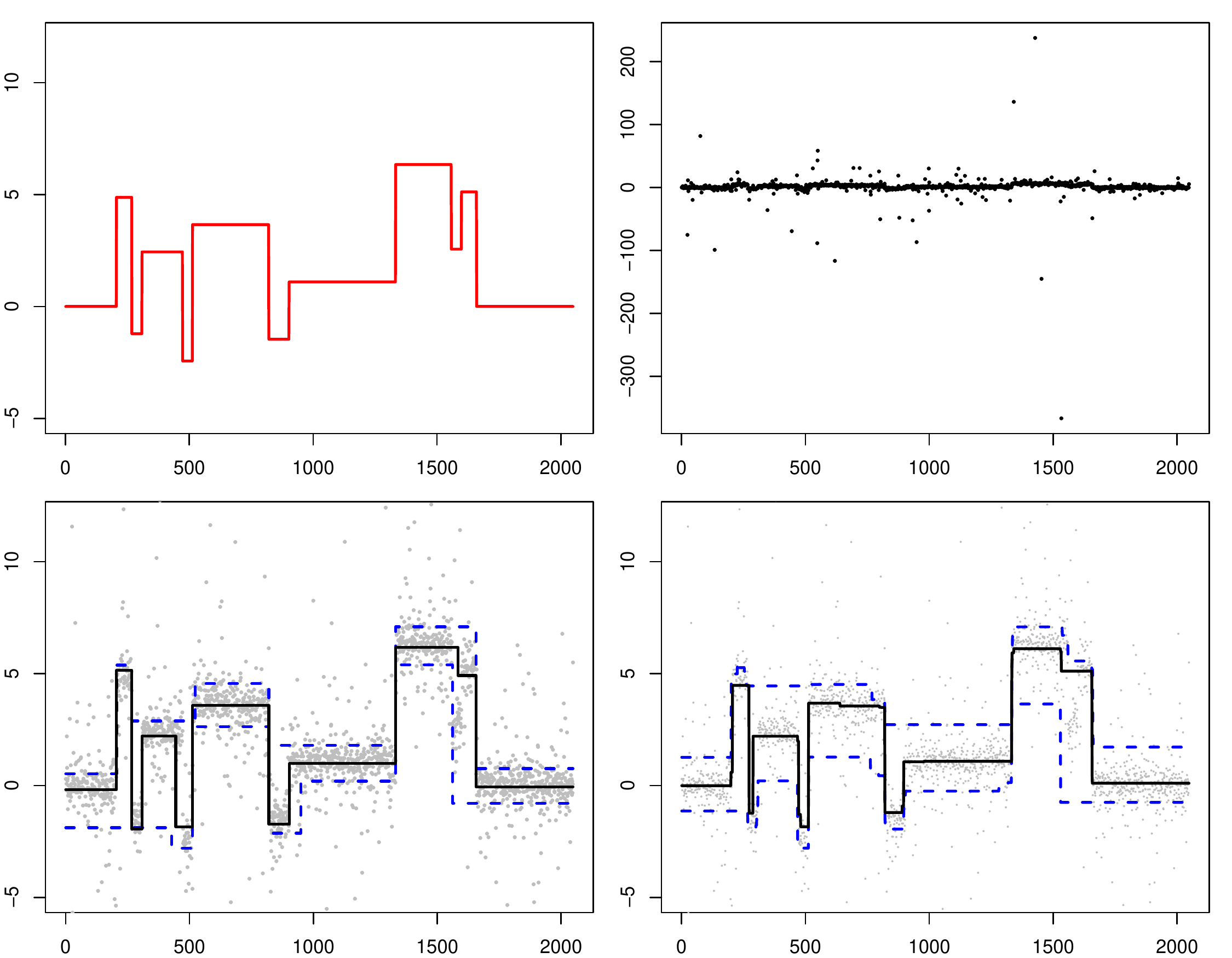}
  \caption{First row: block signal (left) and simulated data (right). Second row: Estimator for median (solid), 0.1 and 0.9-quantiles (dashed) from SMUCE (left) and generalized taut string (right)}
\end{figure}

\subsection{On the coverage of confidence sets $I(q)$}
In Section \ref{sec:convset} we gave asymptotic results on the simultaneous
coverage of the confidence sets $I(q)$ as defined in \eqref{convset:maindef}.
 In our simulations we choose $q=q_{1-\alpha}$ to be the $1-\alpha$-quantile of
 $M$ as in \eqref{limit:limitstat}. It then follows from Corollary
 \ref{confband:corhonest} that asymptotically the simultaneous coverage is
 larger than $1-\alpha$. We now investigate empirically the simultaneous
 coverage of $I(q_{1-\alpha})$. To this end, we consider the test signals  shown
 in Figure \ref{sim:cov:signals} for Gaussian observations with varying mean, Gaussian
observations with varying variance, Poisson observations and Bernoulli
observations.

Table \ref{sim:cov:tab1} summarizes the empirical coverage for
different values for $\alpha$ and $n$ obtained by $500$ simulation runs each and the relative frequencies
of correctly estimated change-points, which are given in brackets.
The results show that for $n = 2000$ the empirical coverage exceeds $1-\alpha$
in all scenarios. The same is not true for smaller $n$ (indicated by bold
letters), since here the number of change-points is misspecified rather
frequently (see numbers in brackets). Given $K$ has been estimated correctly, we find that
the empirical coverage of bands and intervals is in fact larger
than the nominal $1-\alpha$ for all simulations.

\begin{figure}[h!]
 \includegraphics[width=\columnwidth]{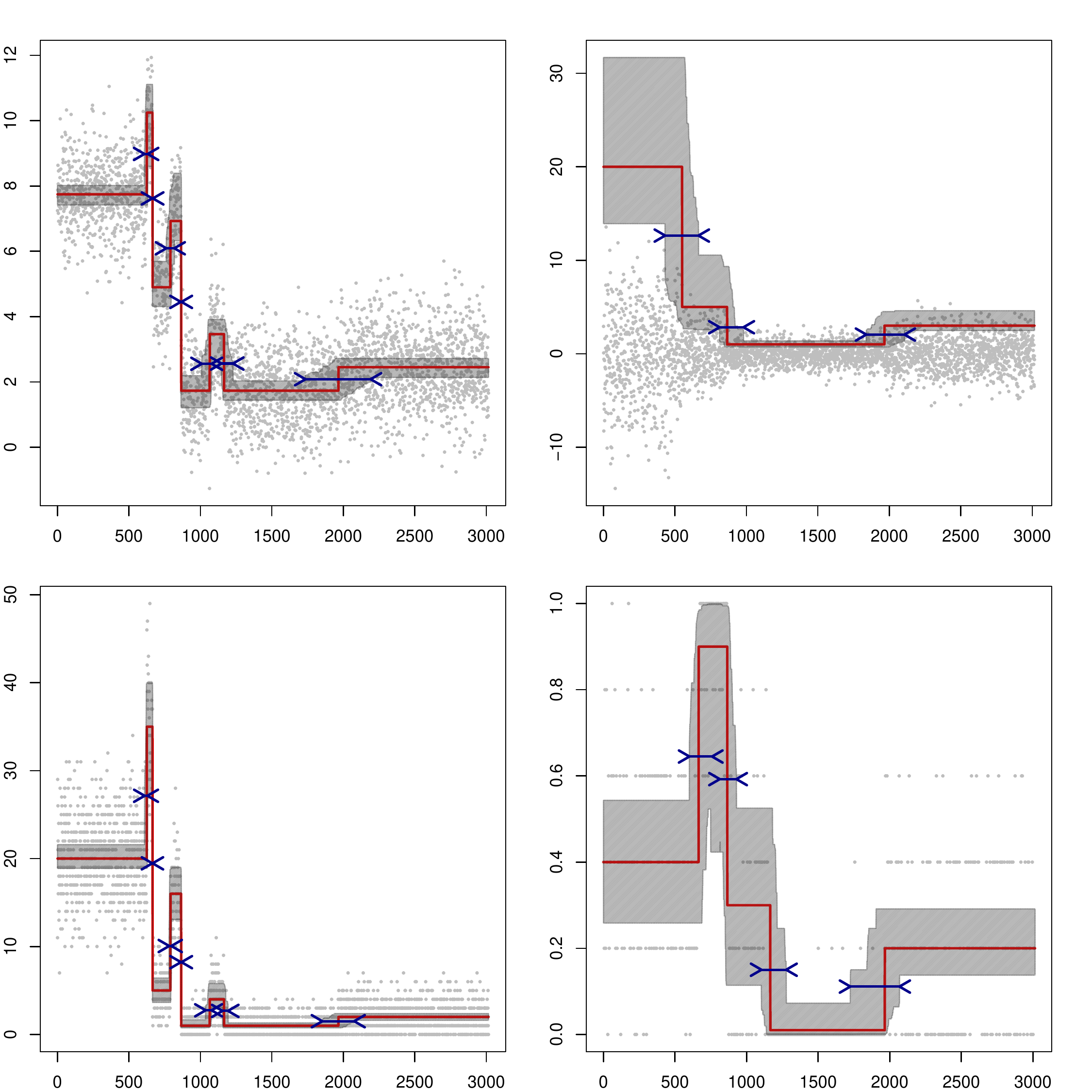}
\caption{f.l.t.r.: Gaussian observations with varying mean, Gaussian
observations with varying variance, Poisson and (binned)
 Bernoulli observations and SMUCE (solid red line) with confidence bands (grey hatched)
 and confidence intervals for change-points (inwards pointing arrows).}
\label{sim:cov:signals}
\end{figure}

\begin{table}[h!]
\tiny
\begin{center}
\begin{tabular}{l|l||ccc|ccc|ccc|ccc}
$n$&$1-\alpha$	 & \multicolumn{3}{c|}{Gaussian} & \multicolumn{3}{c|}{Gaussian} & \multicolumn{3}{c|}{Poisson} & \multicolumn{3}{c}{Bernoulli} \\
   &		 &\multicolumn{3}{c|}{(mean)}&\multicolumn{3}{c|}{(variance)}&&&&&& \\ \hline \hline
& $0.8$ 	& \textbf{0.59} & 0.64 & 0.92 &\textbf{ 0.66}& 0.68& 0.97 & 0.87& 0.89& 0.98 & 0.85& 0.90& 0.94\\ 
$1000$& $0.9$ 	& \textbf{0.48} &0.49& 0.98 & \textbf{0.39}& 0.39& 1.00 & \textbf{0.85}& 0.86& 0.99 & \textbf{0.86}& 0.86& 0.99 \\ 
& $0.95$ 	& \textbf{0.28} &0.28& 1.00 & \textbf{0.16}& 0.18& 0.93 & \textbf{0.71}& 0.74& 0.96 & \textbf{0.66}& 0.70& 0.94 \\ 	\hline

& $0.8$ 	& 0.84& 0.90& 0.93 &  0.87& 0.88& 0.98 & 0.92& 0.95& 0.96 & 0.93& 0.97& 0.96 \\ 
$1500$& $0.9$ 	& \textbf{0.73}& 0.74& 0.98 & \textbf{0.72} & 0.74 & 0.97 & \textbf{0.95}& 0.97& 0.98 &{0.96} &0.97& 0.99 \\
& $0.95$ & 	\textbf{0.55}   &0.56& 0.98 & \textbf{0.45} &0.47& 0.98 & \textbf{0.92} &0.93& 0.99 & \textbf{0.89}& 0.90& 0.99 \\ 	\hline

& $0.8$ 	& 0.94 & 0.99 & 0.95  & 0.98& 1.00& 0.98 & 0.95 &0.99& 0.95 & 0.96 &0.99& 0.97 \\ 
$2000$& $0.9$ 	& 0.98 &1.00& 0.98 & 0.99 &1.00& 0.99 & 0.96 &0.99& 0.96 & 0.97 &0.99& 0.98 \\ 
& $0.95$ 	& 0.99 &1.00& 0.99 & 0.97 &0.99& 0.98 & 1.00 &1.00& 1.00 & 0.99 &1.00& 0.99 \\ 	\hline
%
%
   \hline
\end{tabular}
\caption{Empirical coverage obtained from $500$ simulations for the signals shown in Figure \ref{sim:cov:signals}.
For each choice of $\alpha$ and $n$ we computed the simultaneous coverage of $I(q)$, as in \eqref{convset:maindef} (first value),
the percentage of correctly estimated number of change-points (second value) and the simultaneous coverage of confidence 
 bands and intervals for the change-points given $\hat K(q)=K$ (third value).}
 \label{sim:cov:tab1}
\end{center}
\end{table}

\subsection{Real data results}\label{sec:realdata}
In this section we analyze two real data examples. The examples show the variety of
possible applications for \text{SMUCE}. Moreover, we revisit the issue of
choosing $q$ as proposed in Section \ref{sim:thresh} and illustrate \ha{its}
applicability to the present tasks.
\subsubsection{Array CGH data}\label{sec:realdata:gauss}
Array Comparative Genomic Hybridization (CGH) data show aberrations in genomic
DNA. The observations consist of the log-ratios of normalized intensities from
disease and control samples. The statistical problem at hand is to identify
regions on which the ratio differs significantly from $0$ (which corresponds to
a gain or a loss). These are often referred to as aberration regions.

A thorough overview of the topic and a comparison of
several methods is given in \citep{Lai}. We compute the \text{SMUCE} for
two data sets studied in \citep{Lai} and more recently in \citep{DuKou12, Tib08}.
The data sets show the Array-CGH profile of chromosome 7 
in GBM29 and chromosome 13 in GBM31, respectively (see also again \citep{DuKou12, Lai}).

By means of these two data examples we illustrate how the developed theory in Section \ref{sec:mrstat}
 can be used for applications. As it was stressed in
\citep{Lai} many algorithms in change-point detection do strongly depend on the proper
choice of a tuning parameter, which is often a difficult task in practice. We point out
that our proposed choice of the threshold parameter $q$ has in fact a statistical
meaningful interpretation as it determines the level of the confidence set $C(q)$.
Moreover, we will emphasize the usefulness of confidence
bands and intervals for Array CGH data.
 
\begin{figure}[h!]
 \includegraphics[width=0.9\columnwidth]{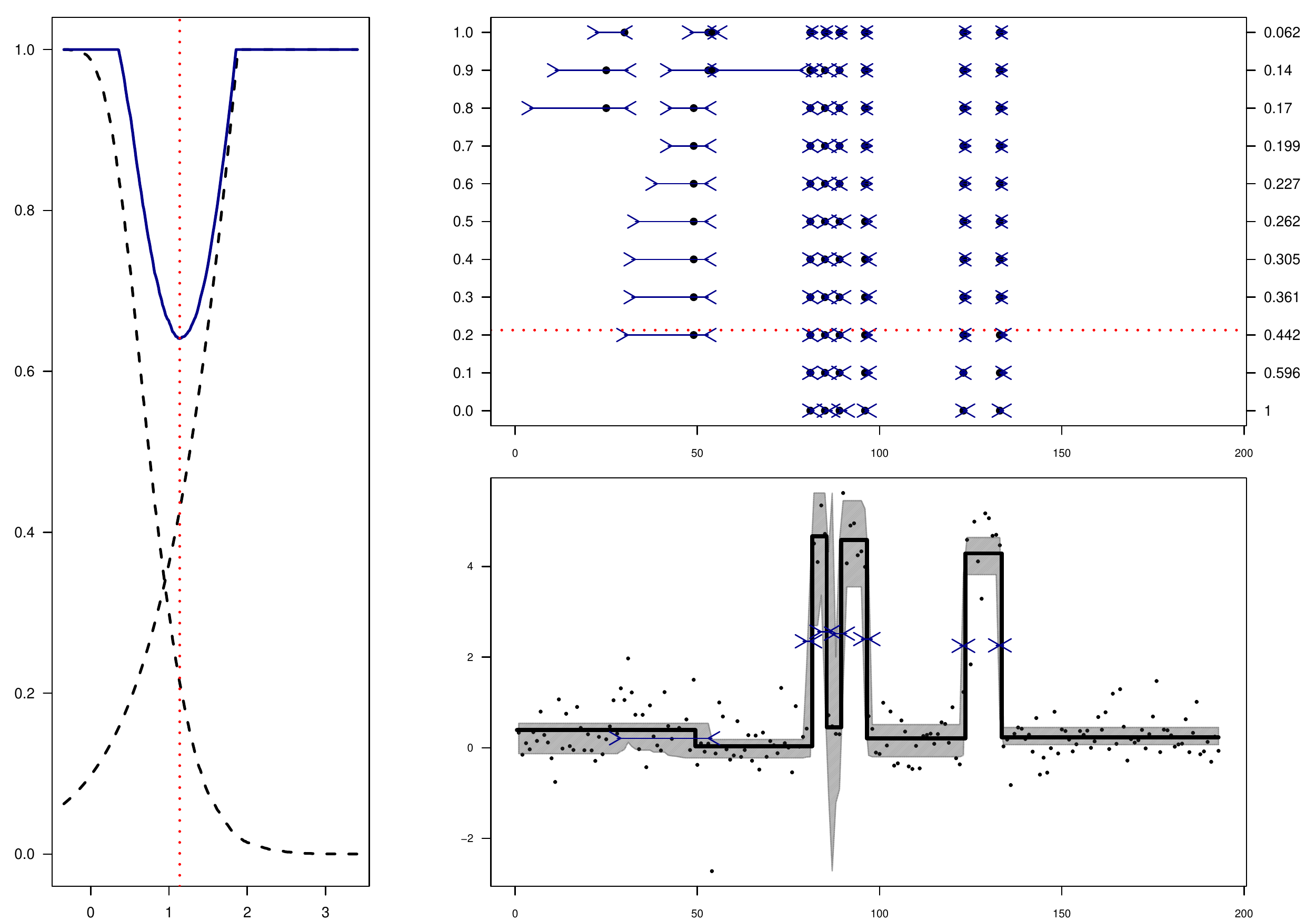}
 \caption{\ha{Left: Probability for over/underestimating  (decreasing/increasing dashed line) the number of
 change-points in dependence of $q$ ($x$-axis) and their sum (solid line). Top right: Detected
 change-points with confidence intervals for different values of $\alpha$ (left $y$-axis) with the probabilty of underestimation (right $y$-axis). Bottom right: SMUCE (solid line) computed for the optimal $q^*\approx 1.1$ with confidence bands (grey hatched) and confidence intervals for change-points (inwards pointing arrows).}}
\label{data:CGH:29:1}
\end{figure}  

\ha{We first consider the GBM29 data.
In order to choose $q$ according to the suggested proceeding in \eqref{sim:q:qopt1}, assumptions 
on $\lambda$ and $\Delta$ have to be imposed.
\ha{
As mentioned above log ratios of copy numbers may take on a finite number of values which are approximately $\left\{\log(1), \log(3/2), \log(2), \log(5/2),\ldots \right\}.$
It therefore seems reasonable to assume that the smallest jumps size is $\Delta=log(3/2)$.
}
Moreover, we choose  $\lambda\geq 0.2$.We stress that the final solution of the
 SMUCE will not be restricted to these assumptions. They enter as prioir assumptions for the choice of $q$. If the data speak strongly against these assumptions SMUCE will adapt to this.} 
  
In the left panel of Figure \ref{data:CGH:29:1}
we depict the probability of overestimating the number of change-points as a function of $q$ (decreasing dashed line)
and the probability of overestimating the number of change-points as a function of $q$ (increasing dashed line) under the
above stated assumption on $\lambda$ and $\Delta$.
\ha{One may interpret the plot in the following way. It provides a tool for finding jumps of minimal height $\Delta=log(3/2)$ on scales of at least $\lambda=0.2$. For the optimized $q^*$ we obtain, that the number jumps is misspecified with probability less than $0.35$. For the corresponding estimate
see Figure \ref{data:CGH:29:1}.}

Moreover, for different choices of $q$ we displayed the \text{SMUCE}.
The top-right panel of Figure \ref{data:CGH:29:1} shows the
estimated change-points with its confidence intervals. Bounds for
the probability that $K$ is overestimated can be found on the left axis,
bounds for underestimation on the right axis.

Note from the top-right image in Figure \ref{data:CGH:29:1} that the SMUCE is quite
robust w.r.t. $q=q_{1-\alpha}$. For $\alpha\in [0.2,0.7]$ SMUCE always detects exactly $7$ change-points
in the signal.
\ha{The results show that a jump of the size $\approx \Delta$ is found in the data on an interval, which length is even slightly smaller than $\lambda$. However, SMUCE is also able to detect larger abberations on smaller intervals, which makes it quite robust against wrong choices of $\Delta$ and $\lambda$.}

\ha{Recall that one goal in Array CGH data analysis} is to determine segments on
which the signals differs from $0$. The confidence sets in the right lower plot 
indicate three intervals with signal different from $0$. 
Moreover, as indicated by the blue arrows, the change-point locations are
detected very precisely. Actually, the estimator suggests one more change-point
in the data. However, it can be seen from the confidence bands that there is only small 
evidence for the signal to be nonzero.
\ha{Further, the confidence bands may be used to decide which segments belong to the same copy number event. In this particular example the confidence bands suggest that these three segments belong to the same copy number event, i.e. have the same mean value.}

Put differently, not only an estimator for the true signal is obtained, but also
$3$ regions of abberation were detected and simultaneous
confidence intervals for the signal's value on this regions at a level of
$1-\alpha=0.9$ are given. This is in accordance with others' findings \citep{Lai, DuKou12}.

\begin{figure}[h!]
 \includegraphics[width=0.9\columnwidth]{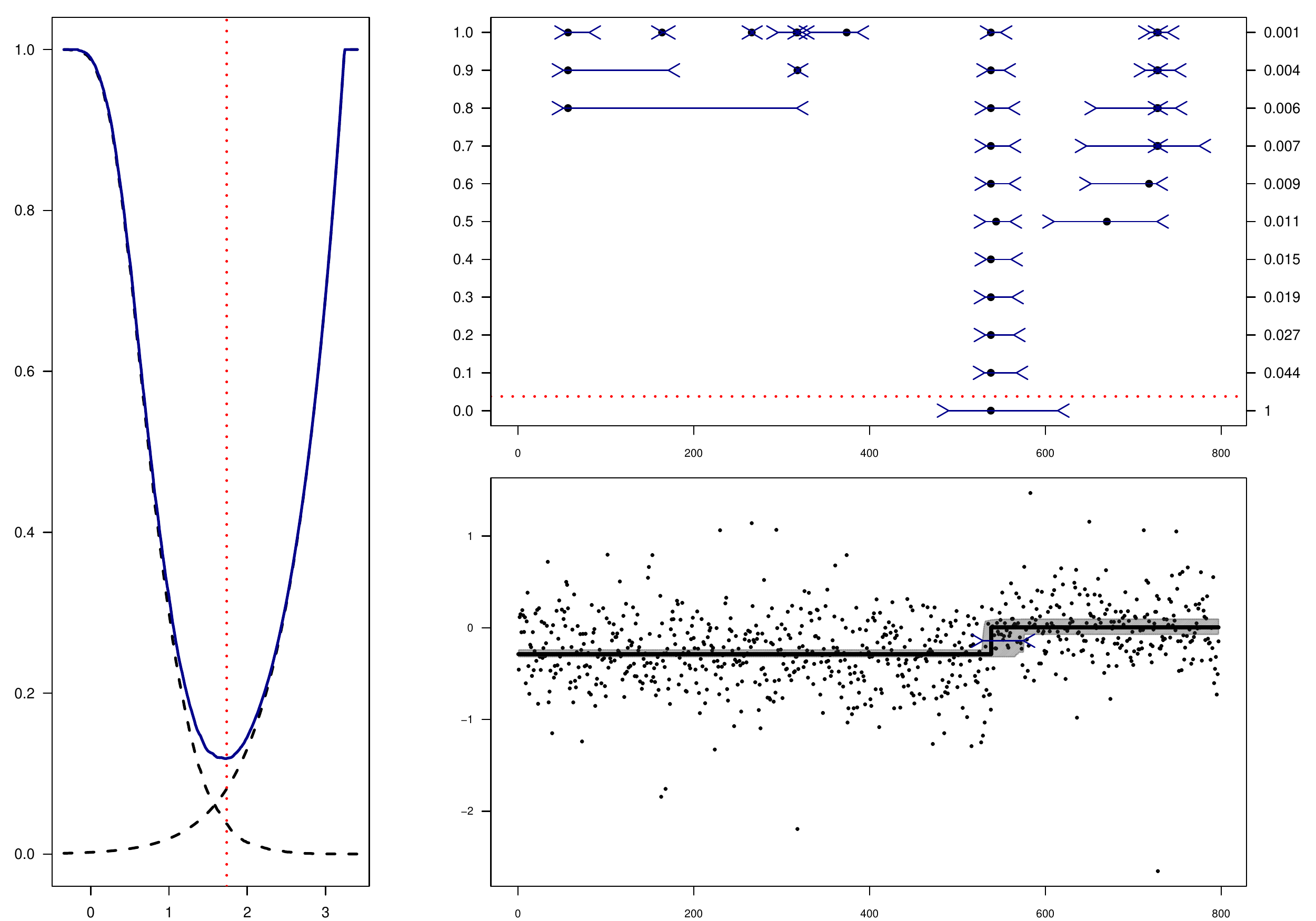}
caption{\ha{Left: Probability for over/underestimating  (decreasing/increasing dashed line) the number of
 change-points in dependence of $q$ ($x$-axis) and their sum (solid line). Top right: Detected
 change-points with confidence intervals for different values of $\alpha$ (left $y$-axis) with the probabilty of underestimation (right $y$-axis). Bottom right: SMUCE (solid line) computed for the optimal $q^*\approx 1.7$ with confidence bands (grey hatched) and confidence intervals for change-points (inwards pointing arrows).}}
\label{data:CGH:31:1}
\end{figure}

The same procedure as above is repeated for the GBM31 data as shown in Figure
\ref{data:CGH:31:1}. For the bounds on underestimating the number of
change-points we assumed again that $\Delta \geq \log(3/2)$ and chose $\lambda\geq 0.025$.
The plots in Figure \ref{data:CGH:31:1} show that $\Delta \geq \log(3/2)$ for the sample size of $n=797$ the probability
of misspecification can be bounded by $\approx 0.12$ for the minimal length $\lambda = 0.025$, which corresponds to $19$
observations.
Using the same reasoning as above we identify one large region of abberation
and obtain a confidence interval for the corresponding change-point as well as
for the signal's value.
Here, the optimized $q^* \approx 1.7$ in the sense of \eqref{sim:thresh:optq} gives $\alpha\approx 0.04$
which yields a SMUCE with one jump with high significamce.

\subsubsection{Photoemission Spectroscopy (PES)}

Electron emission from nanostructures triggered by ultrashort laser pulses has 
numerous applications in time-resolved electron imaging and spectroscopy \citep{RopSolSchLieEls07}.
 In addition, it holds promise for fundamental insight into electron correlations in microscopic volumes,
 including antibunching \citep{KieRenHas02}. Single-shot measurements of the number of electrons emitted 
per laser pulse \citep{BorGulWeiYalRop10, HerSolGulRop12} will allow for the disentanglement of various 
competing processes governing the electron statistics, such as classical fluctuations, Pauli blocking and space charge effects.

\begin{figure}[h!]
\includegraphics[width=0.9\columnwidth]{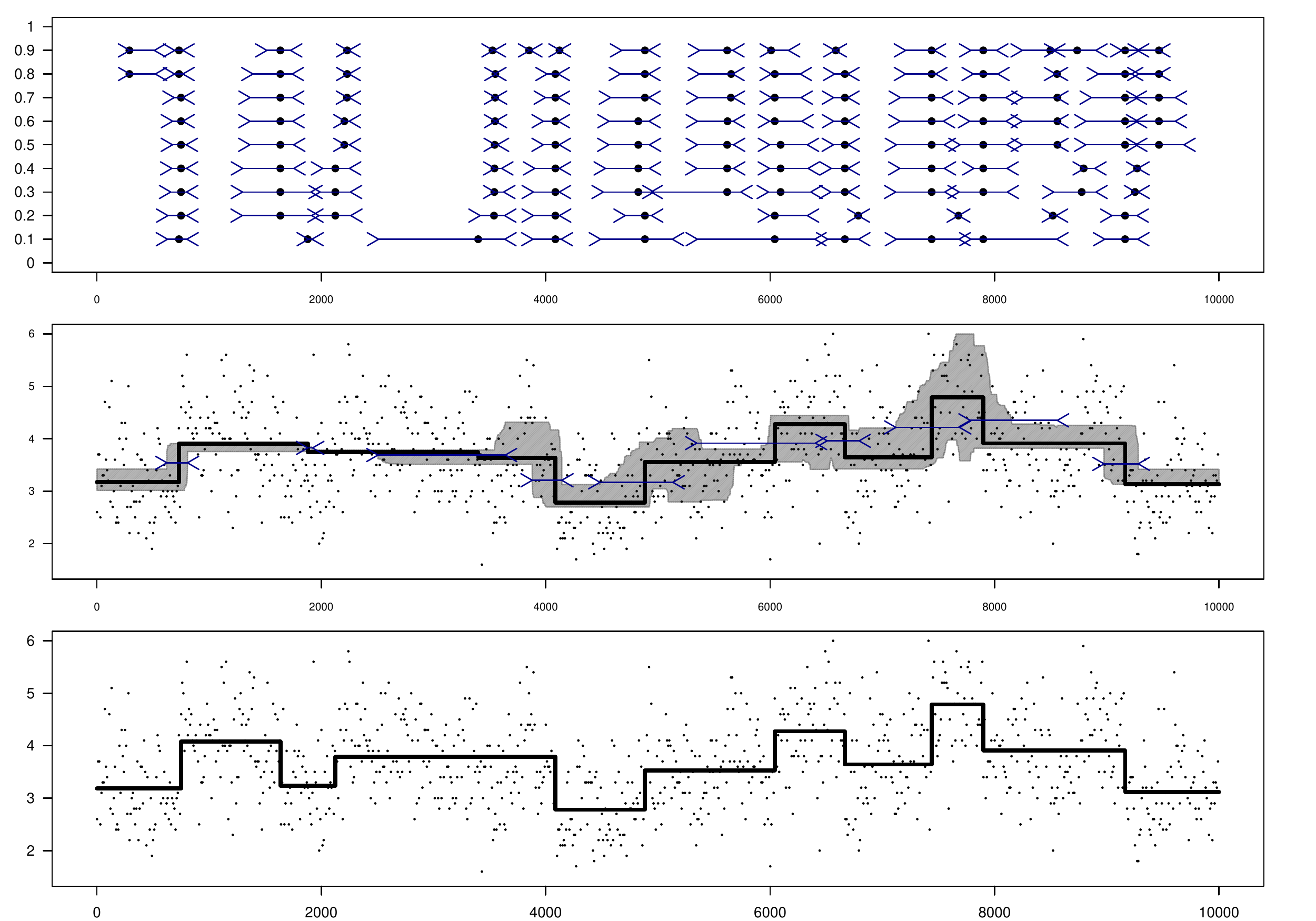}
\caption{Top: Detected change-points and confidence intervals for different values of $\alpha$ \ha{($y$-axis). Middle: SMUCE with confidence bands (grey hatched), confidence intervals for the changepoints (inwards pointing arrows) and binned PES data. Bottom: ML-Estimator
 with 10 change-points.}}
\label{sim:real:est}
\end{figure}

We investigate with the SMUCE approach PES data displayed in the bottom panel of Figure \ref{sim:real:est}.
It represents a time series of electron numbers recorded from a PES experiment performed in the Ropers lab
 (Department of Biophysics, University of Goettingen, see \citep{BorGulWeiYalRop10}).
It is custom to model PES data by Poisson regression with unknown intensity. This intensity is known
to show long term fluctuations which correspond to variation in laser power and laser beam pointing,
which cannot be controlled in the experiment and typically leads to an overall over-dispersion effect. However, on a short time scale, the interesting
task is to investigate underdispersion in the distribution.
Such underdispersion would indicate an electron interaction
 in which the emission of one (or a few) electrons decreases the likelihood of further emission events. Specifically,
 a significant underdispersion in the single-shot electron number histogram would evidence an anticorrelation caused by
 electrons being Fermions that obey the Pauli exclusion principle. A
piecewise constant mean that models sudden changes in the laser intensity to reflect the large scale fluctuations is used for segmentation of the data
 for further investigation of under- or overdispersion in these segments.

Figure \ref{sim:real:est} shows the estimated change-points of SMUCE (and
the corresponding confidence intervals) for $\alpha=0.05,0.1,\ldots,0.9$ in the top panel. 
We also display the SMUCE with confidence bands for $\alpha=0.9$ (middle) and for comparison the MLE
 with $\hat K_{\text{SMUCE}}(q)=10$ change-points (bottom). Note, that
the MLE is computed without the additional constraint $T_n(Y,\hat \mu)\leq q$, in contrast 
to \text{SMUCE}. Remarkably, this results in a different estimator.

We estimate the dispersion of data $Y_1,\ldots,Y_m$ by $\hat\rho=\hat \sigma^2$/$\hat\mu$, where $\hat \mu= 1/m \sum_{i=1}^{m} Y_i$
and $\hat\sigma = 1/m \sum_{i=1}^{m} (Y_i-\hat \mu)^2$.
In Table \ref{sim:nano:tab1} $\hat \mu= 1/m \sum_{i=1}^{m} Y_i$ is shown for the whole dataset as well as for the segments 
identified by SMUCE. It can be seen that our segmentation allows to explain the overall overdispersion to a large extent, by the long term fluctuations.
However, the results in Table \ref{sim:nano:tab1} do not indicate significant underdispersion on any of the identified segments.
This may be explained by a masking effect due to fluctuations of the emission current.
Future experiments using more stable emission currents are underway.

\begin{table}[h]
\tiny
\begin{center}
\begin{tabular}{r|r|rrrrrrrrrrr}
  \hline
segment & overall &1 & 2 & 3 & 4 & 5 & 6 & 7 & 8 & 9 & 10 & 11  \\ 
  \hline
 $\hat \rho$ &1.02 & 0.98 & 1.02 &{0.98} & 1.04 & 1.01 & 1.04 & {0.98} & 1.03 & {0.99} & 0.98 & 1.05  \\ 
   \hline
\end{tabular}
\caption{Dispersion estimator $\hat \rho$ of the whole dataset and on the segments identified by SMUCE }\label{sim:nano:tab1}
\end{center}
\end{table}

\section{Discussion}\label{sec:disc}

\subsection{Dependent Data}\label{subsec:diss:pen}

\kl{So far the theoretical justification for SMUCE relies on the
independence of the data in model \eqref{intro:model} (see Section \ref{sec:mrstat}), as for example the
optimal power results in Section \ref{subsec:gauss:van}. We claim, however, that
SMUCE as introduced in this paper can be extended to piecewise
constant regression problems with \emph{serially dependent} data. A comprehensive discussion is above the scope of this paper an will be addressed in future work. Here,  we confine ourselves to the case of a Gaussian moving average process of order 1, a similar strategy has been applied in \citep{Hot12} for $m$-dependent data.}
\begin{example}
\ha{For a piecewise constant function $\mu\in \mathcal S$ we consider the MA(1) model
\begin{equation*}
 Y_i = \mu (i/n) + \varepsilon_i + \beta \varepsilon_{i-1} \quad \text{ for } \quad  i=1,\ldots,n,
\end{equation*}
where $\beta<1$ and $\varepsilon_0,\varepsilon_1,\ldots, \varepsilon_{n} \overset{i.i.d.}{ \sim} \mathcal{N}(0,\sigma^2)$.
We aim to adapt the SMUCE to this situation. Following the local likelihood approach underlying the multiscale constraint in \eqref{intro:optprob} one simply might replace the local statistic $\sqrt{2 T_i^j(Y,\mu_0)}$ for $\mu_0\in\R$ in \eqref{intro:mrstat} by the (modified) local statistics
\begin{equation} \label{diss:dep:locstat}
\sqrt{2 \tilde{T}_i^j(Y,\mu_0) } = \frac{\abs{\sum_{l=i}^j Y_l-\mu_0}}{\sqrt{\sigma^2 \left[ (j-i+1) (1+\beta^2) + (j-i) \beta \right]}}.
\end{equation}
This is motivated by the fact, that $\text{Var}(\sum_{l=i}^j Y_l) =\sigma^2 \left[ (j-i+1) (1+\beta^2) + (j-i) \beta \right]$.
Under the null-hypothesis the local statistics $\tilde T_i^j$ then marginally have a $\chi^2_1$ distribution, as $T_i^j$ in \eqref{intro:likeratstat} for independent Gaussian observations.}

\ha{
In order to control the overestimation error as in Section \ref{sec:conscp}, one now has to compute
the null distribution of 
\begin{equation*}
 \tilde T_n(Y,\mu ) = \max_{\substack{1\leq i<j\leq n \\ \mu(t) = \mu_0 \text{ for } t
\in [i\slash n, j\slash n]}} \left(  \sqrt{2 \tilde T_i^j(Y,\mu_0)} -
\sqrt{2\log\frac{en}{j-i+1}} \right).
\end{equation*}
To this end, we used Monte-Carlo simulations for a sample size of $n=500$. We reconsider the test signal from Section \ref{subsec:gaussmean} with $\sigma=0.2$ and $a=0$.
The empirical null-distribution of $\tilde T_n$ and a probability-probability plot of the null distribution of $T_n$ against $\tilde T_n$ are shown in Figure \ref{disc:depend:nulldis}.
\begin{figure}
 \includegraphics[width= 0.6 \columnwidth]{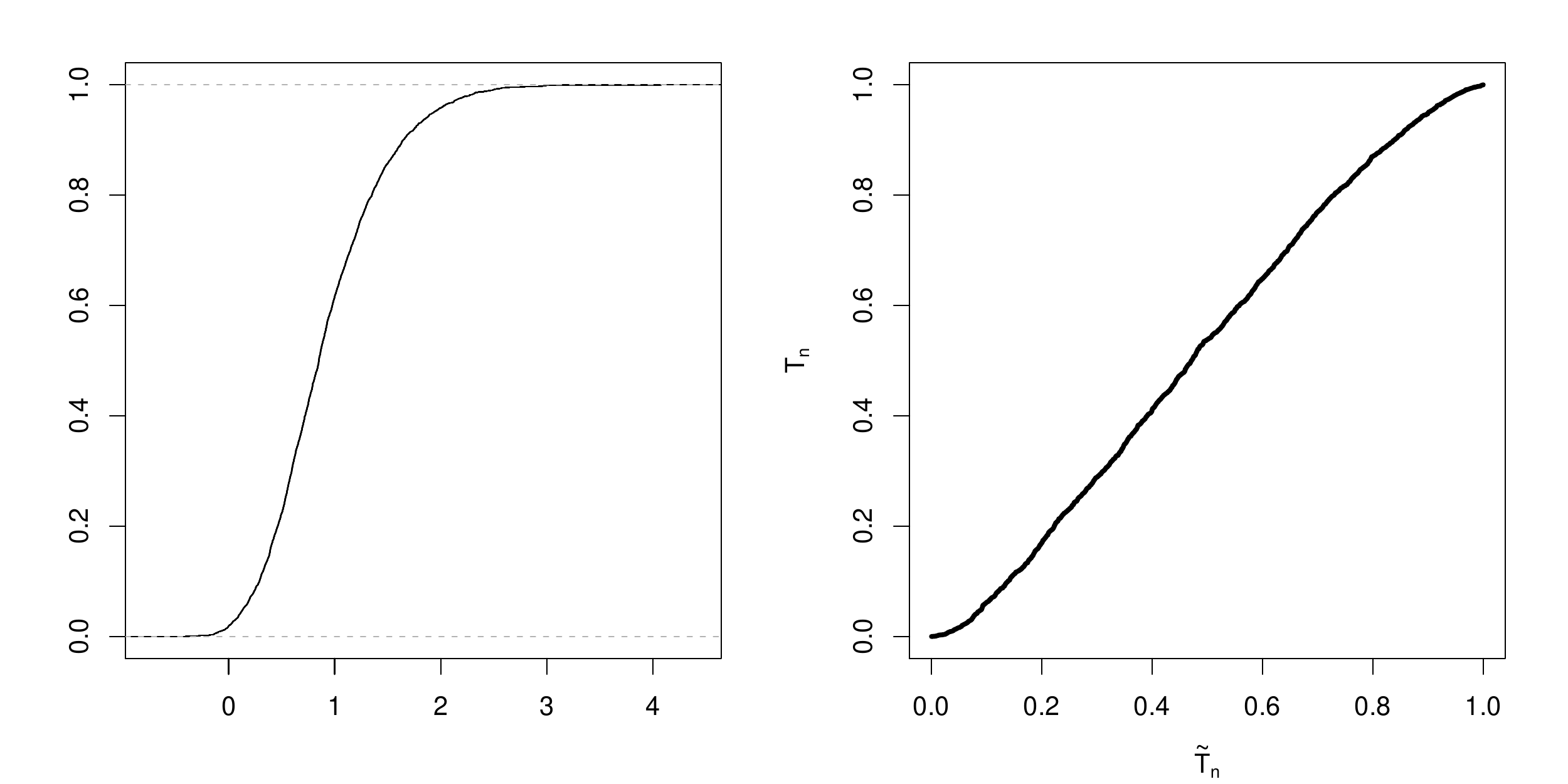}
 \caption{\ha{ecdf of the null distribution for dependent observations with $\beta=0.3$ and PP plot against the null distribution for independent observations.}}
 \label{disc:depend:nulldis}
\end{figure}
For $\beta=0.1$ and $\beta=0.3$, which corresponds to a correlation of $\rho=0.1$ and $\rho=0.27$, we ran $1000$ simulations each. We computed the modified SMUCE, as in \eqref{diss:dep:locstat}, and the SMUCE for independent Gaussian observations. For both procedures we chose $q$ to be the $0.75$-quantile of the null-distribution. The results are shown in Table \ref{diss:depend:sim}. For $\beta=0.1$ both procedures perform similarly, which indicates that SMUCE is robust to such weak dependences, while for $\beta=0.3$ the modified version performs much better w.r.t. the estimated number of change-points.
}
\begin{table}[ht]
\tiny
\centering
\begin{tabular}{lc||c>{\columncolor{tablebg}}cccc|cc}
  \hline
& $\beta$ & 5 & 6 & 7 & 8 & $\geq 9$ & MISE & MIAE \\ 
\hline
 modified SMUCE& 0.1 & 0.02 & \textbf{0.98} & 0.00 & 0.00 & 0.00 & {0.00154} & \textbf{0.02104} \\ 
 SMUCE& 0.1 & 0.00 & 0.95 & 0.04 & 0.00 & 0.00 & \textbf{0.00142} & 0.02117 \\ 
\hline
 modified SMUCE& 0.3 & 0.27 & \textbf{0.73} & 0.00 & 0.00 & 0.00 &{0.00435} & \textbf{0.03084} \\ 
 SMUCE& 0.3 & 0.00 & 0.29 & 0.34 & 0.24 & 0.13 &  \textbf{0.00277} & 0.03229 \\ 
   \hline
\end{tabular}
\caption{\ha{Frequencies of estimated number of change-points and MISE by model
selection for the modified SMUCE and SMUCE.}}
\label{diss:depend:sim}
\end{table}
\end{example}

\kl{The example illustrates that SMUCE as in \eqref{intro:optprob} can be
successfully applied to the case of dependent data after an adjustment of the underlying
multiscale statistic $T_n$ to the dependence structure. The asymptotic null-distribution of this
modified multiscale statistic is certainly not obvious and postponed to future work.}

\ha{
\subsection{Scale-calibration of $T_n$}\label{subsec:diss:pen}
The penalization of different scales as in \eqref{intro:mrstat} is borrowed from \citep{DueSpok01} and calibrates the number of intervals on a given scale. This prevents the small intervals to dominate the statistic. For this purpose, one might also consider the statistic
\begin{equation*}
T^1_n(Y,\vt) = \max_{\substack{1\leq i<j\leq n \\ \vt(t) = \theta \text{ for } t
\in [i\slash n, j\slash n]}} \left( \frac{  T_i^j(Y,\theta) -
2\log\frac{n}{j-i+1}} {\log\log \frac{e^e n}{j-i+1}} \right),
\end{equation*}
which is finite a.s. as $n\rightarrow \infty$ (see again \citep[Theorem 6.1]{DueSpok01} or \citep{SchMunDue11}).
A multiscale statistic without scale calibration
\begin{equation*}
T^2_n(Y,\vt) = \max_{\substack{1\leq i<j\leq n \\ \vt(t) = \theta \text{ for } t
\in [i\slash n, j\slash n]}} T_i^j(Y,\theta).
\end{equation*}
was e.g. considered in \citep{DavHoeKra12}. We illustrate the calibration effect of the statistics $T_n$, as in \eqref{intro:mrstat}, $T_n^1$ and $T_n^2$ in Figure \ref{dis:fig:diffpen}. The graphic shows the frequencies at which the corresponding $0.75$-quantiles of the statistics $T_n$, $T_n^1$ and $T_n^2$ is exceeded at a certain scale (scales are displayed on the $x$-axis). It can be seen, that $T_n^2$ puts much emphasis on small scales, while the penalized statistics $T_n$ and $T_n^1$ distribute the scales more uniformly.
For our purposes this calibration is beneficial in two ways: First it is required to obtain the optimal detection rates in Theorem \ref{gauss:opt:nojumpmissing} and Theorem \ref{gauss:opt:nojumpmissing} as it was shown in \citep{ChaWal11}. Second, the asymptotical behavior is determined by a process of the type \eqref{limit:mainres} and not by a extreme value limit as to be expected in the uncalibrated case, where the maximum is attained 
at scales of the magnitude $\log n$ with high probability (see \citep[Theorem 3.1 and the proof of Theorem 1.1]{KabMun08}) in accordance with Figure \ref{dis:fig:diffpen}.
}
\begin{figure}
 \includegraphics[width = 0.75 \columnwidth]{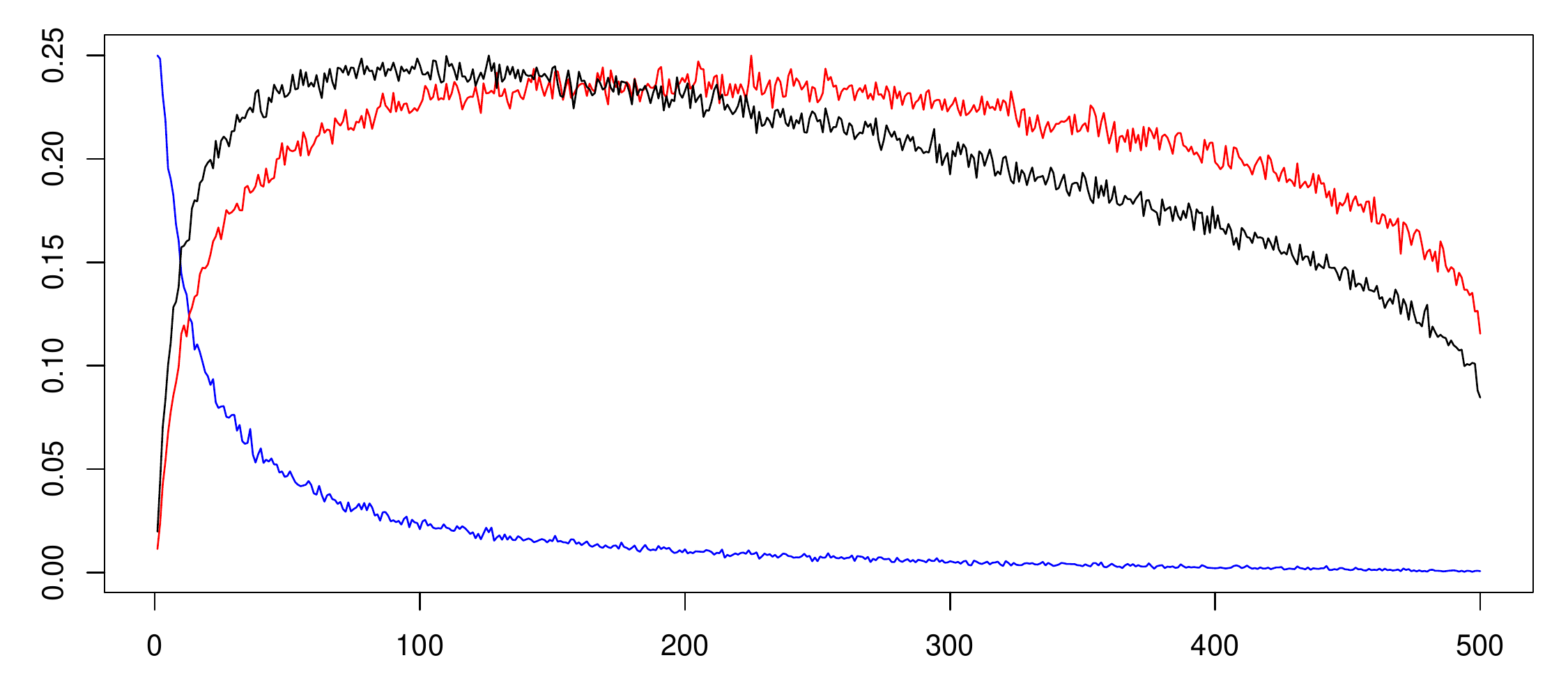}
 \caption{\ha{Frequencies of violations of the multiscale constraint for the different multiscale statistics $T_n$ (black), $T_n^1$ (red) and $T_n^2$ (blue)} obtained from 10.000 simulations on certain scales (scales are on the $x$-axis).}
 \label{dis:fig:diffpen}
\end{figure}

\subsection{SMUCE from a linear models perspective}
For normal mean regression one may rewrite the change-point regression model in \eqref{gauss:opt:model}
as a linear model
\begin{equation*}
 Y=X\beta + \sigma \epsilon,
\end{equation*}
where $\beta_i =\vartheta_i-\vartheta_{i-1}$
denotes the jump heights. If we add a vector of ones and a coefficient
$\beta_0$
to define the offset of the function, then
$X$ is an $(n\times n)$ upper triangular matrix with entries $X_{i,j}= 1$,
$i\geq j$ and zero else.
Hence, in the terminology of high dimensional linear models, we have an
``$n=p$'' problem in contrast to
the ``$p\gg n$'' situation which has perceived enormous attention during the last
two decades.
If we rescale by $1/\sqrt{n}$, then we find that $X^t X/n = \min (i,j)/n$ tends to
the covariance function of a standard Brownian motion.
From this limiting covariance it becomes immediately clear that
assumptions like the restricted isometry property and related conditions (see \citep{BueGee11,CanTao07,MeiYu09}) 
fail without additional restrictions, e.g. \am{an $s$-sparseness ($s\ll p$)} assumption on the jump locations. 
For a thorough discussion see \citep{BoyKemLieMunWit09} or the Appendix in \citep{HarLev10}.
Roughly speaking, these assumptions guarantee that estimators which are based on minimizing $\ell_0(\beta)$,
i.e. the number of jumps, can be obtained  by
the $\ell_1(\beta)$ surrogate with large probability. This is not the case in our set up when the number of jumps can be arbitrarily large.
This may be taken as a rough explanation for the empirical observations that $TV$ and $\ell_1$ penalization method do not perform 
competitive in the multiscale framework discussed in this paper for estimating location and number of change-points, as they built in too many little jumps. SMUCE employs a weaker notion of sparsity, i.e. $s=n=p$.
\subsection{Risk measures}
SMUCE aims to maximize the probability of correctly specifying the number of
jumps $P(\hat K = K)$
uniformly over sequences of models such that $\lambda_n \Delta_n^2$ tends to
zero not as fast as $\log n/n$.
This is conceptually very different from optimizing $\hat\vt$ w.r.t. convex risk measures such as the
mean squared error and related concepts.
The latter measures do not primarily target on the jump locations and number
of jumps. Therefore, we argue that in those applications, where the primary focus is on the jump locations  SMUCE may
be advantageous. In fact, maximizing the probability of correctly estimating the number of
jumps as SMUCE advocates has some analogy to risk measures for variable selection problems, shown to perform adequately successful in high
dimensional models. This includes the false discovery rate (FDR) \citep{BenHoc95} and related
ideas (see e.g. \citep{GenWas04}). Whereas in our context the latter ones aim
to minimize the expected relative number of wrongly selected change-points, SMUCE
is able to give at the same hand a guarantee that the true change-points will be detected
with large probability and hence controls the false acceptance rate (FAR) as well.

\subsection{Computational costs}
In \citep{KillFeaEck12} the authors showed that their pruned exact linear time 
method leads to an algorithm which expected complexity is linear in $n$ in some
cases.
As stressed in Section \ref{sec:alg}, our algorithm includes similar pruning steps.
Due to the complicated structure of the cost functional, however, it seems impossible
to prove such a result for the computation of SMUCE. The computation can, of
course, be further reduced significantly if e.g. only intervals of dyadic lengths 
are incorporated into the multiscale statistic.
Since the dynamic approach leads to a recursive computation, SMUCE can be updated in linear time,
 if applied to sequential data.
Another interesting strategy to reduce the computational costs could be adapted from \citep{Wal10,RivWal12}
 who suggest to restrict the multiscale constraint to a specific system of intervals of size $\bigo (n)$ which still
 guarantees optimal detection.

\subsection{The choice of $\alpha$}
We have offered a strategy to select the threshold $q=q_\alpha$ and hence the
confidence level $\alpha$
 in a sensible way to minimize $\Prob (\hat K \neq K)$, by balancing the
probabilities of over- and underestimation of $K$, simultaneously. This is based
on the inequalities in Section \ref{sim:thresh} depending on
$\lambda,\Delta$ and $n$. As indicated in Figures
\ref{intro:example},\ref{data:CGH:29:1},\ref{data:CGH:31:1}
and \ref{sim:real:est} this can be used to consider the evolution of SMUCE
depending on $\alpha$ as a universal ``objective'' smoothing parameter. The
features (jumps) of each SMUCE given $\alpha$ then may be regarded as ``present
with certain confidence'' similar in spirit to ideas underlying siZer (see
\citep{ChaMar99, ChaMar00}). It is striking that in many simulations we found
that features (jumps) remain persistent for a large range of levels $\alpha$. Of
course, other strategies to balance $\Prob(\hat K(q)>K)$ and $\Prob(\hat
K(q)<K)$ are of interest, e.g. if one of these probabilities is considered as
less important. For a first screening of jumps, $\Prob(\hat K(q)>K)$ is the less
serious error and $\Prob(\hat K(q)<K)$ should be minimized primarily. This can
be achieved by optimizing the convex combination $\delta \Prob(\hat K(q)>K) +
(1-\delta) \Prob(\hat K(q)<K)$ for a weight $\delta$ close to 1 along the lines
described in Section \ref{sim:thresh}.

\section*{Acknowledgments.} Klaus Frick, Axel Munk and Hannes Sieling were supported by DFG/SNF grant FOR 916. Axel Munk was also supported by CRC 803, CRC 755 and Volkswagen Foundation. This paper benefited from discussions with colleagues. We specifically would like to acknowledge L.D. Brown, T. Cai, L. Davies, L. D\"umbgen,
 E. George, C. Holmes, T. Hotz, S. Kou, O. Lepski, R.Samworth, D.Siegmund, A.Tsybakov and G. Walther. Various helpful comments and suggestions of screeners and reviewers of JRRS-B are gratefully acknowledged.

\bibliographystyle{chicago}
\bibliography{literature}

\newpage
\setcounter{page}{1}

\section{Supplement to ``Multiscale change-point inference''}\label{sec:proofs}

\maketitle

In this supplement we collect the proofs of the main assertions in the paper
together with some auxiliary lemmas. We further give more general versions of some results in the paper.  

\subsection{Large deviation and power estimates}

We begin by recalling some large deviation results for
exponential families. By $D(\theta||\tilde\theta)$ we will denote the
\emph{Kullback-Leibler divergence} of $F_\theta$ and $F_{\tilde\theta}$, i.e.
\begin{equation}\label{proofs:kl}
D(\theta||\tilde\theta) = \int_\R f_\theta(x)
\log\frac{f_\theta(x)}{f_{\tilde\theta}(x)} \diff\nu(x) = \psi(\tilde \theta) - \psi(\theta) - (\tilde\theta -
\theta)m(\theta). 
\end{equation}
With the techniques used in \citepsupp[Thm.7.1]{Bro86} it is
readily seen that for a sequence of independent and $F_\theta$-distributed r.v.
$Y_1,\ldots,Y_n$ one has that
\begin{equation}\label{proofs:largedev}
\Prob\left(\overline{Y} - m(\theta)\geq \eta\right) \leq e^{n (D(\theta||\theta
+ \eps) - \eta\eps)}
\end{equation}
for all $\eps>0$ such that $\theta+\eps\in \Theta$.
The following restatement of inequality \eqref{proofs:largedev} turns out to be
very useful.
\begin{lem}\label{proofscons:bound2}
Let $Y = (Y_1,\ldots,Y_n)$ be independent random variables such that $Y_i\sim
F_{\theta}$ and assume that $\delta > 0$ is such that $\theta + \delta \in
\Theta$. Then, 
\begin{equation*}
\Prob(m^{-1}(\overline{Y}) \geq \theta + \delta) \leq e^{-n D(\theta + \delta
|| \theta)}.
\end{equation*}
\end{lem} 
\begin{proof}
First observe that according to \eqref{proofs:largedev}
\begin{align*}
\Prob(m^{-1}(\overline{Y}) \geq \theta + \delta) & = \Prob(\overline Y -
m(\theta) \geq m(\theta + \delta) - m(\theta)) \\
& \leq \exp(n(D(\theta||\theta + \delta) - (m(\theta + \delta) -
m(\theta))\delta)).
\end{align*}
Now it follows from \eqref{proofs:kl} that
\begin{align*}
D(\theta||\theta + \delta) - (m(\theta + \delta) - m(\theta))\delta & =
\psi(\theta+\delta) - \psi(\theta) - m(\theta + \delta)\delta \\
& = -(\psi(\theta) - \psi(\theta + \delta) - (\theta - (\theta +
\delta))m(\theta + \delta)) \\
& = -D(\theta + \delta||\theta). 
\end{align*}
\end{proof}
From \eqref{proofs:largedev} we further derive a basic power estimate for the
likelihood ratio statistic
\eqref{intro:likeratstat}.

\begin{lem}\label{proofs:power}
Let $Y = (Y_1,\ldots,Y_n)$ be independent random variables such that $Y_i\sim
F_{\theta}$ and assume that $\delta \in \R$ is such that $\theta + \delta \in
\Theta$. Then, 
\begin{equation*}
\Prob\left( T_1^n(Y, \theta + \delta) \geq q \right) \geq 1 - \exp\left(
n  \inf_{\eps\in[0,\delta]} \left[
 D(\theta||\theta + \eps)-
 \frac{\eps}{\delta} D(\theta||\theta + \delta)+\frac{\eps q}{n \delta}  \right]
 \right).
\end{equation*}
\end{lem}

\begin{proof}
For
\begin{align*}
J(\overline Y, \theta) & = \phi(\overline Y) - \left( \overline
Y \theta - \psi(\theta) \right)
\end{align*}
we obtain
\begin{equation}\label{proofs:consaux1}
J(\overline Y, \theta + \delta) = J(\overline Y, \theta) -
\delta \overline Y - \psi(\theta) + \psi(\theta + \delta). 
\end{equation}
Thus, we have
\begin{align*}
\Pi(q,n,\delta) & := \Prob\bigl(T_1^n(Y, \theta + \delta) 
 \geq q\bigr)\\ & = \Prob\left( J(\overline
Y, \theta + \delta) \geq \frac{q }{n}
\right) \\
& = \Prob\left( J(\overline Y, \theta) -\delta \overline Y
  \geq \frac{q}{n} -
\psi(\theta + \delta) + \psi(\theta)\right) \\
& \geq \Prob\left( -\delta \overline Y
  \geq \frac{q}{n} -
\psi(\theta + \delta) + \psi(\theta)\right),
\end{align*}
where in the last inequality holds since $J(x,\theta)\geq 0$
for all $x\in\R$ and $\theta\in\Theta$. Now, let us first assume that $\delta >
0$. Then by \eqref{proofs:kl} we find
\begin{align}\label{proofs:consaux2}
\Prob\left( -  \delta\overline Y  \geq \frac{q}{n} - \psi(\theta + \delta) + \psi(\theta)\right) 
  & = \Prob\left( \overline Y - m(\theta)\leq - \frac{q}{\delta n } +
  \frac{D(\theta||\theta + \delta)}{\delta}\right).
\end{align}
Combining this
with the large deviation inequality \eqref{proofs:largedev} yields
\begin{align*}
\Pi(q,n,\delta) & \geq 1-\exp\left(n( D(\theta||\theta + \eps) -
{\eps\over\delta}D(\theta||\theta + \delta)) + {\eps q \over \delta} \right),
\end{align*}
for all $0\leq \eps\leq \delta$. The case when $\delta < 0$ follows analogously.
\end{proof}

For Gaussian observations the estimate can be made explicit.

\begin{lem}\label{proofs:power:gauss}
Let $Y_1,\ldots,Y_n$ be i.i.d. random variables such that $Y_1\sim
\mathcal{N}(0,1)$ and let $x_+ = \max(0,x)$ for $x \in \R$. Then,
\begin{equation}\label{proofs:gauss:handyest}
\Prob\left( T_1^n(Y,\delta) \geq q \right) \geq 1 -
\exp\left(-\frac{1}{8}\left(\sqrt{n}\delta - \sqrt{2q}\right)_+^2\right).
\end{equation}
\end{lem}

\begin{proof}
Since $D(\theta||\theta + \eps) = \eps^2\slash 2$ we find that
\begin{equation*}
\inf_ {\eps\in[0,\delta]} n \left[ D(\theta||\theta + \eps)-
 \frac{\eps}{\delta} D(\theta||\theta + \delta)+\frac{\eps q}{n \delta}  \right]
 =  - \frac{1}{2} \left( \frac{\delta\sqrt n}{2}-\frac{q}{\delta \sqrt n}
 \right)^2 \leq -\frac{1}{8}\left(\sqrt{n}\delta - \sqrt{2q}\right)^2,
\end{equation*}
if $\sqrt{n}\delta \geq \sqrt{2q}$.
\end{proof}

\subsection{Proof of Theorem \ref{limit:mainthm}}

Throughout this section we will assume that $Y=(Y_1,\ldots,Y_n)$ are
independent and identically distributed random variables with $Y_1\sim F_\theta$
and $\theta\in \Theta$. Without loss of generality we will assume that
$m(\theta) = \dot\psi(\theta) = 0$ and $v(\theta) = \ddot\psi(\theta) =
1$. Moreover, assume that $(c_n)_{n\in\N}$ satisfies \eqref{def:lowerbound} and
introduce $\Ic= \{ (i,j) : j-i+1 \geq  c_n n\}$. We start with some
approximation results for the extreme value statistic of the partial sums $\overline Y_i^j$.

\begin{lem}\label{proofs:kmt}
There exist i.i.d standard normally distributed r.v. $Z_1,\ldots,Z_n$ on the
same probability space as $Y_1,\dots,Y_n$ such that
\begin{equation*}
\lim_{n\ra\infty}\sqrt{\log n}\max_{(i,j)\in \Ic}
\left(\sqrt{j-i+1}\abs{\bigl|\overline Y_i^j\bigr| -\bigl|\overline Z_i^j\bigr|}\right) = 0 \quad \text{a.s.}
\end{equation*}
\end{lem}
\begin{proof}
We define the partial sums $S_0^Y= 0$ and $S_l^Y = Y_1+\ldots+ Y_l$
 and observe that $(j-i+1)\bigl|\overline Y_i^j\bigr| = \abs{S_j^Y -
S_{i-1}^Y}$. Analogously we define $S_l^Z$. Now let $(i,j)$ such
that $j-i+1\geq n c_n$ and observe that
\begin{equation*}
\abs{\frac{\abs{S^Y_j - S^Y_{i-1}}}{\sqrt{j-i+1}} - \frac{\abs{S^Z_j -
S^Z_{i-1}}}{\sqrt{j-i+1}}}  
\leq  \frac{\abs{S^Y_j - S^Z_j}}{\sqrt{n c_n}} + \frac{\abs{S^Y_i -
S^Z_i}}{\sqrt{n c_n}} \leq   2\max_{0\leq l \leq n} \frac{\abs{S^Y_l -
S^Z_l}}{\sqrt{nc_n}}.
\end{equation*}
It follows from the KMT inequality \citepsupp[Thm.
1]{KomMajTus76} and \eqref{def:lowerbound} that
\begin{equation*}
\sqrt{\log n}\max_{0\leq l \leq n} \frac{\abs{S^Y_l -
S^Z_l}}{\sqrt{n c_n}} =  \smallo(1) \quad \text{a.s.}
\end{equation*}
\end{proof}

\begin{lem}\label{proofs:taylor}
\begin{equation*}
 \max_{(i,j)\in \Ic}
 \abs{ \sqrt{2T_i^j(Y,\theta)} - \sqrt{j-i+1} \bigl| \overline Y_i^j
 \bigr|} = \smallo_{\Prob}(1)
\end{equation*}
\end{lem}

\begin{proof}
Set $\xi = m^{-1}$ and note that $\xi$ is strictly increasing. Since
$\Theta$ is open, there exists for each given $\delta'>0$ a $\delta > 0$ such
that $\xi(B_\delta(0))\subset B_{\delta'}(\theta) \subset \Theta$. Next define
the random variable 
\begin{equation*}
L_n = \max_{1\leq i< j \leq n}\abs{\overline Y_i^j}\sqrt{j-i+1}.
\end{equation*}
Then it follows from Shao's Theorem \citepsupp{Sha95} that $L_n\slash \sqrt{\log n}$
converges a.s. to some finite constant and we hence find that
\begin{equation*}
 \max_{(i,j)\in \Ic}
\abs{\overline Y_i^j}   \leq \sqrt{\frac{\log n}{n c_n}} 
\frac{L_n}{\sqrt{\log n}} \ra 0\quad\text{ a.s.}
\end{equation*}
Thus, for each $\eps > 0$ there
exists an index $n_0 = n_0(\eps)\in\N$ such that for all $n\geq n_0$
\begin{equation*}
\Prob\left(\max_{(i,j) \in \substack{\Ic}}
\abs{\overline Y_i^j} \geq \delta\right) \leq \eps. 
\end{equation*}
In other words, $\xi(\overline Y_i^j) \in B_\ha{\delta}(\theta)$ uniformly over $\Ic$
with probability not less than $1-\eps$. Consequently, $\phi(\overline Y_i^j) =
\max_{\theta\in\Theta} \theta \overline Y_i^j - \psi(\theta) = \xi(\overline Y_i^j) \overline
Y_i^j - \psi(\xi(\overline Y_i^j))$ which in turn implies that
\begin{equation*}
J(\overline Y_i^j,\theta)  = \phi(\overline Y_i^j) - \theta \overline Y_i^j +
\psi(\theta) 
 = (\xi(\overline Y_i^j) - \theta)\overline Y_i^j - (\psi(\xi(\overline Y_i^j)) -
 \psi(\theta) ).
\end{equation*}
Taylor expansion of $\psi$ around $\theta$ gives (recall that
$\dot\psi(\theta) = 0$ and $\ddot\psi(\theta) = 1$)
\begin{equation*}
\psi(\xi(\overline Y_i^j)) - \psi(\theta) =  \frac{1}{2}(\xi(\overline Y_i^j) - \theta)^2 
+ \frac{1}{6} \dddot\psi(\tilde\theta)(\xi(\overline Y_i^j) - \theta)^3
\end{equation*}
for some $\tilde\theta \in B_\eps(\theta)$. This implies
\begin{equation*}
J(\overline Y_i^j,\theta) = (\xi(\overline Y_i^j) - \theta)(\overline Y_i^j)
- \frac{1}{2}(\xi(\overline Y_i^j) - \theta)^2  - \frac{1}{6}
\dddot\psi(\tilde\theta)(\xi(\overline Y_i^j) - \theta)^3.
\end{equation*}
Again, Taylor expansion of $\xi = m^{-1}$ around $0$ shows
\begin{equation*}
\xi(\overline Y_i^j) - \theta = \overline Y_i^j - \frac{\dddot
\psi(\tilde \theta)}{2(v(\tilde \theta))^2}(\overline Y_i^j)^2
\end{equation*}
for some  $\tilde\theta \in B_{\delta'}(\theta)$. This finally proves
that 
\begin{equation*}
2T_i^j(Y,\theta) = (j-i+1) J(\overline Y_i^j, \theta) = (j-i+1)(\overline
Y_i^j)^2 + (j-i+1)f_n(\overline Y_i^j)
\end{equation*}
where $f_n$ is such that $\bigl|f_n(\overline Y_i^j)\bigr|\leq C^2
\cdot(\overline Y_i^j)^3$ for a constant $C = C(\delta')>0$ (independent of
$\eps$, $i$ and $j$) and for all $n\geq n_0$. It thus holds with probability not
less than $1-\eps$ that
\begin{align*}
 \max_{(i,j)\in \Ic} \abs{\sqrt{2T_i^j(Y,\theta^*)} - \sqrt{j-i+1}\bigl|
\overline Y_i^j\bigr|} 
\leq & C \max_{(i,j)\in \Ic} \abs{(j-i+1)
\left(\overline Y_i^j\right)^3}^{1\slash 2}\\
= & C \max_{(i,j)\in \Ic} \abs{
\frac{\sum_{l=i}^j Y_l}{\sqrt{j-i+1} }(j-i+1)^{-1/6}}^{3\slash 2}\\
\leq & C \left(\frac{L_n}{\sqrt{\log n }} \right)^{3\slash 2}
\sqrt[4]{\frac{\log^{3}n}{{nc_n}}}.
\end{align*}
From Shao's Theorem it follows that the last term vanishes almost surely
as $n\ra\infty$.
\end{proof}

Combination of Lemma \ref{proofs:kmt} and \ref{proofs:taylor} yields

\begin{prop}\label{proofs:standard}
There exist i.i.d standard normally distributed r.v. $Z_1,\ldots,Z_n$ on the
same probability space as $Y_1,\dots,Y_n$ such that
\begin{equation*}
 \max_{(i,j)\in \Ic} \abs{ \sqrt{2T_i^j(Y,\theta)} - \sqrt{j-i+1} \bigl|\bar
 Z_i^j\bigr|} = \smallo_{\Prob}(1).
\end{equation*} 
\end{prop}


\begin{lem}\label{proofs:contfunc}
For $n\in\N$, define the continuous functionals
$h,h_n:\mathcal{C}([0,1])\ra \R$  by
\begin{eqnarray*}
h(x,c) & = & \sup_{\substack{0 \leq s < t \leq 1 \\ t-s\geq c}}
\left(\frac{\abs{x(t) -
x(s)}}{\sqrt{t-s}}-\sqrt{2\log\frac{e}{t-s}}\right)\quad\text{ and } \\ h_n(x,c)
& = &\max_{\substack{1\leq i< j\leq n \\ (j-i+1)\slash n \geq
c}}\left(\frac{\abs{x(j\slash n) - x(i\slash n)}}{\sqrt{(j-i+1)\slash n}} -
\sqrt{2\ha{\log}\frac{en}{j-i+1}}\right),
\end{eqnarray*}
respectively. Moreover assume that $\set{x_n}_{n\in\N}\subset
\mathcal{C}([0,1])$ is such that $x_n \ra x$ for some $x\in
\mathcal{C}([0,1])$. Then $h_n(x_n,c) \ra h(x,c)$.
\end{lem}

\begin{proof}
Let $\delta > 0$. Then
there exists an index $n_0\in\N$ such that $\abs{x_n(t) - x(t)}\leq \delta$ for
all $n\geq n_0$ and $t\in [0,1]$. Thus, it follows directly from the definition that $h_n(x) = h_n(x_n) + \bigo(\delta)$
for $n\geq n_0$. Since $u\mapsto \sqrt{2\log e\slash u}$ is 
uniformly continuous on $[c, 1]$ we consequently have that $h_n(x) \ra h(x)$
as $n\ra\infty$ and the assertion follows.
\end{proof}
 
Before we proceed, recall the definition of $M$ in \eqref{limit:limitstat}.
Moreover, we introduce for $0<c\leq 1$ the statistic
\begin{equation}\label{proofs:limitstatbnd}
M(c) := \sup_{\substack{0\leq s < t \leq 1 \\ t-s > c}} \left(\frac{\abs{B(t) -
B(s)}}{\sqrt{t-s}} - \sqrt{2\log\frac{e}{t-s}}\right).
\end{equation}
From \citepsupp[Thm. 6.1]{DueSpok01supp} (and the subsequent Remark 1) it can be
seen that $M(c)$ converges weakly to $M$ as $c\ra 0^+$.

\begin{prop}\label{proofs:largeintsconv} 
Let $c>0$ and define
\begin{equation*}
T_n^c(Y,\theta) = \max_{(i,j)\in \mathcal{I}(c)} \left(
\sqrt{2T_i^j(Y,\theta)} - \sqrt{2\log\frac{en}{j-i+1}}\right).
\end{equation*}
Then $\lim_{c\ra 0^+} \lim_{n\ra\infty} T_n^c(Y,\theta) = M$, weakly.
\end{prop}

\begin{proof}
Set $S_0 = 0$ and $S_n=Y_1+\ldots+Y_n$ and let $\set{X_n(t)}_{t\geq 0}$ be the
process that is linear on the intervals $[i\slash n, (i+1)\slash n]$ with  
values $X_n(i\slash n) = S_i\slash \sqrt{n}$. We obtain from Donsker's
Theorem that $X_n \stackrel{\mathcal{D}}{\ra} B$. Now, recall the definition of
$h$ and $h_n$ in Lemma \ref{proofs:contfunc} and observe that
\begin{equation*}
h_n(X_n,c) = \max_{(i,j)\in \mathcal{I}(c)}\left( \sqrt{j-i+1}\bigl|\overline
Y_i^j\bigr|-\sqrt{2\log\frac{en}{j-i+1}}\right).
\end{equation*}
It hence follows from Lemma \ref{proofs:taylor} that 
\begin{equation}\label{proofs:aux1}
  \abs{T_n^c(Y,\theta) - h_n(X_n,c)} \leq \max_{(i,j)\in \mathcal{I}(c)} \abs{\sqrt{2T_i^j(Y,\theta)} -
  \sqrt{j-i+1}\bigl|\overline Y_i^j \bigr|} =\smallo_{\Prob}(1).
\end{equation}
 Since $X_n \stackrel{\mathcal{D}}{\ra} B$, Lemma \ref{proofs:contfunc} and
 \citepsupp[Thm. 5.5]{Bil68} imply that $h_n(X_n,c)\stackrel{\mathcal{D},}{\ra}
h(B,c)$. Theorem 4.1 in \citep{Bil68} and \eqref{proofs:aux1} thus imply
that $T_n^c(Y,\theta)\stackrel{\mathcal{D}}{\ra} h(B,c) = M(c)$ as
$n\ra\infty$ for all $c>0$. Thus, the assertion finally follows, since $M(c)\ra
M$ weakly as $c\ra 0^+$
\end{proof}

\begin{thm}\label{proofs:limitiidcase}
Let $Y = (Y_1,\ldots,Y_n)$ be independent and identically distributed random
variables with distribution $F_\theta$, $\theta\in\Theta$. Moreover, assume that
$\set{c_n}_{n\in\N}$ is a sequence of positive numbers such that 
$n^{-1}\log^3 n\slash c_n\ra 0$ and set
\begin{equation*}
T_n(Y,\theta,c_n) = \max_{(i,j)\in \Ic} \left( \sqrt{2 T_i^j(Y,\theta)} -\sqrt{2
\log\frac{en}{j-i+1}}\right).
\end{equation*} 
Then, $T_n(Y,\theta,c_n)\ra M$ weakly as $n\ra\infty$. 
\end{thm}
 
\begin{proof}
First observe that according to Proposition \ref{proofs:standard} we have for
all $t > 0$ that
\begin{equation*}
\begin{split}
\Prob\left(T_n(Y,\theta;c_n) \leq t \right) & = 
\Prob\left( \max_{(i,j)\in\Ic} \left(\sqrt{j-i+1}
 \bigl|\overline Z_i^j\bigr| - \sqrt{2\log\frac{en}{j-i+1}}\right) \leq
 t\right) + \smallo(1)\\ & \geq \Prob\left(  \sup_{0\leq s < t
 \leq 1} \left(\frac{\abs{B(t) - B(s)}}{\sqrt{t-s}}
 - \sqrt{2\log\frac{e}{t-s}}\right) \leq t\right) + \smallo(1)
\end{split}
\end{equation*}
This shows that for all $t>0$
\begin{equation*}
\liminf_{n\ra\infty} \Prob(T_n(Y,\theta,c_n)\leq t) \geq
\Prob(M\leq t)
\end{equation*}
Now let $c>0$ be fixed and assume w.l.o.g. $c_n<c$ for all $n\in\N$. With
$T_n^c$ as defined in Proposition \ref{proofs:largeintsconv} we
conversely find
\begin{equation*}
\limsup_{n\ra\infty} \Prob(T_n(Y,\theta,c_n)\leq t) \leq
\limsup_{n\ra\infty}\Prob(T_n^c(Y,\theta,c_n)\leq
t) = \Prob(M(c)\leq t).
\end{equation*}
Hence the assertion follows from Proposition \ref{proofs:largeintsconv} after
letting $c\ra 0^+$ and the fact that  $M > 0$ a.s.
\end{proof}
 
\begin{proof}[Proof of Theorem \ref{limit:mainthm}]
Let $T_n(Y,\vt; c_n)$ be defined as in \eqref{smre:mrstat}. From Theorem
\ref{proofs:limitiidcase} it then follows that
\begin{equation*}
T_n(Y,\vt;c_n) \stackrel{\mathcal{D}}{\ra} \max_{0\leq k\leq K }\sup_{\tau_k\leq
s < t \leq \tau_{k+1}} \left(\frac{\abs{B(t) - B(s)}}{\sqrt{t-s}} -
\sqrt{2\log\frac{e}{t-s}}\right).
\end{equation*}
Clearly the limiting statistic on the right hand side is stochastically bounded
from above by $M$. Conversely, we
observe by the scaling property of the Brownian motion that
\begin{multline*}
\sup_{\tau_k\leq s < t \leq \tau_{k+1}} \left(\frac{\abs{B(t) -
B(s)}}{\sqrt{t-s}} - \sqrt{2\log\frac{e}{t-s}}\right) \\
\stackrel{\mathcal{D}}{=} \sup_{0\leq s < t \leq 1} \left(\frac{\abs{B(t) -
B(s)}}{\sqrt{t-s}} - \sqrt{2\log\frac{e}{t-s} + 2\log\frac{1}{\tau_{k+1} -
\tau_k}}\right) \stackrel{\mathcal{D}}{\geq} M - \sqrt{2\log\frac{1}{\tau_{k+1}
- \tau_k}}.
\end{multline*}
\end{proof}

\subsection{A general exponential inequality}

In this section we give a general exponential inequality for the probability
that SMUCE underestimates the number of change-points. To this end, we will
make use of the functions
\begin{align}
\kappa_1^\pm(v,w,x,y)  & = \inf_{\substack{v\leq\theta\leq w \\ \theta \pm x \in
[v,w]}} \sup_{\eps\in[0,x]} \left[\frac{\eps}{x} \left( D(\theta||\theta \pm x )
-y\right) - D(\theta||\theta \pm \eps) \right], \label{conscp:kappa1}\\
\kappa_2^\pm(v,w,x) & = \inf_{\substack{v\leq\theta\leq w \\ \theta \pm x \in
[v,w]}} D(\theta\pm x||\theta)\label{conscp:kappa2}.
\end{align}

\begin{thm}\label{conscp:thmunderest}
Let $q\in \R$ and $\hat K(q)$ be defined as in \eqref{conscp:estnocp}.
Moreover, assume that $\kappa_1^\pm$ and $\kappa_2^\pm$ are defined as in
\eqref{conscp:kappa1} and \eqref{conscp:kappa2}, respectively and set
\begin{align*}
\kappa_1 & = \min\set{\kappa_1^+\left(\underline\theta, \overline \theta,
\frac{\Delta}{2},\frac{\left(q+\sqrt{2\log\frac{2e}{\lambda}}\right)^2}{n\lambda}\right),\kappa_1^-\left(\underline\theta,
\overline \theta, \frac{\Delta}{2},\frac{\left(q+\sqrt{2\log\frac{2e}{\lambda}}\right)^2}{n\lambda}\right)}\;\text{ and }\\
\kappa_2 & = \min\set{\kappa_2^+\left(\underline\theta, \overline \theta,
\frac{\Delta}{2}\right),\kappa_2^-\left(\underline\theta, \overline \theta,
\frac{\Delta}{2}\right)}.
\end{align*}
If $\lambda \geq 2 c_n$, then 
\begin{equation}
\Prob\left(\hat K(q) < K\right) \leq 2 K \left[e^{- \frac{n\lambda\kappa_1 }{2}
} + e^{-\frac{n\lambda\kappa_2}{2}}\right].
\end{equation}
\end{thm}

\begin{proof}
Let $\Delta$ and $\lambda$ be the smallest jump size and the smallest interval
length of the true regression function $\vartheta$, i.e.
\begin{equation*}
\Delta = \inf_{1\leq k\leq K} \abs{\theta_k - \theta_{k-1}}\quad \text{ and }\quad \lambda =
\inf_{0\leq k\leq K} \tau_{k+1} - \tau_k.
\end{equation*}
Now define $K$ disjoint intervals
$I_i=\left(\tau_i-\lambda/2,\tau_i+\lambda/2\right)\subset [0,1]$. Let
$\theta_i^+=\max\set{\theta_{i-1},\theta_i}$,
$\theta_i^-=\min\set{\theta_{i-1},\theta_i}$ and split each interval $I_i$ 
accordingly, i.e. $I_i^+=\{t\in I_i: \vartheta(t)=\theta_i^+\}$  and
$I_i^-=\{t\in I_i: \vartheta(t)=\theta_i^-\}$. Clearly $I_i=I_i^- \cup I_i^+$.

From the definition of the estimator $\hat K(q)$ it is clear that
\begin{equation*}
\hat K(q)<K\quad\Leftrightarrow\quad \exists \hat\vt \in
\S_n[K-1]\text{ such that }T_n(Y,\hat\vt)\leq q.
\end{equation*}
If $\hat \vartheta\in \S_n[K-1]$, then there exists an
index $k\in\set{1,\ldots,K}$ such that $\hat \vartheta$ is constant on \ha{$I_k$}.
Let $\Omega_k=\left\{\exists \hat\theta \in
 \Theta: \sqrt{T_{I_k^+}(Y,\hat\theta)} - \sqrt{\log\frac{en}{\# I_k^+}}\leq \frac{q}{\sqrt 2} \text{ and }
 \sqrt{T_{I_k^-}(Y,\hat\theta)} - \sqrt{\log\frac{en}{\# I_k^-}}\leq \frac{q}{\sqrt 2} \right\}$
Since the $K$ intervals $I_i$ are disjoint we find
\begin{align*}
 \Prob(\hat K(q) < K) & \leq  \sum_{k=1}^K \Prob\left(\Omega_k\right).
\end{align*}
If $\hat \vt\in \S_n[K-1]$ is constant on some $I_k$ with value $\hat
\theta$, then  either $\hat \theta \leq \theta_k^+ -
\Delta/2$ or $\hat \theta \geq \theta_k^- + \Delta/2$, by construction.
Set
\begin{align*}
 \Omega_k^{+}&=\left\{ \exists \hat\theta\leq \theta_k^{+} - \Delta/2 : \sqrt{T_{I_k^+}(Y,\hat\theta)} - \sqrt{\log\frac{en}{\# I_k^+}}\leq \frac{q}{\sqrt 2} \right\}\ \\
 \Omega_k^{-}&=\left\{ \exists \hat\theta\geq \theta_k^{-} + \Delta/2 
: \sqrt{T_{I_k^-}(Y,\hat\theta)} - \sqrt{\log\frac{en}{\# I_k^-}}\leq \frac{q}{\sqrt 2} \right\}
\end{align*}
and observe that $\Prob(\Omega_k)\leq \Prob(\Omega_k^+)+\Prob(\Omega_k^-)$.
We proof an upper bound for $\Prob(\Omega_k^-)$, the same bound can be obtained for $\Prob(\Omega_k^+)$ analogously.
Recall that $\theta\mapsto T_{I_k^-}(Y,\cdot)$ is
convex and has its minimum at $m^{-1}(\overline Y_{{I_k^-}})$. Thus,
$T_{I_k^-}(Y,\hat \theta)\geq T_{I_k^-}(Y,\theta^-_k + \Delta/2)$ whenever
$m^{-1}(\overline Y_{{I_k^-}})\leq \theta^-_k + \Delta/2$. This yields
\begin{align*} 
\Prob \left(\Omega_k^-\right) & \leq \Prob\left(\Omega_k^- \cap
\left\{ m^{-1}(\overline Y_{I_k^-}) \leq 
\theta_k^- + {\Delta\over 2} \right\}\right) + \Prob\left( m^{-1}(\overline
Y_{I_k^-}) > \theta_k^- + {\Delta\over 2} \right)\\ & \leq  1-
\Prob\left(T_{I_k^-}\left(Y,\theta_k^- + {\Delta\over 2}\right)\geq 1/2\left( q + \sqrt{2 \log (2e/ \lambda)} \right)^2 \right )  + \Prob\left( m^{-1}(\overline
Y_{I_k^-}) > \theta_k^- + {\Delta\over 2} \right) \\ 
& \leq \exp \left( \frac{\lambda n}{2} 
\inf_{\eps\in [0,\Delta\slash 2]}\left( D(\theta_k^-||\theta_k^- +
\eps) - \frac{\eps}{\Delta/ 2}D(\theta_k^-||\theta_k^- + \Delta / 2)  +
\frac{2\eps \left( q + \sqrt{2 \log (2e/ \lambda)} \right)^2}{\Delta \lambda n} \right) \right) \\ &
\hspace{0.05\textwidth} + \exp\left( - \frac{\lambda n}{2} D(\theta_k^- + \Delta \slash 2 || \theta_k^-) \right) \\
& \leq \exp\left(-\frac{n\lambda}{2}
\kappa_1^+\left(\underline\theta, \overline\theta,\frac{\Delta}{2},
\frac{\left( q + \sqrt{2 \log (2e/ \lambda)} \right)^2}{\lambda n}\right)\right)
+ \exp\left(-\frac{n\lambda}{2}\kappa_2^+\left(\underline\theta,
 \overline\theta,\frac{\Delta}{2}\right) \right)
\end{align*}
by Lemma \ref{proofscons:bound2} and Lemma \ref{proofs:power}. With the
definition of the constants $\kappa_j$ as in the Theorem ($j=1,2$) we eventually obtain
\begin{equation*}
 \Prob(\hat K(q) < K)\leq 2 K\left[ \exp\left(-\frac{n\lambda\kappa_1}{2}\right) +
 \exp\left(-\frac{n\lambda\kappa_2}{2}\right)\right].
\end{equation*}
\end{proof}
  
The constants $\kappa_i^\pm$ ($i=1,2$) basically depend on the exponential
family $\mathcal{F}$. Their explicit computation can be rather tedious and has
to be done for each exponential family separately (for the Gaussian case see
see below). Therefore, it is useful to have a lower bound for these constants.

\begin{lem}\label{conscp:kappalemma}
Let $v$ be as in \eqref{model:not_varmean} and $\kappa_1^\pm$ and $\kappa_2^\pm$ be defined as in \eqref{conscp:kappa1} and
\eqref{conscp:kappa2}, respectively. Then, 
\begin{equation*}
\kappa_1^\pm(v,w,x,y) \geq \frac{x^2}{8}
 \frac{\inf_{v\leq t \leq w} v(t)^2}{\sup_{v\leq t \leq w} v(t)} - y \; \text{
 and }\; \kappa_2^\pm(v,w,x) \geq \frac{x^2}{2}\inf_{v\leq t \leq w} v(t).
\end{equation*}
\end{lem}

\begin{proof}
First observe from \eqref{proofs:kl}, that for any $\theta\in\Theta$ and
$\eps>0$ such that $\theta + \eps \in\Theta$ one has $D(\theta||\theta+\eps) = 
\int_\theta^{\theta+\eps}(\theta + \eps - t)v(t)\diff t$. Thus if follows that
for all $0\leq\eps\leq x$
\begin{align*}
\frac{\eps}{x} D(\theta||\theta + x)  - D(\theta||\theta +
\eps)  & = \frac{\eps}{x}\int_\theta^{\theta + x} (\theta +
x - t)v(t)\diff t - \int_\theta^{\theta + \eps}(\theta + \eps -t)v(t)\diff t 
\\ & \geq \frac{\eps x}{2} \inf_{t\in[\theta,\theta+x]} v(t) -
\frac{\eps^2}{2}\sup_{t\in[\theta,\theta+x]} v(t).
\end{align*}
Maximizing over $0\leq \eps \leq x$ then yields 
\begin{equation*}
 \sup_{\eps\in[0,x]}\frac{\eps}{x} D(\theta||\theta + x)  -
 D(\theta||\theta + \eps) \geq \frac{x^2}{8} \frac{\inf_{t\in[\theta,\theta+x]}
v(t)^2}{\sup_{t\in[\theta,\theta+x]} v(t)}. 
\end{equation*}
This proves that
\begin{equation*}
\kappa_1^+(v,w,x,y) \geq \frac{x^2}{8} \frac{\inf_{v\leq t
\leq w} v(t)^2}{\sup_{v\leq t \leq w} v(t)} - y.
\end{equation*}
Likewise, one finds
\begin{equation*}
 \kappa_2^+(v,w,x) \geq \frac{x^2}{2} \inf_{v\leq t \leq w} v(t).
\end{equation*}
The estimates for $\kappa_1^-$ and $\kappa_2^-$ are derived analogously.
\end{proof}

The combination of Theorem \ref{conscp:thmunderest} and the estimates in Lemma
\ref{conscp:kappalemma} yield the handy result in Theorem
\ref{conscp:corunderest}. For the case of Gaussian observations, the constants
$\kappa_i^\pm$ ($i=1,2$) can be computed explicitly and in particular
$\kappa_1$ is strictly larger than the
approximations obtained from Lemma \ref{conscp:kappalemma} by setting
$v(t)\equiv 1$. 

\begin{thm}\label{gauss:opt:mainthm}
Let $q\in \R$ and $\hat K(q)$ be defined as in \eqref{conscp:estnocp} and assume
that $\mathcal{F}$ is the family of Gaussian distributions with fixed variance
$1$. Then, 
\begin{align*}
 \Prob \left( \hat K (q) < K  \right) \leq & 2 K\left[\exp\left( - \frac{1}{8}
 \left( \frac{\Delta\sqrt{\lambda n}}{2\sqrt{2}} - 
 q-\sqrt{2\log\frac{2e}{\lambda}}
 \right)_+^2 \right) +  \exp\left(- \frac{\lambda n
 \Delta^2}{16} \right)\right]
\end{align*}
\end{thm}
 
\begin{proof}
The proof is similar to the proof of Lemma \ref{proofs:power:gauss}. From Lemma \ref{conscp:kappalemma} it follows that  \ha{$\kappa_2^\pm(v,w,x) = \frac{x^2}{2}$}
and one computes explicitly that $\kappa_1^\pm(v,w,x,y) =
\frac{1}{2}(\frac{x}{2} - \frac{y}{x})^2\geq \frac{1}{8}(x-\sqrt{2y})^2$ if
$x^2\geq 2y$. The assertion now follows from Theorem \ref{conscp:thmunderest}.
\end{proof}

We close this section with the proof of Theorem \ref{convset:mainthm} which is
very much in the same spirit than the proof of Theorem \ref{conscp:thmunderest}
above.

\begin{proof}[Proof of Theorem \ref{convset:mainthm}] 
Let again $\Delta$ be the smallest jump of the true
signal $\vt$ and recall that $\vt(t)\in[\underline\theta,\overline\theta]$ for all
$t\in[0,1]$. Moreover, define the $K$ disjoint intervals $I_i=(\tau_i-c_n,
\tau_i+c_n)\subset [0,1]$ and accordingly $I_i^-$, $I_i^+$, $\theta_i^-$,
$\theta_i^+$ and $\hat \vartheta_i$ as in the proof of Theorem \ref{conscp:thmunderest}. 

Now assume that $\hat K\in \N$ and that $\hat\vartheta\in\S_n[\hat K]$ is
an estimator of $\vartheta$ such that $T_n(Y,\hat\vt)\leq q$ and\kl{
\begin{equation*}
\max_{0\leq k \leq  K} \min_{0\leq l \leq \hat K} \abs{\hat \tau_l -
\tau_k}>c_n.
\end{equation*}
Put differently, there exists an index
$i\in\set{1,\ldots,K}$ such that $\abs{\hat \tau_l -
\tau_i}>c_n$ for all $0\leq l\leq \hat K$ or, in other words, $\hat
\vartheta$ contains no change-point in the interval $I_i$.} With the very same reasoning as in the proof of Theorem \ref{conscp:thmunderest} we find
that
\begin{align*}
&\Prob\left(\exists \hat K\in\N,\hat\vt\in \S_n[\hat K]:
T_n(Y,\hat\vt)\leq q\text{ and }\kl{ \max_{0\leq k \leq  K} \min_{0\leq l \leq
\hat K}} \abs{\hat \tau_l - \tau_k}>c_n \right) \\
\leq & \sum_{k=1}^K \Prob\left(\exists \hat\theta \in
 \Theta:T_{I_k^+}(Y,\hat\theta) \leq \frac{1}{2} \left( q+\sqrt{\log\frac{e}{c_n}} \right)^2 \; \text{and} \;
T_{I_k^-}(Y,\hat\theta) \leq \frac{1}{2} \left( q+\sqrt{\log\frac{e}{c_n}} \right)^2 \right).
\end{align*}
By replacing $\lambda\slash 2$ in the proof of Theorem \ref{conscp:thmunderest}
by $c_n$ and applying Lemma \ref{conscp:kappalemma} the assertion follows. 
\end{proof}

\subsection{Proof of Theorems \ref{gauss:opt:unknownbackground} and
\ref{gauss:opt:nojumpmissing}}

\begin{proof}[Proof of Theorem \ref{gauss:opt:unknownbackground}]
W.l.o.g. we shall assume that $\delta_n \geq 0$. The main idea of the proof is
as follows: Let $J_n=\argmax\left\{\abs{J}~:~ J\subset[0,1],\, J\cap I_n =
\emptyset\right\}$. In order to show that \eqref{gauss:opt:unknownbackgroundeqn}
holds, we construct a sequence $\theta_n^{*}\in\Theta$  such that
\begin{eqnarray}
 &\sup_{\theta \geq \theta_n^*}\Prob\left( T_{J_n}
 (Y,\theta) \leq 1/2 \left( \ha{q_n}+ \sqrt{2 \log{(e/\abs{J_n})}} \right)^2\right)\ra 0\;
 \text{ and } \label{proofs:gauss:opt:eq1} \\
 &\sup_{\theta \leq \theta_n^*}\Prob\left( T_{I_n} (Y,\theta)
 \leq  1/2 \left( \ha{q_n}+ \sqrt{2 \log (e/\abs{\ha{I_n})}} \right)^2 \right)\ra 0. \label{proofs:gauss:opt:eq2}
\end{eqnarray} 
 Note that the true signal $\vartheta_n$ takes the value
 $\theta_0+\delta_n$ on $I_n$ and $\theta_0$ on $J_n$ and it is not restrictive
 to assume that $\inf_{n\in\N}\abs{J_n} > 0$. We construct
 $\theta^{*}_n=\theta_0+\sqrt{\beta_n\slash n}$ for a sequence $(\beta_n)_{n\in\N}$ that satisfies 
 $\sqrt\beta_n\slash q_n\ra \infty$.
  
 We first consider \eqref{proofs:gauss:opt:eq1}. To this end
 observe that for all $t\in J_n$ we have $\abs{\theta_n^* -
 \vt_n(t)}\sqrt{\abs{J_n} n} = \sqrt{\beta_n \abs{J_n}}$. We further find that
 \begin{equation*}
 \Gamma_{J_n} := \sqrt{\beta_n \abs{J_n}} -  q_n-\sqrt{2\log(e/\abs{J_n})} =
 q_n \left( \frac{\sqrt{\beta_n}}{q_n} - 1 - \frac {\sqrt{2\log(e/\abs{J_n})}}{q_n}  \right) \ra \infty.
 \end{equation*}
 Thus, we can apply \eqref{proofs:gauss:handyest} and find for all $\theta \geq
 \theta_n^*$
\begin{equation*}
\Prob\left(T_{J_n} (Y,\theta) \leq 1/2 \left( \ha{q_n}+ \sqrt{2 \log(e/\abs{J_n})} \right)^2\right) \leq
\exp\left(-\frac{\Gamma_{J_n}^2}{8} \right) \ra 0.
\end{equation*}

 Now observe that for $t\in I_n$ we have $\abs{\theta_n^* -
 \vt_n(t)}\sqrt{\abs{I_n} n} =\delta_n \sqrt{\abs{I_n}n} - \sqrt{\beta_n
 \abs{I_n}}$. Thus \eqref{proofs:gauss:opt:eq2} follows from
 \eqref{proofs:gauss:handyest} given
 \begin{equation*}
 \Gamma_{I_n} := \delta_n \sqrt{\abs{I_n} n} - \sqrt{\beta_n \abs{I_n}} -
  q_n-\sqrt{2\log(e/\abs{\ha{I_n}})} \ra \infty.
 \end{equation*}
 It hence remains to construct sequences $(\beta_n)$ for each case $(1)$ and
 $(2)$ such that the previous condition holds while $\sqrt{\beta_n}\slash
 q_n\ra\infty$.
 
 We assume first that $\liminf_{n\ra\infty} \abs{I_n} > 0$ and define $\beta_n$
 through the equation $\sqrt{\beta_n \abs{I_n}} = c\left( \delta_n
 \sqrt{\abs{I_n} n} -  q_n-\sqrt{2\log(e/\abs{\ha{I_n}})} \right)$ for some $0< c < 1$. Then,
 \begin{equation*}
 \frac{\sqrt{\beta_n \abs{I_n}}}{q_n} = c\left(\frac{\delta_n \sqrt{\abs{I_n}
 n}}{{q_n}} -1 -\frac{\sqrt{2\log(e/\abs{\ha{I_n}})}}{q_n}\right)
 \end{equation*}
 From the condition in case $(1)$ of the theorem and the fact that $\abs{I_n}$
 is bounded away from zero for large $n$, we find that $\sqrt\beta_n\slash q_n
 \ra\infty$. Further we find $\Gamma_{I_n} = (1-c)\sqrt{\beta_n
 \abs{I_n}}\ra\infty$.
 
 Finally we consider the case when $\abs{I_n}\ra 0$ and define $\beta_n$
 through the equation $\sqrt{\beta_n \abs{I_n}} = c \eps_n \sqrt{-
 \log\abs{I_n}}$ for some $0<c<1$. From the conditions in case $(2)$ of the
 theorem and the inequality $\sqrt{x+1}-\sqrt{x} \leq 1/(2 \sqrt{x})$, which holds for any $x>0$, one obtains
 \begin{align*}
 \Gamma_{I_n} & \geq (\sqrt{2} + \eps_n)\sqrt{-\log\abs{I_n}} - \sqrt{\beta_n \abs{I_n}} -
  q_n-\sqrt{2\log(e/\abs{I_n})} \\
 & = (\sqrt{2}+(1-c)\eps_n) \sqrt{-\log\abs{I_n}} -q_n - \sqrt{2} \sqrt{1+\log(1/\abs{I_n})} \\
 & \geq ((1-c)\eps_n)\sqrt{-\log\abs{I_n}} - \frac{1}{\sqrt{-2 \log\abs{I_n}}} - q_n.
 \end{align*}
 This shows that $\Gamma_{I_n}\ra\infty$ for a suitable small $c$, such that $\sup_{n\in\N} q_n /
 (\epsilon_n\sqrt{\log(1\slash \abs{I_n})})\leq1-2 c$. Again from the assumptions in the
 theorem it follows that $\sqrt \beta_n\slash q_n \ra \infty$. 
 \end{proof}

\begin{proof}[Proof of Theorem \ref{gauss:opt:nojumpmissing}]
Theorem \ref{gauss:opt:mainthm} implies $\Prob(\hat
K(q_n) < K_n)\leq e^{-\Gamma_{1,n}} + e^{-\Gamma_{2,n}}$ with
\begin{align*}
\Gamma_{1,n} = \frac{1}{8}\left(\frac{\sqrt{n
\lambda_n}\Delta_n}{2\sqrt{2}} - q_n - \sqrt{2 \log(2e/ \ha{\lambda_n})} \right)_+^2 - \log K_n\;\text{
and }\; \Gamma_{2,n} = \frac{n \lambda_n \Delta^2_n}{16} - \log K_n.
\end{align*} 
It is easy to see, that any condition $(1)$ - $(3)$
implies $\Gamma_{2,n}\ra\infty$. It remains to check that
$\Gamma_{1,n}\ra\infty$. Under condition $(1)$ we observe that
\begin{equation*}
\frac{\Gamma_{1,n}}{q_n^2} = \frac{1}{8}\left(\frac{\sqrt{n
\lambda_n}\Delta_n}{2 \sqrt 2 q_n} - \frac{q_n + \sqrt{2 \log(2e/\lambda_n)}}{q_n}\right)_+^2 - \frac{\log K_n}{q_n^2}\ra \infty.
\end{equation*}
Since $q_n$ is bounded away from zero, the assertion follows. Next, we consider
conditions $(2)$ and $(3)$. To this end, assume that
$\sqrt{n\lambda_n}\Delta_n\geq (C + \eps_n)\sqrt{\log(1\slash
\lambda_n)}$ for some constant $C>0$ and a sequence $\eps_n$ such that $\eps_n
\sqrt{\log(1\slash \lambda_n)}\ra \infty$. We find that
\begin{align*}
\Gamma_{1,n} & \geq
\frac{1}{8}\left(\frac{(C+\eps_n)\sqrt{\log\frac{1}{\lambda_n}}}{2 \sqrt2} - q_n - \sqrt{2 \log(2e/\lambda_n)}  \right)_+^2 - \log K_n \\
& = \frac{1}{8}\left(\frac{\eps_n\sqrt{\log\frac{1}{\lambda_n}}}{2 \sqrt{2}} \ha{+} \left( \frac{C-4}{2 \sqrt 2}  \right)
\sqrt{\log\frac{1}{\lambda_n}} -q_n -
\frac{1+ \log 2}{\sqrt{2 \log(1/\lambda_n)}} \right)_+^2 - \log K_n,
\end{align*}
where we have used the inequality $\sqrt{x+1}-\sqrt{x} \leq 1/(2 \sqrt{x})$.
If $\sup_{n\in\N} K_n<\infty$, then the choice $C=4$ implies
$\Gamma_{1,n}\ra\infty$. Otherwise, we use the estimate $K_n\leq 1\slash
\lambda_n$ which results in $C = 12$ as a sufficient condition for
$\Gamma_{1,n}\ra\infty$.
\end{proof}

\subsection{Proof of Lemma \ref{impl:equivprob}}
 
\begin{proof}
First observe that the definition of $\hat \vt(q)$ in \eqref{intro:smre} implies
that $q \geq T_n(Y,\hat\vt(q))$ and hence, by identifying $\hat\vt(q)$ with the
pair $(\hat \Pcal(q), \hat \theta(q))$, we find
\begin{align*}
(\hat K(q)+1) q & \geq (\hat K(q)+1) T_n(Y,\hat\vt(q)) \geq \sum_{I \in \hat \Pcal(q)} \left( \sqrt{2 T_I(Y,\hat\vt(q))} - \sqrt{2 \log(e/ \abs{I})}\right) \\
&\geq \sqrt{2} \sqrt{\sum_{I \in \hat \Pcal(q)}\left( \abs{I} \phi(\bar{Y}_I) \right) - l(Y, \hat \vt(q)} - n \sqrt{2\log (e n)} \\
&\geq \sqrt{2} \sqrt{l(\bar Y, \hat\vt(q)) - l(Y, m^{-1}(\bar Y))} - n \sqrt{2\log (e n)} 
\end{align*}
The last inequality follows from the fact that $\phi(\overline Y_I)\geq
\overline Y_I \theta - \psi(\theta)$ for all $\theta\in\Theta$ and all $I\in
\hat \Pcal(q)$ for the choice $\theta = m^{-1}(\overline Y)$.
Summarizing, we find
\begin{equation*}
\gamma \geq \left( (\hat K (q) + 1) q + n \sqrt{2 \log(en)}\right)^2 / 2 + l(Y,m^{-1}(\bar Y)) \geq l(Y,\hat\vt(q)).
\end{equation*}
Now, let $\hat\vartheta = (\hat\Pcal, \hat\theta)$ be a
minimizer of \eqref{impl:penprob}. The definition of $\hat K(q)$ in
\eqref{conscp:estnocp}  implies that $D(\Pcal, \theta) = \infty$ if ${\# \Pcal}
< \hat K(q)$. Thus we have that  $\bigl| \hat\Pcal \bigr| \geq \hat
K(q)$. Assume that there exists $k\geq 1$ such that ${\# \Pcal} = \hat
K(q) + k$ (for $k=0$ nothing is to show). Since $(\hat\Pcal, \hat\theta)$ is a
minimizer of \eqref{impl:penprob} and since $D\geq 0$ we find
\begin{align*}
\gamma(\bigl|\hat\Pcal\bigr| - 1) & \leq
D(\hat\Pcal(q),\hat\theta(q)) - D(\hat\Pcal,\hat\theta) +
\gamma\bigl(\bigl|\hat\Pcal(q)\bigr| - 1\bigr) \\
& \leq D(\hat\Pcal(q),\hat\theta(q)) - k \gamma +
\gamma\bigl(\bigl|\hat\Pcal\bigr| - 1\bigr) \\
& < (1-k)l(Y,\hat\vartheta(q)) +
\gamma\bigl(\bigl|\hat\Pcal\bigr| - 1\bigr).
\end{align*}
This is a contradiction for  $l(\hat\vt)$ being
non-negative and hence we conclude that $\bigl|\hat\Pcal\bigr| = \hat
K(q)$ and that $\hat\vt = (\hat\Pcal, \hat\theta)$ solves \eqref{intro:smre}.
\end{proof}

\appendix

\bibliographysupp{literature}
\bibliographystylesupp{chicago}

\end{document}